\documentclass[a4paper,onecolumn,11pt,accepted=2025-03-20]{quantumarticle}
\pdfoutput=1
\usepackage[utf8]{inputenc}
\usepackage[english]{babel}
\usepackage[T1]{fontenc}
\usepackage{amsmath}

\usepackage[numbers,sort&compress]{natbib}
\usepackage{amsfonts,algorithmic}
\usepackage{xcolor}
\usepackage{qcircuit}
\usepackage{bm}
\usepackage{enumerate}
\usepackage{appendix}
\usepackage{amsthm}
\usepackage{multirow}
\usepackage{tikz}
\usetikzlibrary{matrix,decorations.pathreplacing}
\usepackage{mleftright}
\usepackage{amssymb}
\usepackage[linesnumbered,ruled,vlined]{algorithm2e}
\usepackage{mathabx,mathrsfs}
\usepackage{comment}
\usepackage{soul}
\usepackage{bbm}
\usepackage{authblk}
\usepackage[pagebackref]{hyperref}
\usepackage[capitalize,nameinlink,noabbrev]{cleveref}

\newcommand{\ket}[1]{|{#1}  \rangle}

\newcommand{\norm}[1]{{\|}#1{\|}}
\newcommand{\poly}{\operatorname{poly}}

\DeclareMathOperator{\sign}{sign}

\newcommand{\bfbeta}{\boldsymbol{\beta}}
\newcommand{\bfeta}{\boldsymbol{\eta}}
\newcommand{\bftheta}{\boldsymbol{\theta}}
\newcommand{\bfmu}{\boldsymbol{\mu}}
\newcommand{\bfDelta}{\boldsymbol{\Delta}}
\newcommand{\bfkappa}{\boldsymbol{\kappa}}
\newcommand{\bfhatbeta}{\hat {\boldsymbol{\beta}}}

\newtheorem{theorem}{Theorem}

\newtheorem{lemma}[theorem]{Lemma}
\newtheorem{remark}[theorem]{Remark}

\newtheorem{definition}[theorem]{Definition} 
\newtheorem{fact}[theorem]{Fact}
\newtheorem{assumption}[theorem]{Assumption}

\newtheorem{result}{Result}

\crefname{lemma}{Lemma}{lemma}
\crefname{fact}{Fact}{fact}
\crefname{definition}{Definition}{definition}
\crefname{corollary}{Corollary}{corollary}
\crefname{result}{Result}{result}

\usepackage{mathtools}

\usepackage{titlesec}
\DeclareMathOperator*{\argmin}{arg\,min}

\usepackage{geometry}
\def\be{\begin{eqnarray}}
\def\ee{\end{eqnarray}}

\definecolor{Pr}{rgb}{0.4,0.3,0.9}

\DeclareRobustCommand{\DE}[2]{#2}
\DeclareRobustCommand{\VANDE}[3]{#3}

\begin{document}

\title{Quantum Algorithms for the Pathwise Lasso}

\author[0,1]{Joao F. Doriguello}
\email{doriguello@renyi.hu}
\homepage{www.joaodoriguello.com}
\orcid{0000-0002-8265-7334}

\author[1,2]{Debbie Lim}
\email{limhueychih@gmail.com}

\author[3]{Chi Seng Pun}
\email{cspun@ntu.edu.sg}

\author[1,4]{Patrick Rebentrost}
\email{patrick@comp.nus.edu.sg}

\author[3]{Tushar Vaidya}
\email{tushar.vaidya@ntu.edu.sg}

\affil[0]{HUN-REN Alfréd Rényi Institute of Mathematics, Budapest, Hungary}
\affil[1]{Centre for Quantum Technologies, National University of Singapore, Singapore}
\affil[2]{Center for Quantum Computer Science, Faculty of Computing, University of Latvia, Latvia}
\affil[3]{School of Physical and Mathematical Sciences, Nanyang Technological University, Singapore}
\affil[4]{Department of Computer Science, National University of Singapore, Singapore}
\date{}

\maketitle

\begin{abstract}
We present a novel quantum high-dimensional linear regression algorithm with an $\ell_1$-penalty based on the classical LARS (Least Angle Regression) pathwise algorithm. Similarly to available classical numerical algorithms for Lasso, our quantum algorithm provides the full regularisation path as the penalty term varies, but quadratically faster per iteration under specific conditions. A quadratic speedup on the number of features/predictors $d$ is possible by using the simple quantum minimum-finding subroutine from D\"urr and H\o{}yer (arXiv’96) in order to obtain the joining time at each iteration. We then improve upon this simple quantum algorithm and obtain a quadratic speedup both in the number of features $d$ and the number of observations $n$ by using the approximate quantum minimum-finding subroutine from Chen and de Wolf (ICALP'23). In order to do so, we construct a quantum unitary based on quantum amplitude estimation to approximately compute the joining times to be searched over by the approximate quantum minimum-finding subroutine. Since the joining times are no longer exactly computed, it is no longer clear that the resulting approximate quantum algorithm obtains a good solution. As another main contribution, we prove, via an approximate version of the KKT conditions and a duality gap, that the LARS algorithm (and therefore our quantum algorithm) is robust to errors. This means that it still outputs a path that minimises the Lasso cost function up to a small error if the joining times are only approximately computed. Furthermore, we show that, when the observations are sampled from a Gaussian distribution, our quantum algorithm's complexity only depends polylogarithmically on $n$, exponentially better than the classical LARS algorithm, while keeping the quadratic improvement on $d$. Moreover, we propose a dequantised version of our quantum algorithm that also retains the polylogarithmic dependence on $n$, albeit presenting the linear scaling on $d$ from the standard LARS algorithm. Finally, we prove query lower bounds for classical and quantum Lasso algorithms.
\end{abstract}


\newpage
\tableofcontents

\section{Introduction}
One of the most important research topics that has attracted renewed attention by statisticians and the machine learning community is high-dimensional data analysis~\cite{bayati2013estimating, buhlmann2011statistics, garrigues2008homotopy, johnstone2009statistical, kelner2022lower, wright2022high, zhao2006model}. For linear regression, this means that the number of data points is less than the number of explanatory variables (features). In machine learning parlance, this is usually referred to as overparameterisation. Solutions to such problems are, however, ill defined. Having more variables to choose from is a double-edged sword. While it gives us freedom of choice, computational cost considerations favour choosing sparse models. Some sort of regularisation in the linear model, usually in the form of a penalty term, is needed to favour sparse models.

Let us consider the regression problem defined as follows. Assume to be given a fixed design matrix $\mathbf{X}\in\mathbb{R}^{n\times d}$  and a vector of observations $\mathbf{y}\in\mathbb{R}^n$. The setting is such that $d\geq n$, i.e., more features than observations are given. Typical examples are $(n,d) = (500,12288), (7500,49152),(18737,65536)$ for image reconstructions~\cite{candes2008enhancing,bora2017compressed,asim2020invertible,song2020improved}. In portfolio optimization, an instance would be $(n,d) = (504,600), (252,486)$ \cite{chiu2017big,pun2019linear}.
The un-regularized $\ell_2$-regression problem is defined as $\min_{\bfbeta \in \mathbb{R}^d}\frac{1}{2}\|\mathbf{y}-\mathbf{X}{\bfbeta}\|_2^2$.
Generally, different norms could be used in the penalty term.  Essentially, the three popular sparse regression models reduce to three prototypes recalled below with $\ell_0, \ell_1, \ell_2$-norms and tuning parameter $\lambda$:
\begin{align*}
 &\underset{\bfbeta \in \mathbb{R}^d} {\min}\frac{1}{2}\|\mathbf{y}-\mathbf{X}{\bfbeta}\|_2^2 + \lambda \|\bfbeta\|_0 \quad \text{for all}~ \lambda >0 ~~\text{(best subset selection)}, \\
 &\underset{\bfbeta \in \mathbb{R}^d} {\min}\frac{1}{2}\|\mathbf{y}-\mathbf{X}{\bfbeta}\|_2^2 + \lambda \|\bfbeta\|_1 \quad \text{for all}~ \lambda >0 ~~\text{(Lasso regression)}, \\
 &\underset{\bfbeta \in \mathbb{R}^d} {\min}\frac{1}{2}\|\mathbf{y}-\mathbf{X}{\bfbeta}\|_2^2 + \lambda \|\bfbeta\|_2 \quad \text{for all}~ \lambda >0 ~~\text{(Ridge regression)}.
\end{align*}
The easiest is Ridge regression which has an analytical solution and the problem is convex, but tends to give non-sparse solutions.
The best subset selection is not convex and is NP-hard~\cite{das2008algorithms} in general. 
Lasso regression with $\ell_1$-penalty, on the other hand, is a convex relaxation of the best subset selection problem and will be our central object of study throughout this paper. Lasso stands for Least Absolute Shrinkage and Selection operator and has been studied extensively in the literature because of its convexity and parsimonious solutions, leading to its use across disciplines and in many different fields where sparsity is a desired prerequisite: compressed sensing~\cite{chen2001atomic, candes2006stable,figueiredo2007gradient}, genetics~\cite{usai2009Lasso, wu2009genome}, and finance~\cite{tian2015variable,coad2020catching,peng2022portfolio}.
For large language models that use deep learning, overparameterisation is actually desired. In high-dimensional linear regression models, in contrast, parsimonious models are desired. The conventional statistical paradigm is to keep the number of features fixed and study the asymptotics as the sample size increases. The other paradigm is predictor selection given a large choice that fits into high-dimensional regression. Selecting which of the many variables available is part of model choice and techniques that help us in choosing the right model have significant value. Having a larger number of variables in statistical models makes it simpler to pick the right model and achieve improved predictive performance. Overparametrisation is a \emph{virtue} and increases the likelihood of including the essential features in the model. Yet this should be balanced with regularisation~\cite{efron2004least, fan2010selective, lockhart2014significance}. 
Introducing an $\ell_1$-penalty confers a number of advantages~\cite{chen2001atomic,Hastie2015, tibshirani1996regression}:
\begin{enumerate}
    \item The $\ell_1$-penalty provides models that can be interpreted in a simple way and thus have advantages compared to black boxes (deep neural networks) in terms of transparency and intelligibility~\cite{murdoch2019definitions, rudin2022interpretable};
    \item If the original model generating the data is sparse, then an $\ell_1$-penalty provides the right framework to recover the original signal;
    \item Lasso selects the true model consistently even with noisy data~\cite{zhao2006model}, and this remains an active area of research~\cite{meng2023model};
    \item As $\ell_1$-penalties are convex, sparsity leads to computational advantages. For example, if we have two million features but only 200 observations, estimating two million parameters becomes extremely difficult.
\end{enumerate}
Generally, the optimal tuning parameter $\lambda$ has to be determined from the data and as such we would like an endogenous model that determines the optimal $\bfbeta$ along with a specific $\lambda$. The key point is that we require a path of solutions as $\lambda$ varies~\cite{efron2004least, friedman2010regularization,zhao2018pathwise}. 

On the other hand, quantum computing is a new paradigm that offers faster computation for certain problems, for example, integer factorisation and discrete logarithm~\cite{shor1994algorithms,shor1999polynomial}. Quantum algorithms have been proposed for unregularised linear regression~\cite{wiebe2012quantum,wang2017quantum,chakraborty2019power,kaneko2021linear}, Ridge regression~\cite{yu2019improved,shao2020quantum,shao2023improved,chen2021quantum}, and simple versions of Lasso regression~\cite{chen2021quantum}. In this paper, we propose new quantum algorithms (and also dequantised versions) for the \emph{pathwise} Lasso regression based on the famous Least Angle Regression (LARS) algorithm.

\subsection{Least Angle Regression algorithm}

The pathwise Lasso regression problem is the problem when the regression parameter $\lambda$ varies. To be more precise, consider the Lasso estimator described above,
\begin{equation}\label{eq:Lasso}
    \hat{\bfbeta} \in \argmin_{\bfbeta \in \mathbb{R}^d} \mathcal{L}^{(\lambda)}(\bfbeta) \quad\text{where}\quad \mathcal{L}^{(\lambda)}(\bfbeta) \triangleq \frac{1}{2}\|\mathbf{y}-\mathbf{X}{\bfbeta}\|_2^2 + \lambda \|\bfbeta\|_1 \quad \text{for all}~ \lambda >0.
\end{equation}
The $\ell_1$-penalty term is $\| {\bfbeta} \|_1 = \sum_{i=1}^d | {\beta}_i |$, whereas the predictive loss is the standard Euclidean squared norm. In this work, we are interested in the Lasso solution $\hat{\bfbeta}(\lambda)$ as a function of the regularisation parameter $\lambda > 0$. For such we define the optimal regularisation path $\mathcal{P}$:
\begin{align*}
    \mathcal{P} \triangleq \{\hat{\bfbeta}(\lambda) : \lambda >0\}.
\end{align*}
Efron \textit{et al.}~\cite{efron2004least} showed that the optimal regularisation path $\mathcal{P}$ is \emph{piecewise linear} and \emph{continuous} with respect to $\lambda$. This means that there exist an $m \in \mathbb{N}$ and $\infty > \lambda_0 > \cdots > \lambda_{m-1} > \lambda_m = 0$ and $\bftheta_0,\dots,\bftheta_{m}\in\mathbb{R}^d$ such that\footnote{Define $[n]\triangleq \{1,\dots,n\}$.}
\begin{equation}\label{eqPath}
    \hat{\bfbeta}(\lambda) = \hat{\bfbeta}(\lambda_t) + (\lambda_t - \lambda)\bftheta_t \quad\text{for}~\lambda_{t+1} < \lambda \leq \lambda_{t} ~~(t\in\{0,\dots,m-1\}).
\end{equation}
There is a maximal value $\lambda_{\max} = \lambda_0$ where $\hat{\bfbeta}(\lambda) = \mathbf{0}$ for all $\lambda \geq \lambda_0$ and $\hat{\beta}_j(\lambda) \neq 0$ for some $j\in[d]$ and $\lambda < \lambda_0$. Such value is $\lambda_0 = \|\mathbf{X}^\top\mathbf{y}\|_\infty$. The $m+1$ points $\lambda_0,\dots,\lambda_m$ where $\partial \hat{\bfbeta}(\lambda)/\partial\lambda$ changes are called \emph{kinks}, and the path $\{\hat{\bfbeta}(\lambda): \lambda_{t+1} < \lambda \leq \lambda_t\}$ between two consecutive kinks $\lambda_{t+1}$ and $\lambda_t$ defines a linear segment. The fact that the regularisation path is piecewise linear and the locations of the kinks can be derived from the Karush-Kuhn-Tucker (KKT) optimality conditions (see \cref{sec:Lasso_path} for more details).

Since the regularisation path is piecewise linear, one needs only to compute all kinks $(\lambda_t, \hat{\bfbeta}(\lambda_t))$ for $t\in\{0,\dots,m-1\}$, from which the whole regularisation path follows by linear interpolation. That is exactly what the Least Angle Regression (LARS) algorithm proposed and named by Efron \textit{et al.}~\cite{efron2004least} does (a very similar idea appeared earlier in the works of Osborne \emph{et al.}~\cite{osborne2000new,osborne2000Lasso}). Starting at $\lambda = \infty$ (or $\lambda_0$) where the Lasso solution is $\hat{\bfbeta}(\lambda) = \mathbf{0}\in\mathbb{R}^d$, the LARS algorithm decreases the regularisation parameter $\lambda$ and computes the regularisation path by finding the kinks along the way with the aid of the KKT conditions. Each kink corresponds to an iteration of the algorithm. More specifically, at each iteration $t$, the LARS algorithm computes the next kink $\lambda_{t+1}$ by finding the closest point (to the previous kink) where the KKT conditions break and thus need to be updated. This is done by performing a search over the \emph{active set} $\mathcal{A} \triangleq \{i\in[d]:\hat{\beta}_i \neq 0\}$ and another search over the \emph{inactive set} $\mathcal{I} \triangleq [d]\setminus\mathcal{A}$. The search over $\mathcal{A}$ finds the next point $\lambda_{t+1}^{\rm cross}$, called \emph{crossing time}, where a variable $i_{t+1}^{\rm cross}$ must leave $\mathcal{A}$ and join $\mathcal{I}$. The search over $\mathcal{I}$ finds the next point $\lambda_{t+1}^{\rm join}$, called \emph{joining time}, where a variable $i_{t+1}^{\rm join}$ must leave $\mathcal{I}$ and join $\mathcal{A}$. The next kink is thus $\lambda_{t+1} = \max\{\lambda_{t+1}^{\rm cross}, \lambda_{t+1}^{\rm join}\}$. The overall complexity of the LARS algorithm per iteration is $O(nd + |\mathcal{A}|^2)$: the two searches over the active and inactive sets require $O(n|\mathcal{A}|)$ and $O(n|\mathcal{I}|)$ time, respectively, while computing the new direction $\partial\hat{\bfbeta}/\partial\lambda = \bftheta_{t+1}$ of the regularisation path, which involves the computation of the pseudo-inverse of a submatrix of $\mathbf{X}$ specified by $\mathcal{A}$, requires $O(n|\mathcal{A}| + |\mathcal{A}|^2)$ time. When the Lasso solution is unique, it is known~\cite{tibshirani2013Lasso} that $|\mathcal{A}| \leq \min\{n,d\} = n$ throughout the LARS algorithm.

\subsection{Our work}
In the high-dimensional setting where $d \gg n$, the search over the inactive set $\mathcal{I}$, and thus the computation of the joining time, is by far the most costly step per iteration. 
In this work, we propose quantum algorithms for the pathwise Lasso regression problem based on the Least Angle Regression (LARS) algorithm.
To the best of our knowledge, this work is the first to present a quantum version of the LARS algorithm.
We propose mainly two quantum algorithms based on the LARS algorithm to speedup the computation of the joining time. Our first quantum algorithm, called \emph{simple quantum} LARS \emph{algorithm}, is a straightforward improvement that utilizes the well-known quantum minimum-finding subroutine from D\"urr and H\o{}yer~\cite{durr1996quantum} to perform the search over $\mathcal{I}$. We assume that the design matrix $\mathbf{X}$ is stored in a quantum-readable read-only memory (QROM) and can be accessed in time $O(\poly\log(nd))$. We also assume that data can be written into quantum-readable classical-writable classical memories (QRAM) of size $O(n)$, which can be later queried in time $O(\poly\log{n})$~\cite{giovannetti2008architectures,giovannetti2008quantum}. The complexity of computing the joining time is now $O(n\sqrt{\smash[b]{|\mathcal{I}|}})$, where $O(n)$ is the time required to compute, via classical circuits, the joining times to be maximised over $\mathcal{I}$. The final runtime of our simple quantum algorithm is $\widetilde{O}(n\sqrt{\smash[b]{|\mathcal{I}|}} + n|\mathcal{A}| + |\mathcal{A}|^2)$ per iteration, where we omit polylog factors in $n$ and $d$.

We then improve upon our simple quantum algorithm by approximately computing the joining times to be maximised over $\mathcal{I}$, thus reducing the factor $O(n)$ to $O(\sqrt{n})$. This is done by constructing a quantum unitary based on quantum amplitude estimation~\cite{brassard2002quantum} that approximately computes the joining times of all the variables in the inactive set $\mathcal{I}$. We then employ the approximate quantum minimum-finding subroutine from Chen and de Wolf~\cite{chen2021quantum} to find an estimate $\epsilon$-close to the true maximum joining time. The result is our \emph{approximate quantum} LARS \emph{algorithm} that achieves a quadratic speedup in both the number of features $d$ and the number of observations $n$. We also propose a sampling-based classical LARS algorithm, which can be seen as a quantisation counterpart. Our improved quantum algorithm (and the sampling-based classical algorithm), however, introduces an error in computing the joining time $\lambda_{t+1}^{\rm join}$. As one of our main contributions, we prove that the LARS algorithm is robust to errors. More specifically, we show that, by slightly adapting the LARS algorithm, it returns (and so our approximate quantum algorithm) an \emph{approximate regularisation path} with error proportional to the error from computing $\lambda_{t+1}^{\rm join}$. We first define an approximate regularisation path. In the following, let $\mathbf{X}^+$ be the Moore–Penrose inverse of $\mathbf{X}$ and given a matrix $\mathbf{A}$, let $\|\mathbf{A}\|_2$ be its the spectral norm, $\|\mathbf{A}\|_1 = \max_{j}\sum_{i} |A_{ij}|$, and $\|\mathbf{A}\|_{\max} = \max_{i,j}|A_{ij}|$.
\begin{definition}\label{def:approximate_Lasso_solution}
    Given $\epsilon : \mathbb{R}_{>0} \to \mathbb{R}_{\geq 0}$, we say that a vector $\widetilde{\bfbeta}\in\mathbb{R}^d$ is an approximate Lasso solution with error $\epsilon(\lambda)$, or an $\epsilon(\lambda)$-minimiser, if
    \begin{align*}
        \mathcal{L}^{(\lambda)}(\widetilde{\bfbeta}) - \min_{\bfbeta\in\mathbb{R}^d}\mathcal{L}^{(\lambda)}(\bfbeta) \leq \epsilon(\lambda), \quad\text{where}\quad \lambda>0.
    \end{align*}
    A set $\widetilde{\mathcal{P}} \triangleq \{\widetilde{\bfbeta}(\lambda)\in\mathbb{R}^d: \lambda >0 \}$ is an approximate regularisation path with error $\epsilon(\lambda)$ if, for all $\lambda>0$, the point $\widetilde{\bfbeta}(\lambda)$ of $\widetilde{\mathcal{P}}$ is an approximate Lasso solution with error $\epsilon(\lambda)$.
\end{definition}
\begin{result}[Informal version of \cref{thr:correctness}]\label{res:result1}
    Let $T\in\mathbb{N}$, $\epsilon \geq 0$, and  $\mathbf{X}\in\mathbb{R}^{n\times d}$. For $\mathcal{A}\subseteq[d]$ of size $|\mathcal{A}|\leq T$, let $\alpha_{\mathcal{A}} > 0$ be such that $\|\mathbf{X}_{\mathcal{A}}^+\mathbf{X}_{\mathcal{A}^c}\|_1 \leq \alpha_{\mathcal{A}}$. Consider an approximate {\rm LARS} algorithm that returns a path $\widetilde{\mathcal{P}} = \{\widetilde{\bfbeta}(\lambda)\in\mathbb{R}^d:\lambda > 0\}$ with at most $T$ kinks and wherein, at each iteration $t$, the joining time $\lambda_{t+1}^{\rm join}$ is approximated by $\widetilde{\lambda}_{t+1}^{\rm join}$ such that $\lambda_{t+1}^{\rm join} \leq \widetilde{\lambda}_{t+1}^{\rm join} \leq \big(1-\frac{\epsilon}{1+\alpha_{\mathcal{A}}}\big)^{-1}\lambda_{t+1}^{\rm join}$. Then $\widetilde{\mathcal{P}}$ is an approximate regularisation path with error $\lambda\epsilon\|\widetilde{\bfbeta}(\lambda)\|_1$.
\end{result}
To prove the above result, we consider an approximate version of the KKT conditions and use a duality gap for the Lasso regression. As far as we are aware, this is the first direct result on the robustness of the LARS algorithm. Compared to the work of Mairal and Yu~\cite{mairal2012complexity}, who also introduced an approximate LARS algorithm and employed similar techniques to analyse its correctness, the computation of the joining time in their case is exact and errors only arise by utilising a first-order optimisation method to find an approximate solution when kinks happen to be too close.

\cref{res:result1} guarantees the correctness of all our algorithms. Their time complexities, on the other hand, are given by the theorem below and summarised in \cref{table1}. 
\begin{result}[Informal version of \cref{thr:simple_quantum,thr:approximate_quantum,thr:approximate_classical}]\label{thr:summary_results}
    Let $\mathbf{X}\in\mathbb{R}^{n\times d}$, $\mathbf{y}\in\mathbb{R}^n$, $\delta\in(0,1)$, $\epsilon>0$, and $T\in\mathbb{N}$. Assume that $\mathbf{X}$ is stored in a {\rm QROM} and we have access to {\rm QRAMs} and classical-samplable structures of memory size $O(n)$. For $\mathcal{A}\subseteq[d]$ of size $|\mathcal{A}| \leq T$, assume there are $\alpha_{\mathcal{A}},\gamma_{\mathcal{A}}\in(0,1)$ such that
    \begin{align*}
        \|\mathbf{X}_{\mathcal{A}}^+\mathbf{X}_{\mathcal{A}^c}\|_1 \leq \alpha_{\mathcal{A}} \qquad\text{and}\qquad
        \frac{\|\mathbf{X}^\top(\mathbf{I} - \mathbf{X}_{\mathcal{A}}\mathbf{X}_{\mathcal{A}}^+) \mathbf{y}\|_\infty}{\|\mathbf{X}\|_{\max}\|(\mathbf{I} - \mathbf{X}_{\mathcal{A}}\mathbf{X}_{\mathcal{A}}^+) \mathbf{y}\|_1} \geq \gamma_{\mathcal{A}}, 
    \end{align*}
    where $\mathcal{A}^c = [d]\setminus\mathcal{A}$ and $\mathbf{X}_{\mathcal{A}} \in \mathbb{R}^{n\times |\mathcal{A}|}$ is the matrix formed by the columns of $\mathbf{X}$ in $\mathcal{A}$.

    \begin{itemize}
        \item There is a quantum {\rm LARS} algorithm that returns an optimal regularisation path with $T$ kinks with probability at least $1-\delta$ and in time
        \begin{align*}
            \widetilde{O}\big(n\sqrt{|\mathcal{I}|} + n|\mathcal{A}| + |\mathcal{A}|^2\big)
        \end{align*}
        per iteration, where $\mathcal{A}$ and $\mathcal{I}$ are the active and inactive sets of the corresponding iteration.

        \item There is a quantum {\rm LARS} algorithm that returns an approximate regularisation path with error $\lambda\epsilon\|\widetilde{\bfbeta}\|_1$ and $T$ kinks with probability at least $1-\delta$ and in time
        \begin{align*}
            \widetilde{O}\left(\frac{\gamma_{\mathcal{A}}^{-1} + \sqrt{n}\|\mathbf{X}\|_{\max}\|\mathbf{X}^+_{\mathcal{A}}\|_2}{(1-\alpha_{\mathcal{A}})\epsilon}\sqrt{|\mathcal{I}|} + n|\mathcal{A}| + |\mathcal{A}|^2 \right)
        \end{align*}
        per iteration, where $\mathcal{A}$ and $\mathcal{I}$ are the active and inactive sets of the corresponding iteration.

        \item There is a classical (dequantised) {\rm LARS} algorithm that returns an approximate regularisation path with error $\lambda\epsilon\|\widetilde{\bfbeta}\|_1$ and $T$ kinks with probability at least $1-\delta$ and in time
        \begin{align*}
            \widetilde{O}\left(\frac{\gamma_{\mathcal{A}}^{-2} + n\|\mathbf{X}\|_{\max}^2\|\mathbf{X}^+_{\mathcal{A}}\|_2^2}{(1-\alpha_{\mathcal{A}})^2\epsilon^2}|\mathcal{I}| + n|\mathcal{A}| + |\mathcal{A}|^2 \right)
        \end{align*}
        per iteration, where $\mathcal{A}$ and $\mathcal{I}$ are the active and inactive sets of the corresponding iteration.
    \end{itemize}
    The notation $\widetilde{O}(\cdot)$ omits $\poly\log$ terms in $n$, $d$, $T$, and $\delta$.
\end{result}

\begin{table}[t]
\centering
\def\arraystretch{1.25}
\begin{tabular}{ |p{5.2cm}|c|c|c|c| }
\hline \textbf{Algorithm} & \textbf{Error} & \textbf{Time complexity per iteration} \\
\hline
Classical LARS~\ref{alg:classicalLasso} & $0$ & $O(n|\mathcal{I}| + n|\mathcal{A}| + |\mathcal{A}|^2)$\\
\hline 
Simple quantum LARS~\ref{alg:warmup_classical_quantum} & $0$ & $\widetilde{O}(n\sqrt{|\mathcal{I}|} + n|\mathcal{A}| + |\mathcal{A}|^2)$\\
\hline 
Approximate classical LARS~\ref{alg:approximate_classical} & $\lambda\epsilon\|\widetilde{\bfbeta}\|_1$ & $\widetilde{O}\bigg(\frac{\gamma_{\mathcal{A}}^{-2} + n\|\mathbf{X}\|_{\max}^2\|\mathbf{X}^+_{\mathcal{A}}\|_2^2}{(1-\alpha_{\mathcal{A}})^2\epsilon^2}|\mathcal{I}| + n|\mathcal{A}| + |\mathcal{A}|^2 \bigg)$\\
\hline
Approximate quantum LARS~\ref{alg:classical_quantum} & $\lambda\epsilon\|\widetilde{\bfbeta}\|_1$ & $\widetilde{O}\bigg(\frac{\gamma_{\mathcal{A}}^{-1} + \sqrt{n}\|\mathbf{X}\|_{\max}\|\mathbf{X}^+_{\mathcal{A}}\|_2}{(1-\alpha_{\mathcal{A}})\epsilon}\sqrt{|\mathcal{I}|} + n|\mathcal{A}| + |\mathcal{A}|^2 \bigg)$\\
\hline 
\end{tabular}\caption{Summary of results. Throughout this work, $n$ is the number of observations, $d$ is the number of features, $\mathcal A$ is the active set, and $\mathcal I$ is the inactive set. In addition, $\epsilon>0$ is the approximation error, $\delta\in(0,1)$ is the failure probability, $T$ is the number of kinks in the regularisation path $\mathcal{P}$, $\mathbf{X}_{\mathcal{A}} = [\mathbf{X}_i]_{i\in\mathcal{A}}\in\mathbb{R}^{n\times |\mathcal{A}|}$ is the matrix formed by the columns of $\mathbf{X}$ in $\mathcal{A}$, $\mathbf{X}^+_{\mathcal{A}}$ is the Moore–Penrose inverse of $\mathbf{X}_{\mathcal{A}}$, and the parameters $\alpha_{\mathcal{A}},\gamma_{\mathcal{A}}\in(0,1)$ are defined in \cref{thr:summary_results}. The notation $\widetilde{O}(\cdot)$ omits $\poly\log$ terms in $n$, $d$, $T$, and $\delta$.}
\label{table1}
\end{table}

Our approximate LARS algorithms output a regularisation path with point-wise error $\lambda\epsilon\|\widetilde{\bfbeta}\|_1$. Even though this is not a purely additive error $\epsilon$, it is still better than a purely multiplicative error $\epsilon\mathcal{L}^{(\lambda)}(\widetilde{\bfbeta})$. It is possible to prove that $\mathcal{L}^{(\lambda)}(\widetilde{\bfbeta}) \leq \frac{1}{1-\epsilon}\min_{\bfbeta\in\mathbb{R}^d}\mathcal{L}^{(\lambda)}(\bfbeta)$ (see \Cref{remark1}), therefore our approximate LARS algorithms output a solution with error at most $\frac{\epsilon}{1-\epsilon}\min_{\bfbeta\in\mathbb{R}^d}\mathcal{L}^{(\lambda)}(\bfbeta)$. Moreover, this type of error $\lambda\epsilon\|\widetilde{\bfbeta}\|_1$ is standard in the literature, see \Cref{sect:bounds_noisy} and also~\cite{buhlmann2011statistics,mairal2012complexity,wainwright2019high}.

The complexity of our approximate quantum LARS algorithm depends on a few properties of the design matrix $\mathbf{X}$: the quantity $\|\mathbf{X}\|_{\max}\|\mathbf{X}^+_{\mathcal{A}}\|_2$ and the parameters $\alpha_{\mathcal{A}}$ and $\gamma_{\mathcal{A}}$. We note that $\|\mathbf{X}\|_{\max}\|\mathbf{X}^+_{\mathcal{A}}\|_2 \leq \|\mathbf{X}\|_2\|\mathbf{X}^+\|_2$ is upper bounded by the condition number of $\mathbf{X}$, which is simply a constant for well-behaved matrices. The existence of the parameter $\alpha_{\mathcal{A}}\in[0,1)$ is often called \emph{(mutual) incoherence} or \emph{strong irrepresentable condition} in the literature~\cite{zhao2006model,huang2008adaptive,wainwright2019high} and is related to the level of orthogonality between columns of $\mathbf{X}$: if the column space of $\mathbf{X}_{\mathcal{A}}$ is orthogonal to $\mathbf{X}_i$, then $\|\mathbf{X}_{\mathcal{A}}^+\mathbf{X}_i\|_1 = 0$. 
Mutual incoherence has been considered by several previous works~\cite{donoho2001uncertainty,elad2002generalized,fuchs2004recovery,tropp2006just,zhao2006model,meinshausen2006high,wainwright2009sharp}. Refs.~\cite{zhao2006model,meinshausen2006high} provide several examples of families of matrices that satisfy mutual incoherence. Finally, the parameter $\gamma_{\mathcal{A}}$, introduced and named by us as \emph{mutual overlap (between $\mathbf{y}$ and $\mathbf{X}$)}, measures the overlap between projections of the observation vector $\mathbf{y}$ and the columns of $\mathbf{X}$ (note that $\gamma_{\mathcal{A}} \leq 1$ by H\"older's inequality).

In the case when $\|\mathbf{X}\|_{\max}\|\mathbf{X}^+_{\mathcal{A}}\|_2 = O(1)$ and $\alpha_{\mathcal{A}}$ and $\gamma_{\mathcal{A}}$ are bounded away from $1$ and $0$, respectively, the complexity per iteration is $\widetilde{O}(\sqrt{\smash[b]{n|\mathcal{I}|}}/\epsilon + n|\mathcal{A}| + |\mathcal{A}|^2)$. For $\epsilon = 1/\poly\log{d}$ and iterations when $|\mathcal{A}| = O(n)$, we obtain the overall quadratic improvement $\widetilde{O}(\sqrt{nd})$ over the classical LARS algorithm. A further speedup beyond quadratic can be obtained if $\|\mathbf{X}\|_{\max}\|\mathbf{X}^+_{\mathcal{A}}\|_2 = o(1)$. Our approximate quantum LARS algorithm, however, depends on the design matrix being well behaved in order to bound the three quantities mentioned above. In \cref{sec:examples}, we bound $\|\mathbf{X}\|_{\max}\|\mathbf{X}^+_{\mathcal{A}}\|_2$, $\alpha_{\mathcal{A}}$, and $\gamma_{\mathcal{A}}$ for the case when $\mathbf{X}$ is a random matrix with rows sampled i.i.d.\ from the multivariate Gaussian distribution $\mathcal{N}(\mathbf{0},\mathbf{\Sigma})$ with covariance matrix $\mathbf{\Sigma}\succ \mathbf{0}\in\mathbb{R}^{d\times d}$. 
\begin{result}[Informal version of \Cref{lem:condition_number,lem:bound_mutual_incoherence,lem:bound_mutual_overlap}] \label{res:result_gaussian}
    Let $\mathbf{X}\in\mathbb{R}^{n\times d}$ be a random matrix with rows sampled i.i.d.\ from the Gaussian distribution $\mathcal{N}(\mathbf{0},\mathbf{\Sigma})$ with positive-definite covariance matrix $\mathbf{\Sigma}\in\mathbb{R}^{d\times d}$. Let $\mathcal{A}\subseteq[d]$ of size $|\mathcal{A}| = O(n/\log{d})$. If $\mathbf{\Sigma}$ is ``well behaved'' 
    (which includes $\mathbf{\Sigma} = \mathbf{I}$), then, with high probability,
    \begin{align*}
        \|\mathbf{X}\|_{\max}\|\mathbf{X}^+_{\mathcal{A}}\|_2 = O\bigg(\sqrt{\frac{\log{d}}{n}} \bigg), \qquad  \alpha_{\mathcal{A}} = 1 - \Omega(1), \qquad\text{and}\qquad \gamma_{\mathcal{A}} = \Omega\bigg(\frac{\|\mathbf{y}\|_2/\|\mathbf{y}\|_1}{\sqrt{\log{d}}} \bigg).
    \end{align*}
    As a consequence, the approximate quantum and classical {\rm LARS} algorithms output an approximate regularisation path with error $\lambda\epsilon\|\widetilde{\bfbeta}\|_1$ and $T = O(n/\log{d})$ kinks with probability at least $1-\delta$ in time $\widetilde{O}((\|\mathbf{y}\|_1/\|\mathbf{y}\|_2)(\sqrt{d}/\epsilon))$ and $\widetilde{O}((\|\mathbf{y}\|_1/\|\mathbf{y}\|_2)(d/\epsilon^2))$ per iteration, respectively, up to $\poly\log$ term in $n$, $d$, and $\delta$. If further $\|\mathbf{y}\|_1/\|\mathbf{y}\|_2 = \poly\log{n}$, then the quantum and classical complexities per iteration are $\widetilde{O}(\sqrt{d}/\epsilon)$ and $\widetilde{O}(d/\epsilon^2)$, respectively.
\end{result}

For the exact meaning of a ``well-behaved'' covariance matrix $\mathbf{\Sigma}$, we refer the reader to \Cref{sec:examples}. On a high level, it means that quantities like the minimum singular value of $\mathbf{\Sigma}$, the norm $\|\sqrt{\mathbf{\Sigma}}\|_1$, and other related properties of the conditional covariance matrix of $\mathbf{X}_{\mathcal{A}}|\mathbf{X}_{\mathcal{A}^c}$ are constant. Mostly special, the case $\mathbf{\Sigma} = \mathbf{I}$, which corresponds to a standard Gaussian matrix with entries sampled i.i.d.\ from $\mathcal{N}(0,1)$, trivially fulfills all these requirements.

\Cref{res:result_gaussian} implies that the complexity of our approximate LARS algorithms, both quantum and classical, can be substantially better than the standard LARS algorithm depending on the input data. In the case when $\mathbf{X}$ is a Gaussian matrix and $\mathbf{y}$ is sparse, so that $\|\mathbf{y}\|_1/\|\mathbf{y}\|_2 = \poly\log{n}$, the approximate quantum and classical complexities per iteration are $\widetilde{O}(\sqrt{d}/\epsilon)$ and $\widetilde{O}(d/\epsilon^2)$, respectively, which only depend polylogarithmically on the number of samples $n$ (assuming that $n|\mathcal{A}| = O(\sqrt{d})$). The quantum runtime, in particular, is more than quadratically better than the usual $O(nd)$ classical runtime (at the expense of approximating the solution). The condition that $\mathbf{y}$ is sparse can be artificially generated through setting to zero some of the regression coefficients of the underlying true solution $\bfbeta^\ast$ connecting $\mathbf{X}$ and $\mathbf{y}$ (see \Cref{sect:bounds_noisy}). We note that sparse inputs have also provided speedups in other settings~\cite{berry2007efficient,harrow2009quantum,berry2014exponential,prakash2014quantum,lloyd2016quantum,gilyen2019quantum,lim2023quantum}. Independently, we believe that this result would be of interest to the compressed sensing community. The complexities for the case of random Gaussian design matrix $\mathbf{X}$ are summarised in \Cref{table2}.

Finally, we prove query lower bounds for any classical and quantum algorithms for Lasso. The proof techniques come mostly from the work of Chen and de Wolf~\cite{chen2021quantum} but adapted to the Lasso problem from \Cref{eq:Lasso} instead of $\min_{\bfbeta\in\mathbb{R}^n}\frac{1}{n}\|\mathbf{y} - \mathbf{X}\bfbeta\|_2^2$ such that $\|\bfbeta\|_1 \leq 1$.
\begin{result}[\Cref{thr:lower_bound,thr:classical_lower_bound}]
    Let $n < d$ such that $n = \Omega(\log{d})$, $\lambda = \Omega(n)$, and $\epsilon \in (0,1/5)$ such that $\epsilon = \Omega(\sqrt{n^3\log{d}}/\lambda^2)$. There is $\mathbf{X}\in\{-1,1\}^{n\times d}$ such that every bounded-error classical or quantum algorithm that computes an $\epsilon\lambda\|\widetilde{\bfbeta}\|_1$-minimiser for the Lasso $\mathcal{L}^{(\lambda)}(\bfbeta)$ uses
    \begin{itemize}
        \item classical: $\Omega\Big(\frac{n^4 d}{\lambda^4\epsilon^{2}\log{d}} \Big)$ queries to $\mathbf{X}$. For $\epsilon = \Theta(\sqrt{n^3\log{d}}/\lambda^2)$, this is $\Omega\Big(\frac{nd}{\log^2{d}}\Big)$;
        \item quantum: $\Omega\Big(\frac{n^3\sqrt{d}}{\lambda^3\epsilon^{3/2}} \Big)$ queries to $\mathbf{X}$. For $\epsilon = \Theta(\sqrt{n^3\log{d}}/\lambda^2)$, this is $\Omega\Big(\frac{n^{3/4}\sqrt{d}}{\log^{3/4}{d}}\Big)$.
    \end{itemize}
\end{result}

We stress that our lower bound does not contradict the results from \Cref{table2} since it is a worst-case bound: there is a design matrix $\mathbf{X}$ such that any classical or quantum algorithm requires $\Omega\Big(\frac{n^4 d}{\lambda^4\epsilon^{2}\log{d}} \Big)$ or $\Omega\Big(\frac{n^3\sqrt{d}}{\lambda^3\epsilon^{3/2}} \Big)$ queries to $\mathbf{X}$, respectively, which does not exclude the existence of easy instances, a trivial case being the all-$0$ matrix. Still, for the range of parameters $\lambda = \Theta(n)$ and $\epsilon = \Theta(\sqrt{n^3\log{d}}/\lambda^2) = \Theta(1)$, the time complexity per iteration of our approximate quantum LARS algorithm with a random Gaussian matrix is $\widetilde{O}(\sqrt{d})$, while for the approximate classical LARS it is $\widetilde{O}(d)$. Together with $O(n)$ iterations in order to reach the desired $\lambda$, their final complexities would be $\widetilde{O}(nd)$ (classical) and $\widetilde{O}(n\sqrt{d})$ (quantum), larger than their respective lower bounds $\widetilde{\Omega}(nd)$ and $\widetilde{\Omega}(n^{3/4}\sqrt{d})$. We note that, in this case, the time complexity of the usual classical LARS algorithm is $O(n^2 d)$ (although the solution has no error).

\begin{table}[t]
\centering 
\def\arraystretch{1.25}
\begin{tabular}{ |p{5.2cm}|c|c|c|c| }
\hline \textbf{Algorithm} & \textbf{Error} & \textbf{Time complexity per iteration} \\
\hline
Classical LARS~\ref{alg:classicalLasso} & $0$ & $O(nd)$\\
\hline 
Simple quantum LARS~\ref{alg:warmup_classical_quantum} & $0$ & $\widetilde{O}(n\sqrt{d} + n|\mathcal{A}|)$\\
\hline 
Approximate classical LARS~\ref{alg:approximate_classical} & $\lambda\epsilon\|\widetilde{\bfbeta}\|_1$ & $\widetilde{O}(d/\epsilon^2 + n|\mathcal{A}|)$\\
\hline
Approximate quantum LARS~\ref{alg:classical_quantum} & $\lambda\epsilon\|\widetilde{\bfbeta}\|_1$ & $\widetilde{O}(\sqrt{d}/\epsilon + n|\mathcal{A}|)$\\
\hline 
\end{tabular}\caption{Summary of results for the case when the design matrix $\mathbf{X}\in\mathbb{R}^{n\times d}$ is a random Gaussian matrix with rows sampled i.i.d.\ from $\mathcal{N}(\mathbf{0},\mathbf{\Sigma})$ with well-behaved covariance matrix $\mathbf{\Sigma}\succ \mathbf{0}\in\mathbb{R}^{d\times d}$ and the vector of observations $\mathbf{y}\in\mathbb{R}^n$ is such that $\|\mathbf{y}\|_1/\|\mathbf{y}\|_2 = \poly\log{n}$. These results are valid for active and inactive sets $\mathcal{A}$ and $\mathcal{I}$ of size $|\mathcal{A}| = O(n/\log{d})$ and $|\mathcal{I}| = \Theta(d)$. The notation $\widetilde{O}(\cdot)$ omits $\poly\log$ terms in $n$, $d$, and the failure probability $\delta$.}
\label{table2}
\end{table}

Compared to the previous work of Chen and de Wolf~\cite{chen2021quantum} on quantum algorithms for Lasso, our work is different in that we study the pathwise Lasso regression problem as the tuning parameter $\lambda$ varies, and not focus on a single point $\lambda$. It looks infeasible to retrieve even a small interval $[\hat{\bfbeta}(\lambda),\hat{\bfbeta}(\lambda')]$ of Lasso solutions via their algorithms. On the other hand, if one wishes to compute $\hat{\bfbeta}(\lambda)$ for a specific $\lambda$, then their algorithms would be preferable. Both works are thus incomparable in terms of results. The output solution from the quantum algorithm of Chen and de Wolf can actually be used to initialise our algorithms instead of the point $(\lambda_0,\hat{\bfbeta}(\lambda_0) = \mathbf{0})$. Regarding quantum techniques, both works employ similar quantum subroutines like quantum amplitude estimation and quantum minimum finding. The fact that no advanced subroutines are needed to improve the dependence on both $n$ and $d$ should be seen as a strength of our (and~\cite{chen2021quantum}) work. 

The paper is organised as follows. In \cref{sect:preliminaries}, we introduce important notations, the computational model, and useful quantum subroutines. In \cref{sec:Lasso_path}, we describe the steps to obtain a Lasso path, conditions for the uniqueness of the Lasso solution, and analyse  the per-iteration time complexity of the classical LARS algorithm. In \cref{sect:quantum_algorithm}, we describe our two quantum LARS algorithms and a classical equivalent algorithm based on sampling. We discuss bounds in the noisy regime in \cref{sect:bounds_noisy} and summarise our work with possible future directions in \cref{sect:discussion}. 

\subsection{Related work}

Seeing the importance of variable selection in the high-dimensional data context, algorithms such as best subsets, forward selection, and backward elimination are well studied and widely used to produce sparse solutions~\cite{mao2002fast, mao2004orthogonal, whitley2000unsupervised, ververidis2005sequential, borboudakis2019forward, reif2014efficient, tan2008genetic, zongker1996algorithms, kumar2014feature, wei2006feature}. Apart from the aforementioned feature selection algorithms, Ridge regression, proposed by Hoerl and Kennard~\cite{hoerl1970ridge}, is another popular and well-studied method~\cite{mcdonald2009ridge, dorugade2014new, hoerl1970ridge, marquardt1975ridge, van2015lecture, hoerl1975ridge, saunders1998ridge, kibria2003performance, vinod1978survey}.  

The homotopy method of Osborne \emph{et al.}~\cite{osborne2000Lasso,osborne2000new} and the LARS algorithm of Efron \textit{et al.}~\cite{efron2004least} are feature selection algorithms which use a less greedy approach as compared to classical forward selection approaches. The LARS algorithm can be modified to give rise to algorithms for Lasso and for the efficient implementation of Forward Stagewise linear regression. Based on the work of Efron \textit{et al.}~\cite{efron2004least}, Rosset and Zhu~\cite{rosset2007piecewise} gave a general characterization of loss-penalty pairs to allow for efficient generation of the full regularised coefficient paths. 
Mairal and Yu~\cite{mairal2012complexity} later provided an $\epsilon$-approximate path for Lasso with at most $O(1/\sqrt{\smash[b]{\epsilon}})$ segments can be obtained by developing an approximate homotopy algorithm based on newly defined $\epsilon$-approximate optimality conditions. 
Other works on algorithms for Lasso and its variants include Refs.~\cite{roth2004generalized,tibshirani2013Lasso,stojnic2013framework,kim2008gradient,gaines2018algorithms,arnold2016efficient}. 

In the quantum setting, considerable amount of work has been done on linear regression using the (ordinary/unregularised) least squares approach. Kaneko \emph{et al.}~\cite{kaneko2021linear} proposed a quantum algorithm for linear regression which outputs classical estimates of the regression coefficients. This improves upon the earlier works~\cite{wiebe2012quantum, wang2017quantum} where the regression coefficients are encoded in the amplitudes of a quantum state. Chakraborty \textit{et al.}~\cite{chakraborty2019power} devised quantum algorithms for the weighted and generalised variants of the least squares problem by using block-encoding techniques and Hamiltonian simulation~\cite{low2019hamiltonian} to improve the quantum liner systems solver. As an application of linear regression to machine learning tasks, Schuld \textit{et al.}~\cite{schuld2016prediction} designed a quantum pattern recognition algorithm based on the ordinary least squares approach, using ideas from~\cite{harrow2009quantum,lloyd2014quantum}. Kerenidis and Prakash~\cite{kerenidis2020quantum} provided a quantum stochastic gradient descent algorithm for the weighted least squares problem using a quantum linear systems solver in the QRAM data structure model~\cite{prakash2014quantum, kerenidis2016quantum}, which allows for efficient quantum state preparation. 

While quantum linear regression has been relatively well studied in the unregularised setting, research on algorithms for regression using regularised least squares is not yet well established. In the sparse access model, Yu \textit{et al.}~\cite{yu2019improved} presented a quantum algorithm for Ridge regression using a technique called parallel Hamiltonian simulation to develop the quantum analogue of $K$-fold cross-validation~\cite{hoerl1970ridge}, which serves as an efficient estimator of Ridge regression. Other works on quantum Ridge regression algorithms in the sparse access model include~\cite{shao2020quantum,shao2023improved,chen2023faster}. Chen and de Wolf~\cite{chen2021quantum} designed quantum algorithms for Lasso and Ridge regression from the perspective of empirical loss minimisation. 
Both of their quantum algorithms output a classical vector whose loss is $\epsilon$-close to the minimum achievable loss. 
Around the same time, Bellante and Zanero~\cite{bellante2022quantum,bellante2024phdthesis} gave a polynomial speedup for the classical matching-pursuit algorithm, which is a heuristic algorithm for the best subset selection model. More recently, Chakraborty \textit{et al.}~\cite{chakraborty2023quantum} gave the first quantum algorithms for least squares with general $\ell_2$-norm regularisation, which includes regularised versions of quantum ordinary least squares, quantum weighted least squares, and quantum generalised least squares. Their algorithms use block-encoding techniques and the framework of quantum singular value transformation~\cite{gilyen2019quantum}, and demonstrate substantial improvement compared to previous results on quantum Ridge regression~\cite{shao2020quantum, chen2023faster, yu2019improved}.

\section{Preliminaries}\label{sect:preliminaries}
\subsection{Notations}
For a positive integer $n\in\mathbb{N}$, let $[n]\triangleq\{1, \dots, n\}$. Given a vector $\mathbf{u}\in\mathbb R^n$, we denote its $i$-th entry as $u_i$ for $i\in[n]$. Let $\|\mathbf{u}\|_p \triangleq \big(\sum_{i=1}^n |x_i|^p\big)^{1/p}$ and by $\mathcal{D}_{\mathbf{u}}$ we denote the distribution over $[n]$ with probability density function $\mathcal{D}_{\mathbf{u}}(i) = |u_i|/\|\mathbf{u}\|_1$. Given a matrix $\mathbf{X}\in\mathbb R^{n\times d}$, we denote its $i$-th column as $\mathbf{X}_{i}$ for $i\in[d]$. Let $\mathbf{I}$ be the identity matrix. Given $\mathcal{A} = \{i_1,\dots,i_k\}\subseteq[d]$, define $\mathcal{A}^c = [d]\setminus\mathcal{A}$ and let $\mathbf{X}_{\mathcal{A}} \in \mathbb{R}^{n\times |\mathcal{A}|}$ be the matrix $\mathbf{X}_{\mathcal{A}} = [\mathbf{X}_{i_1}, \dots, \mathbf{X}_{i_k}]$ formed by the columns of $\mathbf{X}$ in $\mathcal{A}$. The previous notation naturally extends to vectors, i.e., $\mathbf{u}_{\mathcal{A}} = [u_{i_1},\dots,u_{i_k}]$. Given $\mathcal{A}\subseteq[d]$ and $S\subseteq[n]$, $\mathbf{X}_{S\mathcal{A}}$ is the submatrix of $\mathbf{X}$ with rows in $S$ and columns in $\mathcal{A}$. We denote by $\operatorname{col}(\mathbf{X})$, $\operatorname{row}(\mathbf{X})$, $\operatorname{null}(\mathbf{X})$, and $\operatorname{rank}(\mathbf{X})$ the column space, row space, null space, and the rank, respectively, of $\mathbf{X}$. We denote by $\mathbf X^+$ the Moore-Penrose inverse of $\mathbf X$. If $\mathbf{X}$ has linearly independent columns, $\mathbf{X}^+ = (\mathbf{X}^\top \mathbf{X})^{-1} \mathbf{X}^\top$. Note that $\mathbf{X}_{\mathcal{A}}\mathbf{X}_{\mathcal{A}}^+$ is the orthogonal projector onto $\operatorname{col}(\mathbf{X}_{\mathcal{A}})$. Let $\|\mathbf{X}\|_p \triangleq \sup_{\mathbf{u}\in\mathbb{R}^d:\|\mathbf{u}\|_p=1}\|\mathbf{X}\mathbf{u}\|_p$ be the $p$-norm of $\mathbf{X}$, $p\in[1,\infty]$. As special cases, $\|\mathbf{X}\|_1 = \max_{j\in[d]}\sum_{i=1}^n |X_{ij}|$ and $\|\mathbf{X}\|_\infty = \max_{i\in[n]}\sum_{j=1}^d |X_{ij}|$. Let also $\lVert \mathbf X\rVert_{\max} \triangleq \max_{i\in[n],j\in[d]}\vert X_{ij}\vert$. 
We use $\ket{\bar 0}$ to denote the state $\ket{0}\otimes \cdots\otimes\ket{0}$, where the number of qubits is clear from the context. 
We use $\widetilde O(\cdot)$ to hide polylogarithmic factors, i.e., $\widetilde{O}(f(n)) = O(f(n)\cdot \operatorname{poly}\log(f(n)))$. See \cref{app:symbols} for a summary of symbols.

\subsection{Concentration bounds}
We revise a few notions of probability theory and concentration bounds. Let $\mathcal{N}(\mathbf{0},\mathbf{\Sigma})$ be the multivariate Gaussian distribution with covariance matrix $\mathbf{\Sigma}\succ \mathbf{0} \in\mathbb{R}^{d\times d}$. We say that a matrix $\mathbf{X}\in\mathbb{R}^{n\times d}$ is drawn from the $\mathbf{\Sigma}$-Gaussian ensemble if the rows of $\mathbf{X}$ are sampled i.i.d.\ according to $\mathcal{N}(\mathbf{0},\mathbf{\Sigma})$. We note that $\mathbf{Z}$ is drawn from the $\mathbf{I}$-Gaussian ensemble if and only if $\mathbf{X} = \mathbf{Z}\sqrt{\mathbf{\Sigma}}$ is drawn from the $\mathbf{\Sigma}$-Gaussian ensemble. We shall use the concept of sub-Gaussian variable.
\begin{definition}[Sub-Gaussianity]\label{def:SG}
    A random single-variable $z$ is sub-Gaussian with parameter $\sigma^2$ if $\mathbb{E}[e^{t(z-\mathbb{E}[z])}] \leq \exp\big(\frac{t^2 \sigma^2}{2}\big)$ for all $t\in\mathbb{R}$.
\end{definition}

We note that any sub-Gaussian random variable satisfies a Chernoff bound.
\begin{fact}\label{fact:Chernoff}
    If $z$ is a sub-Gaussian random variable with parameter $\sigma^2$, then
    \begin{align*}
        \mathbb{P}[z-\mathbb{E}[z]>t] \leq \exp{\left(-\frac{t^2}{2\sigma^2}  \right)} \quad\text{and}\quad \mathbb{P}[|z-\mathbb{E}[z]|>t] \leq 2\exp{\left(-\frac{t^2}{2\sigma^2}  \right)} \quad\text{for all}~t\in\mathbb{R}.
    \end{align*}
\end{fact}
\begin{fact}\label{fact:sub_gaussian}
    If each coordinate of $\mathbf{w}\in\mathbb{R}^n$ is an independent sub-Gaussian random variable with parameter $\sigma^2$, then, for any $\mathbf{u}\in \mathbb{R}^{n}$, $\mathbf{u}^\top \mathbf{w}$ is sub-Gaussian with parameter $\sigma^2\|\mathbf{u}\|_2^2$. 
\end{fact}
\begin{proof}
    $\mathbb{E}\left[ \exp\left( t\, \mathbf{u}^\top \mathbf{w}\right)\right]
        = \prod_{i=1}^n\mathbb{E}\left[ \exp\left( t\, u_i w_i\right)\right] 
        \leq  \prod_{i=1}^n\exp\left( \frac{t^2\sigma^2 u_i^2}{2}\right)
        =\exp\left( \frac{t^2 \sigma^2 \norm{\mathbf{u}}_2^2}{2}\right)$, $\forall t\in\mathbb{R}$.
\end{proof}

More generally, the sum of bounded independent random variables can be bounded via Hoeffding's inequality.
\begin{fact}[Hoeffding's bound]\label{fact:Hoeffding}
    Let $z_1,\dots,z_n$ be independent random variables such that, for $i\in[n]$, $a_{i}\leq z_{i}\leq b_{i}$ almost surely. Let $z \triangleq \sum_{i=1}^n z_i$. Then
    \begin{align*}
        \mathbb{P}[|z - \mathbb{E}[z]| \geq t] \leq 2\exp\left(- \frac{2t^2}{\sum_{i=1}^n (b_i - a_i)^2} \right) \quad\text{for all}~t>0.
    \end{align*}
\end{fact}

L\'evy's lemma~\cite{milman1986asymptotic,ledoux2001concentration,Anderson_Guionnet_Zeitouni_2009,Vershynin_2018,Meckes_2019} will also be useful.
\begin{fact}[L\'evy's lemma, {\cite[Corollary~4.4.28]{Anderson_Guionnet_Zeitouni_2009}}]\label{fact:levy_lemma}
    Let $f: \mathbb{S}^d\to \mathbb{R}$ be a function defined on the $d$-dimensional hypersphere~$\mathbb{S}^d$. Assume $f$ is $K$-Lipschitz, i.e., $|f(\psi) - f(\phi)| \leq K\|\psi - \phi\|$. Then
    \begin{align*}
        \operatorname*{\mathbb{P}}_{\psi\sim\mu_{\rm Haar}}[|f(\psi) - \mathbb{E}[f]| \geq \epsilon] \leq 2\exp\left(-\frac{d\epsilon^2}{4 K^2}\right) \quad\text{for all}~\epsilon > 0.
    \end{align*}
\end{fact}

We will also make use of the following lemma about the relative approximation of a ratio of two real numbers that are given by relative approximations. 
 
\begin{lemma}\label{lem:error_propagation}
    Let $\widetilde{a},\widetilde{b}\in\mathbb{R}$ be estimates of $a,b\in\mathbb{R}\setminus\{0\}$, respectively, such that $|a - \widetilde{a}| \leq \epsilon_a$ and $|b - \widetilde{b}| \leq \epsilon_b$, where $\epsilon_a \geq 0$ and $\epsilon_b\in[0,|b|/2]$. Then 
    \begin{align*}
        \left|\frac{\widetilde{a}}{\widetilde{b}} - \frac{a}{b} \right| \leq 2\frac{|a|}{|b|}\left(\frac{\epsilon_a}{|a|} + \frac{\epsilon_b}{|b|}\right).
    \end{align*}
\end{lemma}
\begin{proof}
    First note that $|b| - |\widetilde{b}| \leq \epsilon_b \implies \frac{1}{|\widetilde{b}|} \leq \frac{1}{|b| - \epsilon_b} \leq \frac{2}{|b|}$. Then
    \begin{align*}
        \left|\frac{\widetilde{a}}{\widetilde{b}} - \frac{a}{b} \right| =  \left|\frac{\widetilde{a}b - ab + ab - a\widetilde{b}}{b\widetilde{b}} \right| \leq \left|\frac{\widetilde{a} - a}{\widetilde{b}} \right| + \left|\frac{a(b-\widetilde{b})}{b\widetilde{b}} \right| &\leq \frac{2\epsilon_a}{|b|} + \frac{|a|}{|b|}\frac{2\epsilon_b}{|b|} = 2\frac{|a|}{|b|}\left(\frac{\epsilon_a}{|a|} + \frac{\epsilon_b}{|b|}\right). \qedhere
    \end{align*}
\end{proof}

\subsection{Computational model}
\label{sec:computational_model}

\paragraph*{Classical computational model.} Our classical computational model is a classical random-access machine. The input to the Lasso problem is a matrix $\mathbf{X}\in\mathbb{R}^{n\times d}$ and a vector $\mathbf{y}\in\mathbb{R}^n$, which are stored in a classical-readable read-only memory (ROM). For simplicity, reading any entry of $\mathbf{X}$ and $\mathbf{y}$ takes constant time. 
The classical computer can write bits to a classical-samplable-and-writable memory. More specifically, the classical computer can write a vector $\mathbf{u}\in\mathbb{R}^m$ (or parts of it) into a low-overhead samplable data structure in time $O(m)$, which allows it to sample an index $i\in[m]$ from the probability distribution $\mathcal{D}_{\mathbf{u}}$ in time $O(\poly\log{m})$. We refer to the ability to sample from $\mathcal{D}_{\mathbf{u}}$ as having sampling access to the vector $\mathbf{u}$. In this work, all classical-samplable-and-writable memories have size $m=O(n)$. We assume an arithmetic model in which we ignore issues arising from the fixed-point representation of real numbers. All basic arithmetic operations in this model take constant time.

\paragraph*{Quantum computational model.} Our quantum computational model is a classical random-access machine with access to a quantum computer. The input to the Lasso problem is a matrix $\mathbf{X}\in\mathbb{R}^{n\times d}$ and a vector $\mathbf{y}\in\mathbb{R}^n$, which are stored in a classical-readable read-only memory (ROM). For simplicity, reading any entry of $\mathbf{X}$ and $\mathbf{y}$ takes constant time. We assume that $\mathbf{X}$ is also stored in a quantum-readable read-only memory (QROM), whose single entries can be queried. This means that the quantum computer has access to an oracle $\mathcal{O}_{\mathbf{X}}$ that performs the mapping
\begin{align*}
    \mathcal{O}_{\mathbf{X}} : |i,j\rangle|\bar{0}\rangle \mapsto |i,j\rangle|X_{ij}\rangle, \quad \forall i\in[n], j\in[d],
\end{align*}
in time $O(\operatorname{poly}\log(nd))$. The classical computer can also write bits to a quantum-readable classical-writable classical memory (QRAM)~\cite{giovannetti2008architectures,giovannetti2008quantum,allcock2023constant}, which can be accessed in superposition by the quantum computer. More specifically, the classical computer can write a vector $\mathbf{u}\in\mathbb{R}^m$ (or parts of it) into the memory of a QRAM in time $O(m)$, which allows the quantum computer to invoke an oracle $\mathcal{U}_{\mathbf{u}}$ that performs the mapping
\begin{align*}
    \mathcal{U}_{\mathbf{u}} : |i\rangle|\bar{0}\rangle \mapsto |i\rangle|u_{i}\rangle, \quad \forall i\in[m],
\end{align*}
in time $O(\operatorname{poly}\log{m})$. We refer to the ability to invoke $\mathcal{U}_{\mathbf{u}}$ as having quantum access to the vector $\mathbf{u}$. In this work, all QRAMs have size $m = O(n)$. We do not assume that $\mathbf{y}$ has been previously stored in a QROM. Instead, it can be stored in a QRAM in time $O(n)$ by the classical computer.

The classical computer can send the description of a quantum circuit to the quantum computer, which is a sequence of quantum gates from a universal gate set plus queries to the bits stored in the QROM and/or QRAM; the quantum computer runs the circuit, performs a measurement in the computational basis, and returns the measurement outcome to the classical computer. We refer to the runtime of a classical/quantum computation as the number of basic gates performed plus the time complexity of all the calls to the QROM and QRAMs. We assume an arithmetic model in which we ignore issues arising from the fixed-point representation of real numbers. All basic arithmetic operations in this model take constant time. 

Kerenidis and Prakash~\cite{prakash2014quantum,kerenidis2016quantum} introduced a classical data structure (sub-sequentially called KP-tree) to store a vector $\mathbf{u}\in\mathbb{R}^m$ to enable the efficient preparation of the state
\begin{align*}
    \sum_{i=1}^m \sqrt{\frac{|u_{i}|}{\|\mathbf{u}\|_1}}\ket{i}\ket{\operatorname{sign}(u_i)}.
\end{align*}
We shall commonly use KP-trees of auxiliary vectors in order to prepare the above states.
\begin{fact}[KP-tree~\cite{prakash2014quantum,kerenidis2016quantum,chen2021quantum}]\label{KP-tree}
    Let $\mathbf{u}\in\mathbb R^{m}$ be a vector with $w \in \mathbb N$ non-zero entries. There is a data structure called ${\rm KP}$-tree of size $O(w\poly\log{m})$ that stores each input $(i,u_{i})$ in time $O(\poly\log{m})$. After the data structure has been constructed, there is a quantum algorithm that prepares the state $\sum_{i=1}^m \sqrt{\frac{|u_{i}|}{\|\mathbf{u}\|_1}}\ket{i}\ket{\operatorname{sign}(u_i)}$ up to negligible error in time $O(\operatorname{poly}\log{m})$. 
\end{fact}
\begin{remark}
    It can be argued that the {\rm QROM/QRAM} model is too strong of an assumption and it might not be realistic~{\rm \cite{jaques2023qram}}. Nonetheless, we believe it is a natural and fair quantum equivalent to the classical computational model. If we make the assumption that the input can be classically accessed in constant time and/or sampled in poly-logarithmically time, then a quantum computer should also require, in all fairness, a comparable amount of time to quantumly access the input.
\end{remark}

\subsection{Quantum subroutines}

In this section, we review some useful quantum subroutines that shall be used in the rest of the paper, starting with the quantum minimum-finding algorithm from D\"urr and H\o{}yer~\cite{durr1996quantum}.

\begin{fact}[Quantum minimum finding~\cite{durr1996quantum}]\label{quantum_minimum_finding}
    Given quantum access to a vector $\mathbf{u}\in\mathbb R^m$, there is a quantum algorithm that finds $\min_{i\in[m]} u_i$ with success probability $1-\delta$ in time $\widetilde{O}(\sqrt{m}\log \frac{1}{\delta})$. 
\end{fact}
The above subroutine was later generalised by Chen and de Wolf~\cite{chen2021quantum} in the case when one has quantum access to the entries of $\mathbf{u}$ up to some additive error. We note that a similar result, but in a different setting, appeared in~\cite{van2017quantum,quek2020robust}.
\begin{fact}[{\cite[Theorem~2.4]{chen2021quantum}}]\label{fact:minimum_finding}
    Let $\delta_1,\delta_2\in(0,1)$ such that $\delta_2 = O\big(\delta_1^2/(m\log(1/\delta_1))\big)$, $\epsilon>0$, and $\mathbf{u}\in\mathbb{R}^m$. Suppose access to a unitary that maps $|k\rangle|\bar{0}\rangle \mapsto |k\rangle|R_k\rangle$ such that, for every $k\in[m]$, after measuring the state $|R_k\rangle$, with probability at least $1-\delta_2$ the first register $r_k$ of the measurement outcome satisfies $|r_k - u_k| \leq \epsilon$. Then there is a quantum algorithm that finds an index $k\in[m]$ such that $u_k \leq \min_{j\in[m]}u_j + 2\epsilon$ with probability at least $1-\delta_1$ and in time $\widetilde{O}(\sqrt{m}\log(1/\delta_1))$.
\end{fact}
We shall also need the amplitude estimation subroutine from Brassard \emph{et al.}~\cite{brassard2002quantum}.
\begin{fact}[{\cite[Theorem~12]{brassard2002quantum}}]\label{fact:amplitude_estimation}
    Given $M\in\mathbb{N}$ and access to an $(n+1)$-qubit unitary $\mathbf{U}$ satisfying
    \begin{align*}
        \mathbf{U}|0^n\rangle|0\rangle = \sqrt{a}|\psi_1\rangle|1\rangle + \sqrt{1-a}|\psi_0\rangle|0\rangle,
    \end{align*}
    where $|\psi_0\rangle$ and $|\psi_1\rangle$ are arbitrary $n$-qubit states and $a\in(0,1)$, there is a quantum algorithm that uses $O(M)$ applications of $\mathbf{U}$ and $\mathbf{U}^\dagger$ and $\widetilde{O}(M)$ elementary gates, and outputs a state $|\Lambda\rangle$ such that, after measuring $|\Lambda\rangle$, with probability at least $9/10$, the first register $\lambda$ of the outcome satisfies
    \begin{align*}
        |a - \lambda| \leq \frac{\sqrt{a(1-a)}}{M} + \frac{1}{M^2}.
    \end{align*}
\end{fact}
Finally, the following result, which is based on~\cite[Theorem~3.4]{chen2021quantum}, constructs the mapping $|j\rangle|\bar{0}\rangle \mapsto |j\rangle|R_j\rangle$ such that $R_j$ holds an approximation to $\mathbf{A}_j^\top \mathbf{u}$ up to an additive error, where $\mathbf{A}\in\mathbb{R}^{n\times d}$ and $\mathbf{u}\in\mathbb{R}^n$ can be accessed quantumly via KP-trees.

\begin{lemma}\label{lem:terms_estimation}
    Given $\mathbf{A}\in\mathbb{R}^{n\times d}$ and $\mathbf{u}\in\mathbb{R}^n$, assume they are stored in {\rm KP}-trees and we have quantum access to their entries. Assume that $\|\mathbf{u}\|_\infty$ and $\|\mathbf{A}_j\|_\infty$, $j\in[d]$, are known. Let $\delta\in(0,1)$ and $\epsilon > 0$. Then there is a quantum algorithm that implements the mapping $|j\rangle|\bar{0}\rangle \mapsto |j\rangle|R_j\rangle$ such that, for all $j\in[d]$, after measuring the state $|R_j\rangle$, with probability at least $1-\delta$, the outcome $r_j$ of the first register satisfies
    \begin{align*}
        \left|r_j - \mathbf{A}_j^\top \mathbf{u} \right| \leq \epsilon\min(\|\mathbf{A}_j\|_\infty\|\mathbf{u}\|_1,\|\mathbf{A}_j\|_1\|\mathbf{u}\|_\infty)
    \end{align*}
    in time $O(\epsilon^{-1}\log(1/\delta)\poly\log(nd))$.
\end{lemma}
\begin{proof}
    Fix $j\in[d]$. We start by approximating $\mathbf{A}_j^\top \mathbf{u} = \sum_{i=1}^n A_{ij}u_i$ up to error $\epsilon\|\mathbf{A}_j\|_{\infty}\|\mathbf{u}\|_1$. Apply the operator from \cref{KP-tree} to create 
    \begin{align}\label{state1}
        \sum_{i=1}^n \sqrt{\frac{|u_i|}{\|\mathbf{u}\|_1}} |i\rangle|s_i\rangle|\bar{0}\rangle|0\rangle
    \end{align}
    in time $O(\poly\log{n})$, where $s_i \triangleq \operatorname{sign}(u_i)$. On the other hand, using query access to $\mathbf{A}$ create the mapping $|i\rangle|s_i\rangle|\bar{0}\rangle \mapsto |i\rangle|s_i\rangle|s_i A_{ij}\rangle$ costing time $O(\poly\log(nd))$. Using this operator on the first three registers in \cref{state1}, we obtain
    \begin{align}\label{state2}
        \sum_{i=1}^n \sqrt{\frac{|u_i|}{\|\mathbf{u}\|_1}} |i\rangle|s_i\rangle|s_i A_{ij}\rangle|0\rangle.
    \end{align}
    Now, define the positive-controlled rotation for each $a\in\mathbb{R}$ such that
    \begin{align*}
        U_{\rm CR^+} :|a\rangle|0\rangle \mapsto \begin{cases}
            |a\rangle(\sqrt{a}|1\rangle + \sqrt{1-a}|0\rangle) \quad &\text{if}~a\in(0,1],\\
            |a\rangle|0\rangle \quad &\text{otherwise}.
        \end{cases}
    \end{align*}
    Applying $U_{\rm CR^+}$ on the last two registers in \cref{state2}, we obtain
    \begin{align*}
        &\sum_{\substack{i\in[n] \\ s_i A_{ij} > 0}} \sqrt{\frac{|u_i|}{\|\mathbf{u}\|_1}} |i\rangle|s_i\rangle|s_i A_{ij}\rangle\left(\sqrt{\frac{s_i A_{ij}}{\|\mathbf{A}_j\|_{\infty}}}|1\rangle + \sqrt{1 - \frac{s_i A_{ij}}{\|\mathbf{A}_j\|_{\infty}}}|0\rangle\right) + \sum_{\substack{i\in[n] \\ s_i A_{ij} \leq 0}} \sqrt{\frac{|u_i|}{\|\mathbf{u}\|_1}}|i\rangle|s_i\rangle|s_i A_{ij}\rangle|0\rangle\\
    &\begin{multlined}[b][\textwidth]
        = \sum_{\substack{i\in[n] \\ s_i A_{ij} > 0}} \sqrt{\frac{A_{ij}u_i}{\|\mathbf{A}_j\|_{\infty}\|\mathbf{u}\|_1}}|i\rangle|s_i\rangle|s_i A_{ij}\rangle|1\rangle\\
        + \Bigg(\sum_{\substack{i\in[n] \\ s_i A_{ij} > 0}}  \sqrt{\frac{|u_i|}{\|\mathbf{u}\|_1}-\frac{A_{ij}u_i}{\|\mathbf{A}_j\|_{\infty}\|\mathbf{u}\|_1}}|i\rangle|s_i\rangle|s_i A_{ij}\rangle  + \sum_{\substack{i\in[n] \\ s_i A_{ij} \leq 0}} \sqrt{\frac{|u_i|}{\|\mathbf{u}\|_1}} |i\rangle|s_i\rangle|s_i A_{ij}\rangle\Bigg)|0\rangle
    \end{multlined}\\
    &= \sqrt{a_+}|\phi_1\rangle|1\rangle + \sqrt{1-a_+}|\phi_0\rangle|0\rangle, 
    \end{align*}
    where
    \begin{align*}
        a_+ = \sum_{\substack{i\in[n] \\ s_i A_{ij} > 0}} \frac{A_{ij}u_i}{\|\mathbf{A}_j\|_{\infty}\|\mathbf{u}\|_1} \quad\text{and}\quad \ket{\phi_1} =\sum_{\substack{i\in[n] \\ s_i A_{ij} > 0}} \sqrt{\frac{A_{ij}u_i}{\sum_{k\in[n], s_k A_{kj} > 0} A_{kj}u_k}} |i\rangle|s_i\rangle|s_i A_{ij}\rangle.
    \end{align*}
    Here, $\ket{\phi_1}$ and $\ket{\phi_0}$ are unit vectors.
    We now use \cref{fact:amplitude_estimation} to get an estimate $\widetilde{a}_+$ such that
    \begin{align*}
        |\widetilde{a}_+ - a_+ | \leq \frac{\epsilon}{2}\sqrt{a_+} \implies 
        \left|\|\mathbf{A}_j\|_{\infty}\|\mathbf{u}\|_1\widetilde{a}_+ - \!\sum_{\substack{i\in[n] \\ s_i A_{ij} > 0}} \!\! A_{ij}u_i\right| &\leq \frac{\epsilon}{2}\sqrt{\|\mathbf{A}_j\|_{\infty}\|\mathbf{u}\|_1 \!\sum_{\substack{i\in[n] \\ s_i A_{ij}>0}} \!\! A_{ij}u_i} \leq\frac{\epsilon}{2}\|\mathbf{A}_j\|_{\infty}\|\mathbf{u}\|_1
    \end{align*}
    in time $O(\epsilon^{-1}\log(1/\delta)\poly\log(nd))$ and probability at least $1-\delta/2$. Notice that we know $\|\mathbf{u}\|_1$ since it is stored in a KP-tree. Similarly, we estimate
    \begin{align*}
        a_- = - \sum_{\substack{i\in[n] \\ s_i A_{ij} < 0}} \frac{A_{ij}u_i}{\|\mathbf{A}_j\|_{\infty}\|\mathbf{u}\|_1}
    \end{align*}
    with additive error $\frac{\epsilon}{2}\|\mathbf{A}_j\|_{\infty}\|\mathbf{u}\|_1$. Thus $\mathbf{A}_j^\top\mathbf{u} = \sum_{i=1}^n A_{ij}u_i = \|\mathbf{A}_j\|_{\infty}\|\mathbf{u}\|_1(a_+ - a_-)$ is estimated with additive error $\epsilon\|\mathbf{A}_j\|_{\infty}\|\mathbf{u}\|_1$ and success probability at least $1-\delta$ in time $O(\epsilon^{-1}\log(1/\delta)\poly\log(nd))$. 

    Swapping the roles of $\mathbf{A}_j$ and $\mathbf{u}$ and repeating the above steps leads to an estimate of $\mathbf{A}_j^\top\mathbf{u}$ with additive error $\epsilon\|\mathbf{A}_j\|_1\|\mathbf{u}\|_\infty$ in time $O(\epsilon^{-1}\log(1/\delta)\poly\log(nd))$.
\end{proof}

It is possible to prove a classical result analogous to the one above. More specifically, if we have sampling access to $\mathbf{A}\in\mathbb{R}^{n\times d}$ and $\mathbf{u}\in\mathbb{R}^n$, then there is a classical algorithm that allows to approximate $\mathbf{A}_j^\top \mathbf{u}$ up to an additive error.
\begin{lemma}\label{lem:sampling-based_inner_product_estimation}
    Let $\mathbf{A}\in\mathbb{R}^{n\times d}$ and $\mathbf{u}\in\mathbb{R}^n$. Assume sampling access to $\mathbf{u}$ and to the columns of $\mathbf{A}$. Let $\delta\in(0,1)$ and $\epsilon > 0$. There is a classical algorithm that, for any $j\in[d]$, outputs $r_j$ such that 
    \begin{align*}
        | r_j - \mathbf{A}_j^\top \mathbf{u} |\leq \epsilon\min(\|\mathbf{A}_j\|_\infty\|\mathbf{u}\|_1,\|\mathbf{A}_j\|_1\|\mathbf{u}\|_\infty)
    \end{align*}
    with probability at least $1-\delta$ and in time $O(\epsilon^{-2}\log(1/\delta)\poly\log(nd))$.
\end{lemma}
\begin{proof}
    Fix $j\in[d]$. Let $z$ be the random variable defined by
    \begin{align*}
        \mathbb{P}\left[z = \|\mathbf{A}_j\|_1u_i \operatorname{sign}(A_{ij})\right] = \frac{|A_{ij}|}{\|\mathbf{A}_j\|_1} \quad\text{for}~ i\in[n].
    \end{align*}
    Then
    \begin{align*}
        \mathbb{E}[z] &= \sum_{i=1}^n \frac{|A_{ij}|}{\|\mathbf{A}_j\|_1}\|\mathbf{A}_j\|_1u_i \operatorname{sign}(A_{ij}) = \mathbf{A}_j^\top \mathbf{u}.
    \end{align*}
    Sample $q = \lceil 2\epsilon^{-2}\ln(2/\delta)\rceil$ indices $\{i_1,\dots,i_q\}\subseteq[n]$ from the distribution $\mathcal{D}_{\mathbf{A}_j}$ by using sampling access to $\mathbf{A}$ and set $z_k = \|\mathbf{A}_j\|_1 u_{i_k} \operatorname{sign}(A_{i_k j})$, $k\in[q]$. The algorithm outputs $\widehat{z} = \frac{1}{q}\sum_{k=1}^q z_k$ as an estimate for $\mathbf{A}_j^\top \mathbf{u}$. Since $z_1,\dots,z_q$ are i.i.d.\ copies of $z$, a Hoeffding's bound (\cref{fact:Hoeffding}) gives
    \begin{align*}
        \mathbb{P}[|\widehat{z} - \mathbf{A}_j^\top \mathbf{u}| \geq \epsilon\|\mathbf{A}_j\|_1\|\mathbf{u}\|_\infty] \leq 2e^{-\epsilon^2q/2} = \delta,
    \end{align*}
    so the classical algorithm approximates $\mathbf{A}_j^\top \mathbf{u}$ with additive error $\epsilon\|\mathbf{A}_j\|_1\|\mathbf{u}\|_\infty$ and success probability at least $1-\delta$. The final runtime is the number of samples $q$ times the complexity $O(\poly\log(nd))$ of sampling each index. Finally, it is possible to repeat the same procedure but swapping the roles of $\mathbf{A}_j$ and $\mathbf{u}$, so that $z$ is now defined as $\mathbb{P}\left[z = \|\mathbf{u}\|_1 A_{ij} \operatorname{sign}(u_i)\right] = |u_i|/\|\mathbf{u}\|_1$, $i\in[n]$.
\end{proof}

\section{Lasso path and the LARS algorithm}\label{sec:Lasso_path}
\subsection{Lasso solution}

In this section, we review several known properties of the solution to the Lasso regression problem with inputs $(\mathbf{X},\mathbf{y})\in\mathbb{R}^{n\times d}\times\mathbb{R}^n$ defined as
\begin{align}\label{eq:lasso_problem}
    \hat{\bfbeta} \in \argmin_{\bfbeta \in \mathbb{R}^d} \mathcal{L}^{(\lambda)}_{\mathbf{X},\mathbf{y}}(\bfbeta) \quad\text{where}\quad \mathcal{L}^{(\lambda)}_{\mathbf{X},\mathbf{y}}(\bfbeta) \triangleq \frac{1}{2}\|\mathbf{y}-\mathbf{X}{\bfbeta}\|_2^2 + \lambda \|\bfbeta\|_1 \quad \text{for}~ \lambda >0.
\end{align}
Here $\mathbf{X}\in\mathbb{R}^{n\times d}$ is the design matrix and $\mathbf{y}\in\mathbb{R}^n$ is a vector of observations. When the inputs are clear from context, we shall simply write $\mathcal{L}^{(\lambda)}(\bfbeta)$.
We will cover optimality conditions for the Lasso problem, the general form for its solution, and continuity and uniqueness conditions. Most of these results can be found in~\cite{rosset2007piecewise,mairal2012complexity,tibshirani2013Lasso,sjostrand2018spasm}. 
We start with obtaining the optimality conditions for the Lasso problem.
\begin{fact}[{\cite[Lemma~1]{tibshirani2013Lasso}}]\label{fact:fit_uniqueness}
    Every Lasso solution $\hat{\bfbeta}\in\mathbb{R}^d$ gives the same fitted value $\mathbf{X}\hat{\bfbeta}$.
\end{fact}
%
\begin{fact}
    \label{fact:KKT1}
    A vector $\hat{\bfbeta}\in\mathbb{R}^d$ is a solution to the Lasso problem if and only if
    \begin{align}\label{eq:KKT1}
        \mathbf{X}^\top (\mathbf{y}-\mathbf{X}\bfhatbeta) = \lambda \mathbf{s} \quad\text{where}~
        s_i \in \begin{cases}
            \{\sign(\hat{\beta}_i)\} &\text{if}~ \hat{\beta}_i \neq 0, \\
            [-1,1] &\text{if}~ \hat{\beta}_i = 0, \\
        \end{cases}\quad
        \text{for}~ i\in[d].
    \end{align}
    Let $\mathcal{A}\triangleq \{ i\in[d]:  |\mathbf{X}_i^\top(\mathbf{y}-\mathbf{X}{\bfhatbeta})| =\lambda\}$, $\mathcal{I}\triangleq[d]\setminus\mathcal{A}$, and $\bfeta \triangleq \operatorname{sign}(\mathbf{X}^\top (\mathbf{y} - \mathbf{X}\hat{\bfbeta}))\in\{-1,0,1\}^d$. Any Lasso solution $\hat{\bfbeta}$ is of the form
    \begin{align}\label{eq:lasso_solution}
        \hat{\bfbeta}_{\mathcal{I}} = \mathbf{0} \quad\text{and}\quad \hat{\bfbeta}_{\mathcal{A}} = \mathbf{X}_{\mathcal{A}}^+(\mathbf{y} - \lambda(\mathbf{X}_{\mathcal{A}}^+)^\top\bfeta_{\mathcal{A}}) + \mathbf{b},
    \end{align}
    where
    \begin{align}\label{eq:restriction_lasso}
        \mathbf{b}\in\operatorname{null}(\mathbf{X}_{\mathcal{A}}) \quad\text{and}\quad \eta_i\big([\mathbf{X}_{\mathcal{A}}^+(\mathbf{y} - \lambda(\mathbf{X}_{\mathcal{A}}^+)^\top\bfeta_{\mathcal{A}})]_i + b_i\big) \geq 0~\text{for}~i\in\mathcal{A}.
    \end{align}
\end{fact}
\begin{proof}
    \cref{eq:KKT1} can be obtained by considering the Karush–Kuhn–Tucker conditions, or sub-gradient optimality conditions. More specifically, a vector $\hat{\bfbeta}\in\mathbb{R}^d$ is a solution to the Lasso problem if and only if (see, e.g.,~\cite[Proposition~3.1.5]{borwein2006convex})
    \begin{align*}
        \mathbf{0} \in \{-\nabla\mathcal{L}^{(\lambda)}(\hat{\bfbeta}) = \mathbf{X}^\top (\mathbf{y} - \mathbf{X}\hat{\bfbeta}) - \lambda \mathbf{s}: \mathbf{s}\in\partial\|\hat{\bfbeta}\|_1\},
    \end{align*}
    where $\partial\|\hat{\bfbeta}\|_1$ denotes the sub-gradient of the $\ell_1$-norm at $\hat{\bfbeta}$. \cref{eq:KKT1} then follows from the fact that the sub-gradients $\mathbf{s}\in\mathbb{R}^d$ of $\|\hat{\bfbeta}\|_1$ are vectors such that $s_i = \operatorname{sign}(\hat{\beta}_i)$ if $\hat{\beta}_i \neq 0$ and $|s_i| \leq 1$ otherwise.

    Moving on, first note that $\bfeta$ is a sub-gradient and that the uniqueness of $\mathbf{X}\hat{\bfbeta}$ (\cref{fact:fit_uniqueness}) implies the uniqueness of $\mathcal{A}$ and $\bfeta$. By the definition of the sub-gradients $\mathbf{s}$ in \cref{eq:KKT1}, we know that $\hat{\bfbeta}_{\mathcal{I}} = \mathbf{0}$, and therefore the block $\mathcal{A}$ of \cref{eq:KKT1} can be written as
    \begin{align*}
        \mathbf{X}_{\mathcal{A}}^\top (\mathbf{y} - \mathbf{X}_{\mathcal{A}}\hat{\bfbeta}_{\mathcal{A}}) = \lambda \bfeta_{\mathcal{A}}.
    \end{align*}
    This means that $\bfeta_{\mathcal{A}}\in\operatorname{row}(\mathbf{X}_{\mathcal{A}})$, and so $\bfeta_{\mathcal{A}} = \mathbf{X}_{\mathcal{A}}^\top (\mathbf{X}_{\mathcal{A}}^\top)^+ \bfeta_{\mathcal{A}}$. Using this fact and rearranging the above equation, we obtain
    \begin{align*}
        \mathbf{X}_{\mathcal{A}}^\top\mathbf{X}_{\mathcal{A}}\hat{\bfbeta}_{\mathcal{A}} = \mathbf{X}_{\mathcal{A}}^\top (\mathbf{y} - \lambda (\mathbf{X}_{\mathcal{A}}^\top)^+ \bfeta_{\mathcal{A}}).
    \end{align*}
    Therefore, any solution $\hat{\bfbeta}_{\mathcal{A}}$ to the above equation is of the form
    \begin{align*}
        \hat{\bfbeta}_{\mathcal{A}} = \mathbf{X}_{\mathcal{A}}^+(\mathbf{y} - \lambda (\mathbf{X}_{\mathcal{A}}^\top)^+ \bfeta_{\mathcal{A}}) + \mathbf{b}, \quad\quad \mathbf{b}\in\operatorname{null}(\mathbf{X}_{\mathcal{A}}),
    \end{align*}
    where we used the identity $\mathbf{A}^+(\mathbf{A}^\top)^+\mathbf{A}^\top = \mathbf{A}^+$. Any $\mathbf{b}\in\operatorname{null}(\mathbf{X}_{\mathcal{A}})$ produces a valid Lasso solution $\hat{\bfbeta}_{\mathcal{A}}$ as long as $\hat{\bfbeta}_{\mathcal{A}}$ has the correct signs given by $\bfeta_{\mathcal{A}}$, i.e., $\operatorname{sign}(\hat{\beta}_{i}) = \eta_i$ for $i\in\mathcal{A}$. This restriction can be written as $\hat \beta_i \eta_i\geq 0$, which is \cref{eq:restriction_lasso}. This completes the proof.
\end{proof}
We shall abuse the definition of the Karush-Kuhn-Tucker (KKT) conditions and refer to \cref{eq:KKT1} as the KKT conditions themselves for the Lasso problem. The set
\begin{align*}
    \mathcal{A}\triangleq \{ i\in[d]:  |\mathbf{X}_i^\top(\mathbf{y}-\mathbf{X}{\bfhatbeta})| =\lambda\}
\end{align*}
and the sub-gradient vector
\begin{align*}
    \bfeta \triangleq \operatorname{sign}(\mathbf{X}^\top (\mathbf{y} - \mathbf{X}\hat{\bfbeta}))\in\{-1,0,1\}^d
\end{align*}
are normally referred to as \emph{active} (or \emph{equicorrelation}) \emph{set} and \emph{equicorrelation signs}, respectively. The set $\mathcal{I}\triangleq[d]\setminus\mathcal{A}$ is called \emph{inactive set}.
\begin{fact}[{\cite[Lemma~9]{tibshirani2013Lasso}}]
    For any $\mathbf{y}\in\mathbb{R}^n$, $\mathbf{X}\in\mathbb{R}^{n \times d}$, and $\lambda>0$, the Lasso solution
    \begin{align*}
        \hat{\bfbeta}_{\mathcal{I}} = \mathbf{0} \quad\text{and}\quad \hat{\bfbeta}_{\mathcal{A}} = \mathbf{X}_{\mathcal{A}}^+(\mathbf{y} - \lambda(\mathbf{X}_{\mathcal{A}}^+)^\top\bfeta_{\mathcal{A}})
    \end{align*}
    has the minimum $\ell_2$-norm over all Lasso solutions.
\end{fact}
\begin{proof}
    By \cref{fact:KKT1}, any Lasso solution has $\ell_2$-norm
    \begin{align*}
        \|\hat{\bfbeta}\|_2^2 = \|\mathbf{X}_{\mathcal{A}}^+(\mathbf{y} - \lambda(\mathbf{X}_{\mathcal{A}}^+)^\top\bfeta_{\mathcal{A}})\|^2 + \|\mathbf{b}\|^2,
    \end{align*}
    since $\mathbf{b}\in\operatorname{null}(\mathbf{X}_{\mathcal{A}})$. Hence the $\ell_2$-norm is minimised when $\mathbf{b}=\mathbf{0}$.
\end{proof}

It follows from \cref{fact:KKT1} that the Lasso solution is unique if $\operatorname{null}(\mathbf{X}_{\mathcal{A}}) = \{\mathbf{0}\}$ and is given by
\begin{align}\label{eq:lasso_solution_unique}
    \hat{\bfbeta}_{\mathcal{I}} = \mathbf{0} \quad\text{and}\quad \hat{\bfbeta}_{\mathcal{A}} = (\mathbf{X}_{\mathcal{A}}^\top \mathbf{X}_{\mathcal{A}})^{-1}(\mathbf{X}_{\mathcal{A}}^\top \mathbf{y} - \lambda\bfeta_{\mathcal{A}}).
\end{align}
It is known~\cite[Lemma~14]{tibshirani2013Lasso} that there always is a Lasso solution such that $|\mathcal{A}| \leq \min\{n,d\}$. If the Lasso solution is unique, then the active set has size at most $\min\{n,d\}$. The sufficient condition $\operatorname{null}(\mathbf{X}_{\mathcal{A}}) = \{\mathbf{0}\}$ for uniqueness has appeared several times in the literature~\cite{osborne2000Lasso,fuchs2005recovery,wainwright2009sharp,candes2009near,tibshirani2013Lasso}. Since the active set $\mathcal{A}$ depends on the Lasso solution at $\mathbf{y},\mathbf{X},\lambda$, the condition that $\mathbf{X}_{\mathcal{A}}$ has full rank is somewhat circular. Because of this, more natural conditions are assumed which imply $\operatorname{null}(\mathbf{X}_{\mathcal{A}}) = \{\mathbf{0}\}$, e.g., the general position property~\cite{rosset2004following,donoho2006most,dossal2012necessary,tibshirani2013Lasso} and the matrix $\mathbf{X}$ being drawn from a continuous probability distribution~\cite{tibshirani2013Lasso}. We say that a matrix $\mathbf{X}\in\mathbb{R}^{n\times d}$ has columns in \emph{general position} if no $k$-dimensional subspace in $\mathbb{R}^n$, for $k < \min\{n,d\}$, contains more than $k+1$ elements from the set $\{\pm \mathbf{X}_1,\dots, \pm\mathbf{X}_d\}$, excluding antipodal pairs. It is known~\cite{rosset2004following,donoho2006most,dossal2012necessary,tibshirani2013Lasso} that this is enough to guarantee the uniqueness of the Lasso solution.
\begin{fact}[{\cite[Lemma~3]{tibshirani2013Lasso}}]\label{fact:general_position}
    If the columns of $\mathbf{X}\in\mathbb{R}^{n\times d}$ are in general position, then the Lasso solution is unique for any $\mathbf{y}\in\mathbb{R}^n$ and $\lambda>0$ and is given by {\rm \cref{eq:lasso_solution_unique}}.
\end{fact}
Another sufficient condition was given by Tibshirani~\cite{tibshirani2013Lasso}, who proved the following.
\begin{fact}[{\cite[Lemma~4]{tibshirani2013Lasso}}]\label{fact:continuous_distribution}
    If the entries of $\mathbf{X}\in\mathbb{R}^{n\times d}$ are drawn from a continuous probability distribution on $\mathbb{R}^{nd}$, then the Lasso solution is unique for any $\mathbf{y}\in\mathbb{R}^n$ and $\lambda>0$ and is given by {\rm \cref{eq:lasso_solution_unique}} with probability $1$. 
\end{fact}
It is well known that, when $\operatorname{null}(\mathbf{X}_{\mathcal{A}}) = \{\mathbf{0}\}$, the KKT conditions from \cref{fact:KKT1} imply that the regularisation path $\mathcal{P} \triangleq \{\hat{\bfbeta}(\lambda)\in\mathbb{R}^d:\lambda > 0\}$ comprising all the solutions for $\lambda > 0$ is continuous piecewise linear, which was first proven by Efron, Hastie, Johnstone, and Tibshirani~\cite{efron2004least}, cf.\ \cref{eqPath}. 
This fact was later extended to the case when $\operatorname{rank}(\mathbf{X}_{\mathcal{A}}) < |\mathcal{A}|$ by Tibshirani~\cite{tibshirani2013Lasso}, who proved that the path of solutions with $\mathbf{b} = \mathbf{0}$ in \cref{eq:lasso_solution} is continuous and piecewise linear. 
\begin{fact}[{\cite[Appendix~A]{tibshirani2013Lasso}}]
    The regularisation path $\mathcal{P} = \{\hat{\bfbeta}(\lambda)\in\mathbb{R}^d:\lambda > 0\}$ formed by the Lasso solution
    \begin{align*}
        \hat{\bfbeta}_{\mathcal{I}}(\lambda) = \mathbf{0} \quad\text{and}\quad \hat{\bfbeta}_{\mathcal{A}}(\lambda) = \mathbf{X}_{\mathcal{A}}^+(\mathbf{y} - \lambda(\mathbf{X}_{\mathcal{A}}^+)^\top\bfeta_{\mathcal{A}}),
    \end{align*}
    with $\mathcal{A} = \{ i\in[d]:  |\mathbf{X}_i^\top(\mathbf{y}-\mathbf{X}\hat{\bfbeta}(\lambda))| = \lambda\}$ and $\mathcal{I} =[d]\setminus\mathcal{A}$, is continuous piecewise linear.
\end{fact}
%
We call \emph{kink} a point where the direction $\partial\hat{\bfbeta}(\lambda)/\partial\lambda$ changes. The path $\{\hat{\bfbeta}(\lambda): \lambda_{t+1} < \lambda \leq \lambda_t\}$ between two consecutive kinks $\lambda_{t+1}$ and $\lambda_t$ thus defines a linear segment. An important consequence of the above result is that any continuous piecewise linear regularisation path has at most $3^d$ linear segments, since two different linear segments must have different equicorrelation signs $\bfeta\in\{-1,0,1\}^d$. This result was improved slightly to $(3^d+1)/2$ iterations by Mairal and Yu~\cite{mairal2012complexity}, who also exhibited a worst-case Lasso problem with exactly $(3^d+1)/2$ number of linear segments.

\subsection{The LARS algorithm}

The Lasso solution path $\{\hat{\bfbeta}(\lambda):\lambda>0\}$ can be computed by the Least Angle Regression (LARS) algorithm\footnote{By LARS algorithm we mean the version of the algorithm that computes the Lasso path and not the version that performs a kind of forward variable selection.} which was proposed and named by Efron \textit{et al.}~\cite{efron2004least}, although similar ideas have already appeared in the work of Osborne \textit{et al.}~\cite{osborne2000Lasso,osborne2000new}. Since its introduction, the LARS algorithm have been improved and generalised~\cite{rosset2007piecewise,mairal2012complexity,tibshirani2013Lasso,sjostrand2018spasm}. In compressed sensing literature, this algorithm is better known as the Homotopy method~\cite{foucart2013invitation}.

Starting at $\lambda = \infty$, where the Lasso solution is trivially $\mathbf{0}\in\mathbb{R}^d$, the LARS algorithm iteratively computes all the kinks in the Lasso solution path by decreasing the parameter $\lambda$ and checking for points where the path’s linear trajectory changes. This is done by making sure that the KKT conditions from \cref{fact:KKT1} are always satisfied. When a kink is reached, a coordinate from the solution $\hat{\bfbeta}(\lambda)$ will go from non-zero (being in the active set $\mathcal{A}$) to zero (join the inactive set $\mathcal{I}$), or vice-versa, i.e., from zero (leave $\mathcal{I}$) to non-zero (into $\mathcal{A}$). We shall now describe the LARS algorithm in more details, following~\cite{tibshirani2013Lasso}. \cref{alg:classicalLasso} summarises the LARS algorithm.

The LARS algorithm starts, without loss of generality, at the first kink $\lambda_0 = \|\mathbf{X}^\top\mathbf{y}\|_\infty$, since $\hat{\bfbeta}(\lambda) = \mathbf{0}$ for $\lambda \geq \lambda_0$ (plug $\hat{\bfbeta}(\lambda_0) = \mathbf{0}$ in \cref{eq:KKT1}). The corresponding active and inactive sets are $\mathcal{A} = \operatorname{argmax}_{j\in[d]}|\mathbf{X}^\top_j \mathbf{y}|$ and $\mathcal{I} = [d]\setminus\mathcal{A}$, respectively, while the equicorrelation sign is $\bfeta_{\mathcal{A}} = \operatorname{sign}(\mathbf{X}^\top_{\mathcal{A}}\mathbf{y})$. Since every iteration after the initial setup is identical, assume by induction that the LARS algorithm has computed the first $t+1$ kinks $(\lambda_0,\hat{\bfbeta}(\lambda_0)),\dots,(\lambda_t,\hat{\bfbeta}(\lambda_t))$. At the beginning of the $t$-th iteration, the regularisation parameter is $\lambda = \lambda_t$, the Lasso solution is $\hat{\bfbeta}(\lambda_t)$, and the active and inactive sets are $\mathcal{A}$ and $\mathcal{I}$. From the previous solution $\hat{\bfbeta}(\lambda_t)$ we obtain the equicorrelation signs 
\begin{align*}
    \bfeta_{\mathcal{A}} = \operatorname{sign}\big(\mathbf{X}^\top_{\mathcal{A}}(\mathbf{y} - \mathbf{X}_{\mathcal{A}}\hat{\bfbeta}_{\mathcal{A}}(\lambda_t))\big).
\end{align*}
According to \cref{eq:lasso_solution}, for $\lambda \leq \lambda_t$ sufficiently close to $\lambda_t$, the Lasso solution is\footnote{From now on, we shall assume $\mathbf{b} = \mathbf{0}$ in order to guarantee the continuity of the Lasso path.}
\begin{align}\label{eq:lasso_solution2}
    \hat{\bfbeta}_{\mathcal{I}}(\lambda) = \mathbf{0} \quad\text{and}\quad \hat{\bfbeta}_{\mathcal{A}}(\lambda) = \mathbf{X}_{\mathcal{A}}^+(\mathbf{y} - \lambda(\mathbf{X}_{\mathcal{A}}^+)^\top\bfeta_{\mathcal{A}}) = \bfmu - \lambda\bftheta,
\end{align}
where $\bfmu = \mathbf{X}_{\mathcal{A}}^+ \mathbf{y}$ and $\bftheta = (\mathbf{X}_{\mathcal{A}}^\top \mathbf{X}_{\mathcal{A}})^+ \bfeta_{\mathcal{A}}$. We now decrease $\lambda$ while keeping \cref{eq:lasso_solution2}. As $\lambda$ decreases, two important checks must be made. First, we check when (i.e., we compute the next value of $\lambda$ at which) a variable $i_{t+1}^{\rm join}$ in $\mathcal{I}$ should join the active set $\mathcal{A}$ because it has attain the maximum correlation $|\mathbf{X}_{i_{t+1}^{\rm join}}^\top (\mathbf{y} - \mathbf{X}\hat{\bfbeta})| = \lambda$. We call this event the next \emph{joining time} $\lambda_{t+1}^{\rm join}$. Second, we check when a variable $i_{t+1}^{\rm cross}$ in the active set $\mathcal{A}$ should leave $\mathcal{A}$ and join $\mathcal{I}$ because $\hat{\beta}_{i_{t+1}^{\rm cross}}$ crossed through zero. We call this event the next \emph{crossing time} $\lambda_{t+1}^{\rm cross}$.

For the first check, for each $i\in\mathcal{I}$, consider the equation
\begin{align*}
    \mathbf{X}_{i}^\top (\mathbf{y} - \mathbf{X}_{\mathcal{A}}\hat{\bfbeta}_{\mathcal{A}}(\lambda)) = \pm \lambda \implies \mathbf{X}_{i}^\top (\mathbf{y} - \mathbf{X}_{\mathcal{A}}(\bfmu - \lambda\bftheta)) = \pm \lambda.
\end{align*}
The solution to the above equation for $\lambda$, interpreted as the joining time of the $i$-th variable, is
\begin{equation} \label{eqSolution}
    \Lambda_i^{\rm join} (\mathcal A) \triangleq \frac{\mathbf{X}_i^\top (\mathbf{y} - \mathbf{X}_{\mathcal{A}}\bfmu)}{\pm 1 - \mathbf{X}_i^\top \mathbf{X}_{\mathcal{A}}\bftheta} = \frac{\mathbf{X}_i^\top (\mathbf{I} - \mathbf{X}_{\mathcal{A}} \mathbf{X}_{\mathcal{A}}^+)\mathbf{y}}{\pm 1 - \mathbf{X}_i^\top (\mathbf{X}_{\mathcal{A}}^\top)^+ \bfeta_{\mathcal{A}}},
\end{equation}
where we used that $\mathbf{X}_{\mathcal{A}}\mathbf{X}_{\mathcal{A}}^+(\mathbf{X}_{\mathcal{A}}^\top)^+ = (\mathbf{X}_{\mathcal{A}}^\top)^+$. Only one of the possibilities $+1$ or $-1$ will lead to a value in the necessary interval $[0,\lambda_t]$, and the following expressions always consider that solution. Therefore, the next joining time $\lambda_{t+1}^{\rm join}$ and coordinate $i_{t+1}^{\rm join}$ are
\begin{align*}
    \lambda_{t+1}^{\rm join} = \max_{i\in\mathcal{I}}\{\Lambda_i^{\rm join} (\mathcal A) \} \quad\text{and}\quad i_{t+1}^{\rm join} = \operatorname*{argmax}_{i\in\mathcal{I}}\{\Lambda_i^{\rm join} (\mathcal A) \}.
\end{align*}
For the second check, for each $i\in\mathcal{A}$, we solve the equation
\begin{align*}
    \hat{\beta}_i(\lambda) = 0 \implies \mu_i - \lambda \theta_i = 0.
\end{align*}
The crossing time of the $i$-th variable, which is the solution to the above equation restricted to values $\lambda\leq \lambda_t$, is thus defined as
\begin{align*}
    \Lambda_i^{\rm cross} (\mathcal A) \triangleq \frac{\mu_i}{\theta_i}\cdot\mathbf{1}\left[\frac{\mu_i}{\theta_i} \leq \lambda_t\right] = \frac{[\mathbf{X}_{\mathcal{A}}^+ \mathbf{y}]_i}{[(\mathbf{X}_{\mathcal{A}}^\top \mathbf{X}_{\mathcal{A}})^+ \bfeta_{\mathcal{A}}]_i}\cdot\mathbf{1}\left[\frac{[\mathbf{X}_{\mathcal{A}}^+ \mathbf{y}]_i}{[(\mathbf{X}_{\mathcal{A}}^\top \mathbf{X}_{\mathcal{A}})^+ \bfeta_{\mathcal{A}}]_i} \leq \lambda_t\right].
\end{align*}
Therefore, the next crossing time $\lambda_{t+1}^{\rm join}$ and coordinate $i_{t+1}^{\rm join}$ are
\begin{align*}
    \lambda_{t+1}^{\rm cross} = \max_{i\in\mathcal{A}}\{\Lambda_i^{\rm cross} (\mathcal A)\} \quad\text{and}\quad i_{t+1}^{\rm cross} = \operatorname*{argmax}_{i\in\mathcal{A}}\{\Lambda_i^{\rm cross} (\mathcal A)\}.
\end{align*}
Finally, the next kink of the regularisation path is $(\lambda_{t+1},\hat{\bfbeta}(\lambda_{t+1}))$ where $\lambda_{t+1} = \max\{\lambda_{t+1}^{\rm join}, \lambda_{t+1}^{\rm cross}\}$ and $\hat{\bfbeta}_{\mathcal{I}}(\lambda_{t+1}) = \mathbf{0}$ and $\hat{\bfbeta}_{\mathcal{A}}(\lambda_{t+1}) = \bfmu - \lambda_{t+1}\bftheta$. If $\lambda_{t+1}^{\rm join} > \lambda_{t+1}^{\rm cross}$, then the coordinate $i_{t+1}^{\rm join}$ should join the active set $\mathcal{A}$. Otherwise, the coordinate $i_{t+1}^{\rm cross}$ should leave the active set $\mathcal{A}$. This ends the $t$-th iteration of the LARS algorithm.

\begin{algorithm}[t]
\caption{Classical LARS algorithm for the pathwise Lasso}
\DontPrintSemicolon
\label{alg:classicalLasso}

\KwIn{Vector $\mathbf{y}\in\mathbb{R}^n$ and matrix $\mathbf{X}\in\mathbb{R}^{n\times d}$}

\KwOut{Exact regularisation path $\mathcal{P}$}

Initialise $\mathcal{A}=\operatorname{argmax}_{j\in[d]} |\mathbf{X}_j^\top \mathbf{y}|$, $\mathcal{I}=[d]\setminus\mathcal{A}$, $\lambda_0 = \|\mathbf{X}^\top\mathbf{y}\|_\infty$, $\hat{\bfbeta}(\lambda_0) = \mathbf{0}$, $t=0$

    
   \While{$\mathcal{I}\neq \varnothing$}
   {
        $\bfeta_{\mathcal{A}} \gets \sign\big(\mathbf{X}_{\mathcal{A}}^\top(\mathbf{y} - \mathbf{X}_{\mathcal{A}}\hat{\bfbeta}_{\mathcal{A}}(\lambda_{t}))\big)$

        $\bfmu \gets \mathbf{X}_{\mathcal{A}}^+ \mathbf{y}$ \tcp*{$\bfmu \in \mathbb{R}^{|\mathcal{A}|}$}
        
        $\bftheta \gets (\mathbf{X}_{\mathcal{A}}^\top \mathbf{X}_{\mathcal{A}})^+ \bfeta_{\mathcal{A}}$ \tcp*{$\bftheta \in \mathbb{R}^{|\mathcal{A}|}$}

        $\Lambda_i^{\rm cross} \gets \frac{\mu_i}{\theta_i} \cdot \mathbf{1}\left[\frac{\mu_i}{\theta_i} \leq \lambda_t\right]\quad \forall i\in\mathcal{A}$
        
        $\Lambda_i^{\rm join} \gets \frac{\mathbf{X}_i^\top (\mathbf{y}-\mathbf{X}_{\mathcal{A}}\bfmu)}{\pm 1 - \mathbf{X}_i^\top \mathbf{X}_{\mathcal{A}}\bftheta}\quad \forall i \in \mathcal I$ 
   
        $i_{t+1}^{\rm cross}\gets \operatorname{argmax}_{i\in \mathcal{A}} \{\Lambda_i^{\rm cross}\}$ and $\lambda_{t+1}^{\rm cross}\gets \max_{i\in \mathcal{A}} \{\Lambda_i^{\rm cross}\}$
        
        $i_{t+1}^{\rm join} \gets \operatorname{argmax}_{i\in \mathcal{I}}  \{\Lambda_i^{\rm join} \}$ and $\lambda_{t+1}^{\rm join} \gets \max_{i\in \mathcal{I}}  \{\Lambda_i^{\rm join}\}$

        $\lambda_{t+1} \gets \max\{\lambda_{t+1}^{\rm join}, \lambda_{t+1}^{\rm cross} \}$\;
        
        $\hat{\bfbeta}_{\mathcal{A}}(\lambda_{t+1}) \gets  \bfmu - \lambda_{t+1}\bftheta$\;
        
        \If{$\lambda_{t+1} = \lambda_{t+1}^{\rm cross}$}
        {        

            Move $i_{t+1}^{\rm cross}$ from $\mathcal{A}$ to $\mathcal{I}$
            
        }
        \Else
        {

            Move $i_{t+1}^{\rm join}$ from $\mathcal{I}$ to $\mathcal{A}$
        }

        $t\gets t+1$
    
  }
\Return coefficients $[(\lambda_0,\hat{\bfbeta}(\lambda_0)),(\lambda_1,\hat{\bfbeta}(\lambda_1)),\dots]$
\end{algorithm}

The time complexity of each iteration in the LARS algorithm is given by the following result.
\begin{fact}
    The time complexity per iteration of the {\rm LARS} algorithm is $O(nd + |\mathcal{A}|^2)$, where $\mathcal{A}$ is the active set of the corresponding iteration.
\end{fact}
\begin{proof}
    Let us start with the computation of $\mathbf{X}_{\mathcal{A}}^+$. Simply computing $\mathbf{X}_{\mathcal{A}}^+$ would require $O(n|\mathcal{A}|^2+|\mathcal{A}|^3)$ time. However, the time complexity can be reduced to $O(n|\mathcal{A}| + |\mathcal{A}|^2)$ by using the Sherman–Morrison formula applied to the Moore-Penrose inverse. More specifically, suppose we have previously computed $\mathbf{X}_{\mathcal{A}}^+$ and the index $i_{t+1}^{\rm join}$ is then added to the active set $\mathcal{A}$ to obtain the new active set $\mathcal{A}' = \mathcal{A}\cup\{i_{t+1}^{\rm join}\}$. Without lost of generality, write $\mathcal{A} = \{j_1,\dots,j_k\}$ and $\mathcal{A}' = \{j_1,\dots,j_k,i_{t+1}^{\rm join}\}$, where $k=|\mathcal{A}|$. Let $\mathbf{A} \in \mathbb{R}^{n\times (|\mathcal{A}| + 1)}$ be the matrix $\mathbf{X}_{\mathcal{A}}$ augmented with a zero column at the position 
    corresponding to $i_{t+1}^{\rm join}$, i.e., $\mathbf{A}_\ell = \mathbf{X}_{j_\ell}$ if $\ell\in [k]$, and $\mathbf{A}_\ell = \mathbf{0}$ if $\ell = k+1$. 
    The augmented matrix is written using a rank-$1$ update as
    \begin{align*}
    \mathbf{X}_{\mathcal{A}'} = \mathbf{A} + \mathbf{X}_{i^{\rm join}_{t+1}}\mathbf{e}_{k+1}^\top,
    \end{align*}
    where $\mathbf{e}_\ell\in\{0,1\}^{|\mathcal{A}|+1}$ is the column vector with $1$ in position $\ell$ and $0$ elsewhere. Therefore, by using the Sherman–Morrison formula applied to the Moore-Penrose inverse (see~\cite{meyer1973generalized}), the inverse $\mathbf{X}_{\mathcal{A}'}^+$ can be computed from the matrix $\mathbf{X}_{\mathcal{A}}^+$ and the vectors $\mathbf{X}_{i^{\rm join}_{t+1}}$ and $\mathbf{e}_{k+1}$ using simple matrix and vector multiplication, which requires $O(n|\mathcal{A}| + |\mathcal{A}|^2)$ time. Regarding the other quantities,
    \begin{itemize}
        \item Computing $\bfeta_{\mathcal{A}}$, $\bfmu$, and $\bftheta$ requires $O(n|\mathcal{A}| + |\mathcal{A}|^2)$ time (the updating argument can be used as well for $(\mathbf{X}_{\mathcal{A}}^\top \mathbf{X}_{\mathcal{A}})^+$);

        \item Finding $i_{t+1}^{\rm cross}$ and $\lambda_{t+1}^{\rm cross}$ requires $O(|\mathcal{A}|)$ time;

        \item Finding $i_{t+1}^{\rm join}$ and $\lambda_{t+1}^{\rm join}$ requires $O(n|\mathcal{I}|)$ time (the computations $\mathbf{y} -\mathbf{X}_{\mathcal{A}}\bfmu$ and $\mathbf{X}_{\mathcal{A}}\bftheta$ must be performed just once);

        \item Updating $\hat{\bfbeta}_{\mathcal{A}}(\lambda_{t+1})$ requires $O(|\mathcal{A}|)$ time.

    \end{itemize}
    In total, a single iteration requires $O(n|\mathcal{I}| + n|\mathcal{A}| + |\mathcal{A}|^2) = O(nd + |\mathcal{A}|^2)$ time. 
\end{proof}

Regarding the number of iterations taken by the LARS algorithm, heuristically, it is on average $O(n)$~\cite{rosset2007piecewise}. This can be understood as follows. If $n>d$, it would take $O(d)$ steps to add all the variables. If $n < d$, then we only add at most $n$ variables in the case of a unique solution, since $|\mathcal{A}| \leq \min\{n,d\}$. Dropping variables is rare as $\lambda$ is successfully decreased, usually $O(1)$ times. In the worst case, though, the number of iterations can be at most $(3^d+1)/2$ as previously mentioned.

\section{Quantum algorithms}\label{sect:quantum_algorithm}

In this section, we introduce our quantum algorithms based on the LARS algorithm. In \cref{sec:simple_quantum}, we propose our simple quantum LARS algorithm that exactly computes the pathwise Lasso, while in \cref{sec:approximate_quantum} we improve upon this simple algorithm and propose the approximate quantum LARS algorithm. In \cref{sec:approximate_classical}, we dequantise the approximate LARS algorithm using sampling techniques. Finally, in \cref{sec:examples}, we analyse the algorithms' complexity for the case when $\mathbf{X}$ is a random matrix from the $\mathbf{\Sigma}$-Gaussian ensemble.

\subsection{Simple quantum LARS algorithm}
\label{sec:simple_quantum}

As a warm-up, we propose a simple quantum LARS algorithm for the Lasso path by quantising the search step of $i_{t+1}^{\rm join}$ and $\lambda_{t+1}^{\rm join}$ using the quantum minimum-finding subroutine from D\"urr and H\o{}yer~\cite{durr1996quantum} (\cref{quantum_minimum_finding}). For this, we input the vectors $\mathbf{y} - \mathbf{X}_{\mathcal{A}}\bfmu$ and $\mathbf{X}_{\mathcal{A}}\bftheta$ into KP-tree data structures accessed by QRAMs. Moreover, as described in \cref{sec:computational_model}, we assume the matrix $\mathbf{X}$ is stored in a QROM. The final result is an improvement in the runtime over the usual LARS algorithm to $\widetilde{O}(n\sqrt{\smash[b]{|\mathcal{I}|}} + n|\mathcal{A}| + |\mathcal{A}|^2)$ per iteration.

\begin{algorithm}[ht]
\caption{Simple quantum LARS algorithm for the pathwise Lasso}
\DontPrintSemicolon
\label{alg:warmup_classical_quantum}

\KwIn{$T\in\mathbb{N}$, $\delta\in(0,1)$, $\mathbf{y}\in\mathbb{R}^n$, and $\mathbf{X}\in\mathbb{R}^{n\times d}$}

\KwOut{Exact regularisation path $\mathcal{P}$ with $T$ kinks with probability $\geq 1-\delta$}

Initialise $\mathcal{A}=\operatorname{argmax}_{j\in[d]} |\mathbf{X}_j^\top \mathbf{y}|$, $\mathcal{I}=[d]\setminus\mathcal{A}$, $\lambda_0 = \|\mathbf{X}^\top\mathbf{y}\|_\infty$, $\hat{\bfbeta}(\lambda_0) = \mathbf{0}$, $t=0$

    
   \While{$\mathcal{I}\neq \varnothing$, $t \leq T$}
   {
        $\bfeta_{\mathcal{A}} \gets \sign\big(\mathbf{X}_{\mathcal{A}}^\top(\mathbf{y} - \mathbf{X}_{\mathcal{A}}\hat{\bfbeta}_{\mathcal{A}}(\lambda_{t}))\big)$

        $\bfmu \gets \mathbf{X}_{\mathcal{A}}^+ \mathbf{y}$ \tcp*{$\bfmu \in \mathbb{R}^{|\mathcal{A}|}$}
        
        $\bftheta \gets (\mathbf{X}_{\mathcal{A}}^\top \mathbf{X}_{\mathcal{A}})^+ \bfeta_{\mathcal{A}}$ \tcp*{$\bftheta \in \mathbb{R}^{|\mathcal{A}|}$}

        Compute $\mathbf{y} - \mathbf{X}_{\mathcal{A}}\bfmu$ and $\mathbf{X}_{\mathcal{A}}\bftheta$ classically and input them into QRAMs

        Define $\Lambda_i^{\rm join} \triangleq \frac{\mathbf{X}_i^\top (\mathbf{y}-\mathbf{X}_{\mathcal{A}}\bfmu)}{\pm 1 - \mathbf{X}_i^\top \mathbf{X}_{\mathcal{A}}\bftheta}$ and $\Lambda_i^{\rm cross} \triangleq \frac{\mu_i}{\theta_i} \cdot \mathbf{1}\left[\frac{\mu_i}{\theta_i} \leq \lambda_t\right]$
   
        $i_{t+1}^{\rm cross}\gets \operatorname{argmax}_{i\in \mathcal{A}} \{\Lambda_i^{\rm cross}\}$ and $\lambda_{t+1}^{\rm cross}\gets \max_{i\in \mathcal{A}} \{\Lambda_i^{\rm cross}\}$
        
        $i_{t+1}^{\rm join} \gets \operatorname{argmax}_{i\in \mathcal{I}}  \{\Lambda_i^{\rm join} \}$ and $\lambda_{t+1}^{\rm join} \gets \max_{i\in \mathcal{I}}  \{\Lambda_i^{\rm join} \}$ with failure probability $\frac{\delta}{T}$ (\cref{quantum_minimum_finding})

        $\lambda_{t+1} \gets \max\{\lambda_{t+1}^{\rm join}, \lambda_{t+1}^{\rm cross} \}$\;
        
        $\hat{\bfbeta}_{\mathcal{A}}(\lambda_{t+1}) \gets  \bfmu - \lambda_{t+1}\bftheta$\;
        
        \If{$\lambda_{t+1} = \lambda_{t+1}^{\rm cross}$}
        {        

            Move $i_{t+1}^{\rm cross}$ from $\mathcal{A}$ to $\mathcal{I}$
            
        }
        \Else
        {

            Move $i_{t+1}^{\rm join}$ from $\mathcal{I}$ to $\mathcal{A}$
        }

        $t\gets t+1$
    
  }
\Return coefficients $[(\lambda_0,\hat{\bfbeta}(\lambda_0)),(\lambda_1,\hat{\bfbeta}(\lambda_1)),\dots]$
\end{algorithm}

\begin{theorem}\label{thr:simple_quantum}
    Let $\mathbf{X}\in\mathbb{R}^{n\times d}$ and $\mathbf{y}\in\mathbb{R}^n$. Assume $\mathbf{X}$ is stored in a {\rm QROM} and we have access to {\rm QRAMs} of size $O(n)$. Let $\delta\in(0,1)$ and $T\in\mathbb{N}$. The simple quantum {\rm LARS} {\rm \cref{alg:warmup_classical_quantum}} outputs, with probability at least $1-\delta$, at most $T$ kinks of the optimal regularisation path in time 
    \begin{align*}
        O\big(n\sqrt{|\mathcal{I}|}\log(T/\delta)\poly\log(nd) + n|\mathcal{A}| + |\mathcal{A}|^2\big) 
    \end{align*}
    per iteration, where $\mathcal{A}$ and $\mathcal{I}$ are the active and inactive sets of the corresponding iteration.
\end{theorem}
\begin{proof}
    The time complexity of \cref{alg:warmup_classical_quantum} is basically the same as \cref{alg:classicalLasso}, the only difference being that uploading $\mathbf{y} - \mathbf{X}_{\mathcal{A}}\bfmu$ and $\mathbf{X}_{\mathcal{A}}\bftheta$ into KP-trees takes $O(n\log{n})$ time, and finding $i_{t+1}^{\rm join}$ and $\lambda_{t+1}^{\rm join}$ now requires $O\big(n\sqrt{|\mathcal{I}|}\log(T/\delta)\poly\log(nd)\big)$ time, since that accessing $\mathbf{X}_i$ and computing $\Lambda_i^{\rm join}(\mathcal A) = \frac{\mathbf{X}_i^\top (\mathbf{y}-\mathbf{X}_{\mathcal{A}}\bfmu)}{\pm 1 - \mathbf{X}_i^\top \mathbf{X}_{\mathcal{A}}\bftheta} $ requires $O(n\poly\log(nd))$ and $O(n)$ time, respectively. 
\end{proof}

\subsection{Approximate quantum LARS algorithm}
\label{sec:approximate_quantum}

We now improve our previous simple quantum LARS algorithm. The main idea is to \emph{approximately} compute the joining times $\Lambda_i^{\rm join}(\mathcal A)$ (see \cref{eqSolution}) within the quantum minimum-finding algorithm instead of exactly computing them. This will lead to a quadratic improvement on the $n$ dependence in the term $\widetilde{O}(n\sqrt{\smash[b]{|\mathcal{I}|}})$ to $\widetilde{O}(\sqrt{\smash[b]{n|\mathcal{I}|}})$. Such approximation, however, hinders our ability to exactly find the joining variables $i_{t+1}^{\rm join}$ and points $\lambda_{t+1}^{\rm join}$. We can now only obtain a joining point $\widetilde{\lambda}_{t+1}^{\rm join}$ which is $\epsilon$-close to the true joining point $\lambda_{t+1}^{\rm join}$. This imposes new complications on the correctness analysis of the LARS algorithm, since we can no longer guarantee that the KKT conditions are satisfied. In order to tackle this issue, we consider an approximate version of the KKT conditions. We note that a similar concept was already introduced in~\cite{mairal2012complexity}.

\begin{definition}\label{def:appr_kkt}
    Let $\epsilon\geq 0$. A vector $\widetilde{\bfbeta}\in\mathbb{R}^d$ satisfies the $\mathsf{KKT}_{\lambda}(\epsilon)$ condition if and only if, $\forall j\in[d]$,
    \begin{align*}
        \lambda(1 - \epsilon) \leq \mathbf{X}_j^\top(\mathbf{y} - \mathbf{X}\widetilde{\bfbeta})\sign(\widetilde{\beta}_j) \leq \lambda &\quad~\text{if}~\widetilde{\beta}_j \neq 0,\\
        |\mathbf{X}_j^\top(\mathbf{y} - \mathbf{X}\widetilde{\bfbeta})| \leq \lambda &\quad~\text{if}~\widetilde{\beta}_j = 0.
    \end{align*}
\end{definition}
The reason for introducing the above approximate version of the KKT conditions is that it leads to approximate Lasso solutions, as proven in the next lemma.
\begin{lemma}\label{lem:kkt_error}
    If a vector $\widetilde{\bfbeta}\in\mathbb{R}^d$ satisfies the $\mathsf{KKT}_{\lambda}(\epsilon)$ condition for $\epsilon\geq 0$, then $\widetilde{\bfbeta}$ minimises the Lasso cost function $\mathcal{L}^{(\lambda)}(\bfbeta)$ up to an error $\lambda\epsilon\|\widetilde{\bfbeta}\|_1$, i.e.,
    \begin{align*}
        \frac{1}{2}\|\mathbf{y} - \mathbf{X}\widetilde{\bfbeta}\|_2^2 + \lambda \|\widetilde{\bfbeta}\|_1 - \left(\operatorname*{min}_{\bfbeta\in\mathbb{R}^d} \frac{1}{2}\|\mathbf{y} - \mathbf{X}\bfbeta\|_2^2 + \lambda \|\bfbeta\|_1\right) \leq \lambda\epsilon  \|\widetilde{\bfbeta}\|_1.
    \end{align*}
\end{lemma}
\begin{proof}
    The primal problem  $\min_{\bfbeta\in\mathbb{R}^d, {\bf z}\in \mathbb R^n} \frac{1}{2}\|\mathbf{y} - {\bf z} \|_2^2 + \lambda \|\bfbeta\|_1 ~\text{s.t.}~{\bf z}=\mathbf{X} \bfbeta$, with dummy variable $\bf z$, leads to the dual problem
    \begin{align*}
        \max_{\bfkappa\in\mathbb{R}^n} - \frac{1}{2}\bfkappa^\top\bfkappa - \bfkappa^\top \mathbf{y} \quad\text{s.t.}~\|\mathbf{X}^\top\bfkappa\|_\infty \leq \lambda.
    \end{align*}
    It is known that, given a pair of feasible primal and dual variables $(\widetilde{\bfbeta}, \widetilde{\bfkappa})$, the difference
    \begin{align*}
        \delta_\lambda(\widetilde{\bfbeta}, \widetilde{\bfkappa}) \triangleq \left(\frac{1}{2}\|\mathbf{y} - \mathbf{X}\widetilde{\bfbeta}\|_2^2 + \lambda \|\widetilde{\bfbeta}\|_1\right) - \left(- \frac{1}{2}\widetilde{\bfkappa}^\top\widetilde{\bfkappa} - \widetilde{\bfkappa}^\top \mathbf{y}\right)
    \end{align*}
    is called a duality gap and provides an optimality guarantee (see, e.g.,~\cite[Section~4.3]{borwein2006convex}):
    \begin{align*}
        0 \leq \left(\frac{1}{2}\|\mathbf{y} - \mathbf{X}\widetilde{\bfbeta}\|_2^2 + \lambda \|\widetilde{\bfbeta}\|_1\right) - \left(\min_{\bfbeta\in\mathbb{R}^d} \frac{1}{2}\|\mathbf{y} - \mathbf{X}\bfbeta\|_2^2 + \lambda \|\bfbeta\|_1\right) \leq \delta_\lambda(\widetilde{\bfbeta},\widetilde{\bfkappa}).
    \end{align*}
    The vector $\widetilde{\bfkappa} = \mathbf{X}\widetilde{\bfbeta} - \mathbf{y}$ is feasible for the dual problem since $\widetilde{\bfbeta}$ satisfies the $\mathsf{KKT}_\lambda(\epsilon)$ condition ($\|\mathbf{X}^\top\widetilde{\bfkappa}\|_\infty = \|\mathbf{X}^\top(\mathbf{y} - \mathbf{X}\widetilde{\bfbeta})\|_\infty \leq \lambda$). We can thus compute the duality gap,
    \begin{align*}
        \delta_\lambda(\widetilde{\bfbeta}, \widetilde{\bfkappa}) &= \left(\frac{1}{2}\|\mathbf{y} - \mathbf{X}\widetilde{\bfbeta}\|_2^2 + \lambda \|\widetilde{\bfbeta}\|_1\right) - \left(- \frac{1}{2}\widetilde{\bfkappa}^\top\widetilde{\bfkappa} - \widetilde{\bfkappa}^\top \mathbf{y}\right)\\
        &= \frac{1}{2}\widetilde{\bfkappa}^\top \widetilde{\bfkappa} + \lambda\|\widetilde{\bfbeta}\|_1 + \frac{1}{2}\widetilde{\bfkappa}^\top\widetilde{\bfkappa} + \widetilde{\bfkappa}^\top \mathbf{y}\\
        &= \lambda\|\widetilde{\bfbeta}\|_1 + (\widetilde{\bfkappa} + \mathbf{y})^\top \widetilde{\bfkappa},
    \end{align*}
    but
    \begin{align*}
        (\widetilde{\bfkappa} + \mathbf{y})^\top \widetilde{\bfkappa} = \widetilde{\bfbeta}^\top \mathbf{X}^\top (\mathbf{X}\widetilde{\bfbeta} - \mathbf{y}) = \sum_{i=1}^d |\widetilde{\beta}_i| \mathbf{X}_i^\top (\mathbf{X}\widetilde{\bfbeta} - \mathbf{y})\operatorname{sign}(\widetilde{\beta}_i) \leq -\lambda(1-\epsilon)\|\widetilde{\bfbeta}\|_1,
    \end{align*}
    using that $\widetilde{\bfbeta}$ satisfies the $\mathsf{KKT}_\lambda(\epsilon)$ condition. Therefore $\delta_\lambda(\widetilde{\bfbeta}, \widetilde{\bfkappa}) \leq \lambda\epsilon \|\widetilde{\bfbeta}\|_1$. 
\end{proof}

\begin{remark}\label{remark1}
    Even though the error $\lambda\epsilon\|\widetilde{\bfbeta}\|_1$ might look unnatural, it can be relaxed to a relative error by observing that it implies
    \begin{align*}
        \mathcal{L}^{(\lambda)}(\widetilde{\bfbeta}) - \min_{\bfbeta\in\mathbb{R}^d}\mathcal{L}^{(\lambda)}(\bfbeta) \leq \epsilon\lambda\|\widetilde{\bfbeta}\|_1 \leq \epsilon \mathcal{L}^{(\lambda)}(\widetilde{\bfbeta}) &\iff (1-\epsilon)\big(\mathcal{L}^{(\lambda)}(\widetilde{\bfbeta}) - \min_{\bfbeta\in\mathbb{R}^d}\mathcal{L}^{(\lambda)}(\bfbeta)\big) \leq \epsilon\min_{\bfbeta\in\mathbb{R}^d} \mathcal{L}^{(\lambda)}(\bfbeta) \\
        &\iff \mathcal{L}^{(\lambda)}(\widetilde{\bfbeta}) - \min_{\bfbeta\in\mathbb{R}^d}\mathcal{L}^{(\lambda)}(\bfbeta) \leq \frac{\epsilon}{1-\epsilon}\min_{\bfbeta\in\mathbb{R}^d} \mathcal{L}^{(\lambda)}(\bfbeta).
    \end{align*}
\end{remark}

We now present our approximate quantum LARS \cref{alg:classical_quantum} for the pathwise Lasso. Its main idea is quite simple: we improve the search of the joining variable $i_{t+1}^{\rm join}$ and time $\lambda_{t+1}^{\rm join}$ over $i\in\mathcal{I}$ by approximately computing the joining times $\Lambda_i^{\rm join}(\mathcal A) = \frac{\mathbf{X}_i^\top (\mathbf{y}-\mathbf{X}_{\mathcal{A}}\bfmu)}{\pm 1 - \mathbf{X}_i^\top \mathbf{X}_{\mathcal{A}}\bftheta}$ and using the approximate quantum minimum-finding subroutine from Chen and de Wolf~\cite{chen2021quantum} (\cref{fact:minimum_finding}). More specifically, the joining variable $\widetilde{i}_{t+1}^{\rm join}$ is obtained to be any variable $j\in\mathcal{I}$ that maximises $\Lambda_j^{\rm join}(\mathcal A)$ up to some small relative error $\epsilon$, i.e., at any iteration, \cref{alg:classical_quantum} randomly samples from the set
\begin{align}\label{eq:joining_time}
    \left\{j\in\mathcal{I}:\Lambda_j^{\rm join}(\mathcal A) \geq \Big(1-\frac{\epsilon}{1+\alpha_{\mathcal{A}}}\Big)\max_{i\in \mathcal{I}}  \{\Lambda_i^{\rm join}(\mathcal A) \}\right\},
\end{align}
where $\alpha_{\mathcal{A}} > 0$ is such that $\|\mathbf{X}_{\mathcal{A}}^+\mathbf{X}_{\mathcal{A}^c}\|_1 \leq \alpha_{\mathcal{A}}$. Since the corresponding joining time $\Lambda^{\rm join}_{\widetilde{i}_{t+1}^{\rm join}}(\mathcal A)$ can be smaller than the true joining time $\lambda_{t+1}^{\rm join} = \max_{i\in\mathcal{I}}\{\Lambda_i^{\rm join}(\mathcal A) \}$, we rescale it by a factor $\big(1-\frac{\epsilon}{1+\alpha_{\mathcal{A}}}\big)^{-1}$ and set $\widetilde{\lambda}_{t+1}^{\rm join} = \big(1-\frac{\epsilon}{1+\alpha_{\mathcal{A}}}\big)^{-1}\Lambda^{\rm join}_{\widetilde{i}_{t+1}^{\rm join}}(\mathcal A)$. This ensures that $\widetilde{\lambda}_{t+1}^{\rm join} \geq \lambda_{t+1}^{\rm join}$ and consequently that the current iteration stops before $\lambda_{t+1}^{\rm join}$, 
which guarantees that the approximate KKT conditions are satisfied as shown in \cref{thr:correctness} below.

\begin{algorithm}[ht]
\caption{Approximate quantum LARS algorithm for the pathwise Lasso}
\DontPrintSemicolon
\label{alg:classical_quantum}

\KwIn{$T\in\mathbb{N}$, $\delta\in(0,1)$, $\epsilon>0$, $\mathbf{y}\in\mathbb{R}^n$, $\mathbf{X}\in\mathbb{R}^{n\times d}$, and $\{\alpha_{\mathcal{A}} \geq \|\mathbf{X}_{\mathcal{A}}^+\mathbf{X}_{\mathcal{A}^c}\|_1\}_{\mathcal{A}\subseteq[d]:|\mathcal{A}|\leq T}$}

\KwOut{Regularisation path $\widetilde{\mathcal{P}}$ with $T$ kinks and error $\lambda \epsilon \|\widetilde{\bfbeta}(\lambda)\|_1$ with probability $\geq 1-\delta$}

Initialise $\mathcal{A}=\operatorname{argmax}_{j\in[d]} |\mathbf{X}_j^\top \mathbf{y}|$, $\mathcal{I}=[d]\setminus\mathcal{A}$, $\lambda_0 = \|\mathbf{X}^\top\mathbf{y}\|_\infty$, $\widetilde{\bfbeta}(\lambda_0) = 0$, $t=0$

    
   \While{$\mathcal{I}\neq \varnothing$, $t \leq T$}
   {
        $\bfeta_{\mathcal{A}} \gets \frac{1}{\lambda_{t}}\mathbf{X}_{\mathcal{A}}^\top(\mathbf{y} - \mathbf{X}_{\mathcal{A}}\widetilde{\bfbeta}_{\mathcal{A}}(\lambda_{t}))$

        $\bfmu \gets \mathbf{X}_{\mathcal{A}}^+ \mathbf{y}$ \tcp*{$\bfmu \in \mathbb{R}^{|\mathcal{A}|}$}
        
        $\bftheta \gets (\mathbf{X}_{\mathcal{A}}^\top \mathbf{X}_{\mathcal{A}})^+ \bfeta_{\mathcal{A}}$ \tcp*{$\bftheta \in \mathbb{R}^{|\mathcal{A}|}$}

        Classically compute $\mathbf{y} - \mathbf{X}_{\mathcal{A}}\bfmu$ and $\mathbf{X}_{\mathcal{A}}\bftheta$ and input them into QRAMs and KP-trees

        Define $\Lambda_i^{\rm join} \triangleq \frac{\mathbf{X}_i^\top (\mathbf{y}-\mathbf{X}_{\mathcal{A}}\bfmu)}{\pm 1 - \mathbf{X}_i^\top \mathbf{X}_{\mathcal{A}}\bftheta}$ and $\Lambda_i^{\rm cross} \triangleq \frac{\mu_i}{\theta_i} \cdot \mathbf{1}\left[\frac{\mu_i}{\theta_i} \leq \lambda_t\right]$ 
   
        $i_{t+1}^{\rm cross}\gets \operatorname{argmax}_{i\in \mathcal{A}} \{\Lambda_i^{\rm cross}\}$ and $\lambda_{t+1}^{\rm cross}\gets \max_{i\in \mathcal{A}}\{\Lambda_i^{\rm cross}\}$
        
        Obtain $\widetilde{i}_{t+1}^{\rm join} \in \big\{j\in\mathcal{I}: \Lambda_j^{\rm join} \geq \big(1-\frac{\epsilon}{1+\alpha_{\mathcal{A}}}\big)\max_{i\in \mathcal{I}}  \{\Lambda_i^{\rm join} \} \big\}$ with failure probability $\frac{\delta}{T}$ (\cref{fact:minimum_finding,lem:terms_estimation})
        
        $\widetilde{\lambda}_{t+1}^{\rm join} \gets  \big(1-\frac{\epsilon}{1+\alpha_{\mathcal{A}}}\big)^{-1}\Lambda_{\widetilde{i}_{t+1}^{\rm join}}^{\rm join}$\;
        
        $\lambda_{t+1} \gets \min\{\lambda_t,\max\{\lambda_{t+1}^{\rm cross}, \widetilde{\lambda}_{t+1}^{\rm join}\}\}$\;
        
        $\widetilde{\bfbeta}_{\mathcal{A}}(\lambda_{t+1}) \gets  \bfmu - \lambda_{t+1}\bftheta$\;
        
        \If{$\lambda_{t+1} = \lambda_{t+1}^{\rm cross}$}
        {        
            Move $i_{t+1}^{\rm cross}$ from $\mathcal{A}$ to $\mathcal{I}$
            
        }
        \Else
        {
            Move $\widetilde{i}_{t+1}^{\rm join}$ from $\mathcal{I}$ to $\mathcal{A}$
        }
        
        $t\gets t+1$
    
  }
\Return coefficients $[(\lambda_0,\widetilde{\bfbeta}(\lambda_0)),(\lambda_1,\widetilde{\bfbeta}(\lambda_1)),\dots]$
\end{algorithm}

We notice that, since the joining variable is randomly sampled from the set in \cref{eq:joining_time}, it might be possible that $\widetilde{\lambda}_{t+2}^{\rm join} > \widetilde{\lambda}_{t+1}^{\rm join}$, i.e., the joining time at some later iteration is greater than the joining time at some earlier iteration. While this is counter to the original LARS algorithm, where the regularisation parameter $\lambda$ always decreases, there is in principle no problem in allowing $\lambda$ to increase within a small interval as long as the approximate KKT conditions are satisfied. Nonetheless, in order to avoid redundant segments in the Lasso path, we set the next iteration's regularisation parameter $\lambda_{t+1}$ as the minimum between the current $\lambda_t$ and $\max\{\lambda_{t+1}^{\rm cross}, \widetilde{\lambda}_{t+1}^{\rm join}\}$. The solution path can thus stay ``stationary'' for a few iterations. Ideally, one could just move all the variables from the set in \cref{eq:joining_time} into $\mathcal{A}$ at once and the result would be the same. This is reminiscent to the approximate homotopy algorithm of Mairal and Yu~\cite{mairal2012complexity}.

We show next that the imprecision in finding $\lambda_{t+1}^{\rm join}$ leads to a path that satisfies the approximate KKT conditions from \cref{def:appr_kkt} and thus approximates the Lasso function up to a small error.
\begin{theorem}\label{thr:correctness}
    Let $T\in\mathbb{N}$, $\epsilon\geq 0$, $\mathbf{X}\in\mathbb{R}^{n\times d}$, and $\mathbf{y}\in\mathbb{R}^n$. For $\mathcal{A}\subseteq[d]$ of size $|\mathcal{A}|\leq T$, let $\alpha_{\mathcal{A}} > 0$ be such that $\|\mathbf{X}_{\mathcal{A}}^+\mathbf{X}_{\mathcal{A}^c}\|_1 \leq \alpha_{\mathcal{A}}$. Consider an approximate {\rm LARS} algorithm (e.g., {\rm \cref{alg:classical_quantum}}) that returns a solution path $\widetilde{\mathcal{P}} = \{\widetilde{\bfbeta}(\lambda)\in\mathbb{R}^d: \lambda > 0\}$ with at most $T$ kinks and wherein, at each iteration~$t$, the joining variable $\widetilde{i}_{t+1}^{\rm join}$ is taken from the set
    \begin{align*}
        \left\{j\in\mathcal{I}:  \Lambda_j^{\rm join}(\mathcal A) \geq \Big(1-\frac{\epsilon}{1+\alpha_{\mathcal{A}}}\Big)\max_{i\in \mathcal{I}}  \{\Lambda_i^{\rm join}(\mathcal A) \}\right\}
    \end{align*}
    and the corresponding joining time $\widetilde{\lambda}_{t+1}^{\rm join}$ is
    \begin{align*}
        \widetilde{\lambda}_{t+1}^{\rm join} = \Big(1-\frac{\epsilon}{1+\alpha_{\mathcal{A}}}\Big)^{-1}\Lambda^{\rm join}_{\widetilde{i}_{t+1}^{\rm join}}(\mathcal A),
    \end{align*}
    where $\mathcal{A}$ and $\mathcal{I}$ are the active and inactive sets at the corresponding iteration $t$ and
    \begin{align*}
        \Lambda_j^{\rm join}(\mathcal A) = \frac{\mathbf{X}_j^\top (\mathbf{y}-\mathbf{X}_{\mathcal{A}}\bfmu)}{\pm 1 - \mathbf{X}_j^\top \mathbf{X}_{\mathcal{A}}\bftheta},\quad \bfmu = \mathbf{X}_{\mathcal{A}}^+ \mathbf{y}, \quad \bftheta = (\mathbf{X}_{\mathcal{A}}^\top \mathbf{X}_{\mathcal{A}})^+ \bfeta_{\mathcal{A}}, \quad \bfeta_{\mathcal{A}} = \frac{1}{\lambda_{t}}\mathbf{X}_{\mathcal{A}}^\top(\mathbf{y} - \mathbf{X}_{\mathcal{A}}\widetilde{\bfbeta}_{\mathcal{A}}(\lambda_{t})).
    \end{align*}
    Then $\widetilde{\mathcal{P}}$ is a continuous and piecewise linear approximate regularisation path with error $\lambda\epsilon\|\widetilde{\bfbeta}(\lambda)\|_1$.
\end{theorem}
\begin{proof}
    We shall prove that the Lasso solution $\widetilde{\bfbeta}(\lambda)$ satisfies the $\mathsf{KKT}_{\lambda}(\epsilon)$ condition for all $\lambda>0$. According to \cref{lem:kkt_error}, $\widetilde{\mathcal{P}}$ is then an approximate regularisation path with error $\lambda\epsilon\|\widetilde{\bfbeta}(\lambda)\|_1$.
    
    The proof is by induction on the iteration loop $t$. The case $t=0$ is trivial as $\hat{\bfbeta}(\lambda) = \mathbf{0}$ for $\lambda \geq \lambda_0$. Assume then that the computed path through iteration $t-1$ satisfies $\mathsf{KKT}_{\lambda}(\epsilon)$ for all $\lambda \geq \lambda_t$. Consider the $t$-th iteration with active set $\mathcal{A} = \{i\in[d]:\widetilde{\beta}_i \neq 0\}$ and equicorrelation signs $\bfeta_{\mathcal{A}} = \frac{1}{\lambda_t}\mathbf{X}_{\mathcal{A}}^\top(\mathbf{y} - \mathbf{X}\widetilde{\bfbeta}(\lambda_t))$. The induction hypothesis implies that the current Lasso solution $\widetilde{\bfbeta}(\lambda_t)$ satisfies $\mathsf{KKT}_{\lambda_t}(\epsilon)$ at $\lambda_t$:
    \begin{align*}
        \lambda_t(1-\epsilon) \leq \mathbf{X}_j^\top(\mathbf{y} - \mathbf{X}\widetilde{\bfbeta}(\lambda_t))\sign(\widetilde{\beta}_j(\lambda_t)) \leq \lambda_t&\quad~\text{if}~\widetilde{\beta}_j \neq 0,\\
        |\mathbf{X}_j^\top(\mathbf{y} - \mathbf{X}\widetilde{\bfbeta}(\lambda_t))| \leq \lambda_t &\quad~\text{if}~\widetilde{\beta}_j = 0.
    \end{align*}
    Recall that for $\lambda \leq \lambda_t$ the Lasso solution is
    \begin{align*}
        \widetilde{\bfbeta}_{\mathcal{I}}(\lambda) = \mathbf{0} \quad\text{and}\quad \widetilde{\bfbeta}_{\mathcal{A}}(\lambda) = \mathbf{X}_{\mathcal{A}}^+(\mathbf{y} - \lambda(\mathbf{X}_{\mathcal{A}}^+)^\top\bfeta_{\mathcal{A}}).
    \end{align*}
    Thus
    \begin{align*}
        \mathbf{X}_{\mathcal{A}}^\top(\mathbf{y} - \mathbf{X}_{\mathcal{A}}\widetilde{\bfbeta}_{\mathcal{A}}(\lambda)) &= \mathbf{X}_{\mathcal{A}}^\top \mathbf{y} - \mathbf{X}_{\mathcal{A}}^\top \mathbf{X}_{\mathcal{A}}\mathbf{X}_{\mathcal{A}}^+ \mathbf{y} + \lambda\mathbf{X}_{\mathcal{A}}^\top \mathbf{X}_{\mathcal{A}}(\mathbf{X}_{\mathcal{A}}^\top\mathbf{X}_{\mathcal{A}})^+\bfeta_{\mathcal{A}} \\
        &= \lambda\mathbf{X}_{\mathcal{A}}^\top (\mathbf{X}_{\mathcal{A}}^\top)^+\bfeta_{\mathcal{A}} \tag{by $\mathbf{X}_{\mathcal{A}}^\top \mathbf{X}_{\mathcal{A}}\mathbf{X}^+_{\mathcal{A}} = \mathbf{X}_{\mathcal{A}}^\top$}\\
        &= \lambda \bfeta_{\mathcal{A}}, \tag{by $\bfeta_{\mathcal{A}}\in\operatorname{row}(\mathbf{X}_{\mathcal{A}})$}
    \end{align*}
    where we used the identity $\mathbf{A}^\top \mathbf{A}\mathbf{A}^+ = \mathbf{A}^\top$. Since 
    \begin{align*}
        \eta_j \operatorname{sign}(\widetilde{\beta}_j(\lambda)) = \eta_j \operatorname{sign}(\widetilde{\beta}_j(\lambda_t)) = \frac{1}{\lambda_t}\mathbf{X}_j^\top(\mathbf{y} - \mathbf{X}\widetilde{\bfbeta}(\lambda_t))\operatorname{sign}(\widetilde{\beta}_j(\lambda_t))
    \end{align*}
    for $\lambda$ sufficiently close to $\lambda_t$ and since $\widetilde{\bfbeta}_{\mathcal{A}}(\lambda_t)$ satisfies $\mathsf{KKT}_{\lambda_t}(\epsilon)$, we conclude that
    \begin{align*}
        \lambda(1-\epsilon) \leq \mathbf{X}_j^\top(\mathbf{y} - \mathbf{X}\widetilde{\bfbeta}(\lambda))\sign(\widetilde{\beta}_j(\lambda)) \leq \lambda \quad\text{for all}~ j\in\mathcal{A}
    \end{align*}
    as required. As $\lambda$ decreases, one of the following conditions must break: either $\eta_j \neq 0$ for all $j\in\mathcal{A}$ or $\|\mathbf{X}_{\mathcal{I}}^\top(\mathbf{y} - \mathbf{X}\widetilde{\bfbeta}(\lambda))\|_\infty \leq \lambda$. The first breaks at the next crossing time $\lambda_{t+1}^{\rm cross} = \max_{i\in\mathcal{A}}\{\Lambda_i^{\rm cross}(\mathcal{A})\}$, and the second breaks at the next joining time $\lambda^{\rm join}_{t+1} = \max_{i\in \mathcal{I}} \{\Lambda_i^{\rm join}(\mathcal A)\}$. Since we only decrease $\lambda$ to $\lambda_{t+1} = \min\{\lambda_t,\max\{{\lambda}_{t+1}^{\rm cross}, \widetilde{\lambda}_{t+1}^{\rm join}\}\}$, we have verified that $\widetilde{\bfbeta}(\lambda)$ satisfies $\mathsf{KKT}_{\lambda}(\epsilon)$ for $\lambda \geq \lambda_{t+1}$.

    We now prove that adding or deleting variables from the active set preserves the condition $\mathsf{KKT}_{\lambda_{t+1}}(\epsilon)$ at $\lambda_{t+1}$. More specifically, let $\mathcal{A}^\ast$ and $\bfeta^\ast_{\mathcal{A}^\ast}$ denote the active set and equicorrelation signs at the beginning of the $(t+1)$-th iteration.

    \emph{Case 1 (Deletion)}: Let us start with the case when a variable leaves the active set at $\lambda_{t+1} = \lambda_{t+1}^{\rm cross}$, i.e., $\mathcal{A}^\ast = \mathcal{A}\setminus\{i^{\rm cross}_{t+1}\}$ is formed by removing an element from $\mathcal{A}$. The Lasso solution before deletion with equicorrelation signs $\bfeta_{\mathcal{A}}$ is
    \begin{align*}
        \widetilde{\bfbeta}_{\mathcal{A}}(\lambda_{t+1}) = \begin{bmatrix*}
            \widetilde{\bfbeta}_{\mathcal{A}^\ast}(\lambda_{t+1})\\
            \widetilde{\beta}_{i_{t+1}^{\rm cross}}(\lambda_{t+1})
        \end{bmatrix*} = 
        \begin{bmatrix*}
            \big[\mathbf{X}_{\mathcal{A}}^+ (\mathbf{y} - \lambda_{t+1} (\mathbf{X}_{\mathcal{A}}^+)^\top\bfeta_{\mathcal{A}})\big]_{\mathcal{A}^\ast}\\
            0
        \end{bmatrix*}
    \end{align*}
    since the variable $i_{t+1}^{\rm cross}$ crosses through zero at $\lambda_{t+1} = \lambda_{t+1}^{\rm cross}$. On the other hand, the Lasso solution after deletion is
    \begin{align*}
        \widetilde{\bfbeta}^\ast_{\mathcal{A}}(\lambda_{t+1}) = \begin{bmatrix*}
            \widetilde{\bfbeta}^\ast_{\mathcal{A}^\ast}(\lambda_{t+1}) \\
            0
        \end{bmatrix*} = 
        \begin{bmatrix*}
            \mathbf{X}_{\mathcal{A}^\ast}^+ (\mathbf{y} - \lambda_{t+1} (\mathbf{X}_{\mathcal{A}^\ast}^+)^\top\bfeta^\ast_{\mathcal{A}^\ast})\\
            0
        \end{bmatrix*}, 
    \end{align*}
    where $\bfeta^\ast_{\mathcal{A}^\ast} = \frac{1}{\lambda_{t+1}}\mathbf{X}_{\mathcal{A}^\ast}^\top(\mathbf{y} - \mathbf{X}_{\mathcal{A}^\ast}\widetilde{\bfbeta}_{\mathcal{A}^\ast}(\lambda_{t+1}))$. Therefore
    \begin{align*}
        \mathbf{X}_{\mathcal{A}^\ast}^+ (\mathbf{y} - \lambda_{t+1} (\mathbf{X}_{\mathcal{A}^\ast}^+)^\top\bfeta^\ast_{\mathcal{A}^\ast}) &= \mathbf{X}_{\mathcal{A}^\ast}^+\mathbf{y} -  \mathbf{X}_{\mathcal{A}^\ast}^+(\mathbf{X}_{\mathcal{A}^\ast}^+)^\top\mathbf{X}_{\mathcal{A}^\ast}^\top\mathbf{y} + \mathbf{X}_{\mathcal{A}^\ast}^+(\mathbf{X}_{\mathcal{A}^\ast}^+)^\top\mathbf{X}_{\mathcal{A}^\ast}^\top\mathbf{X}_{\mathcal{A}^\ast}\widetilde{\bfbeta}_{\mathcal{A}^\ast}(\lambda_{t+1})\\
        &= \mathbf{X}_{\mathcal{A}^\ast}^+\mathbf{X}_{\mathcal{A}^\ast}\widetilde{\bfbeta}_{\mathcal{A}^\ast}(\lambda_{t+1}), \tag{by $\mathbf{X}_{\mathcal{A}^\ast}^+(\mathbf{X}_{\mathcal{A}^\ast}^+)^\top\mathbf{X}_{\mathcal{A}^\ast}^\top = \mathbf{X}_{\mathcal{A}^\ast}^+$}
    \end{align*}
    again using that $\mathbf{A}^+(\mathbf{A}^+)^\top \mathbf{A}^\top = \mathbf{A}^+$. The solution $\widetilde{\bfbeta}_{\mathcal{A}^\ast}(\lambda_{t+1})$ to the above equation must be
    \begin{align*}
        \widetilde{\bfbeta}_{\mathcal{A}^\ast}(\lambda_{t+1}) = \mathbf{X}_{\mathcal{A}^\ast}^+(\mathbf{y} - \lambda_{t+1} (\mathbf{X}_{\mathcal{A}^\ast}^+)^\top\bfeta^\ast_{\mathcal{A}^\ast}) + \mathbf{b} = \widetilde{\bfbeta}^\ast_{\mathcal{A}^\ast}(\lambda_{t+1}) + \mathbf{b},
    \end{align*}
    where $\mathbf{b}\in\operatorname{null}(\mathbf{X}_{\mathcal{A}^\ast})$. Since $\widetilde{\bfbeta}_{\mathcal{A}^\ast}(\lambda_{t+1})$ must have minimum $\ell_2$-norm, $\mathbf{b}=\mathbf{0}$ and so
    \begin{align*}
        \widetilde{\bfbeta}^\ast_{\mathcal{A}}(\lambda_{t+1}) = \begin{bmatrix*}
            \widetilde{\bfbeta}^\ast_{\mathcal{A}^\ast}(\lambda_{t+1}) \\
            0
        \end{bmatrix*} = 
        \begin{bmatrix*}
            \widetilde{\bfbeta}_{\mathcal{A}^\ast}(\lambda_{t+1}) \\
            0
        \end{bmatrix*},
    \end{align*}
    which implies that $\widetilde{\mathcal{P}}$ is continuous at $\lambda_{t+1}$. Finally,
    \begin{align*}
         \mathbf{X}_j^\top(\mathbf{y} - \mathbf{X}_{\mathcal{A}^\ast}\widetilde{\bfbeta}^\ast_{\mathcal{A}^\ast}(\lambda_{t+1})) = \mathbf{X}_j^\top(\mathbf{y} - \mathbf{X}_{\mathcal{A}}\widetilde{\bfbeta}_{\mathcal{A}}(\lambda_{t+1})) \quad \text{for all}~j\in[d]
    \end{align*}
    and so $\widetilde{\bfbeta}^\ast(\lambda_{t+1})$ satisfies the $\mathsf{KKT}_{\lambda_{t+1}}(\epsilon)$ condition.

    \emph{Case 2 (Insertion)}: Now we look at the case when a variable joins the active set at $\lambda_{t+1} = \min\{\lambda_t, \widetilde{\lambda}_{t+1}^{\rm join}\}$, i.e., $\mathcal{A}^\ast = \mathcal{A}\cup\{\widetilde{i}_{t+1}^{\rm join}\}$ is formed by adding an element to $\mathcal{A}$. The Lasso solution before insertion is
    \begin{align*}
        \widetilde{\bfbeta}_{\mathcal{A}^\ast}(\lambda_{t+1}) = \begin{bmatrix*}
            \widetilde{\bfbeta}_{\mathcal{A}}(\lambda_{t+1}) \\ 0
        \end{bmatrix*},
    \end{align*}
    while the Lasso solution after insertion with equicorrelation signs
    \begin{align*}
        \bfeta^\ast_{\mathcal{A}^\ast} = \frac{1}{\lambda_{t+1}}\mathbf{X}_{\mathcal{A}^\ast}^\top(\mathbf{y} - \mathbf{X}_{\mathcal{A}^\ast}\widetilde{\bfbeta}_{\mathcal{A}^\ast}(\lambda_{t+1})) = \frac{1}{\lambda_{t+1}}\mathbf{X}_{\mathcal{A}^\ast}^\top(\mathbf{y} - \mathbf{X}_{\mathcal{A}}\widetilde{\bfbeta}_{\mathcal{A}}(\lambda_{t+1}))
    \end{align*}
    is
    \begin{align*}
        \widetilde{\bfbeta}^\ast_{\mathcal{A}^\ast}(\lambda_{t+1}) &= 
        \mathbf{X}_{\mathcal{A}^\ast}^+ (\mathbf{y} - \lambda_{t+1} (\mathbf{X}_{\mathcal{A}^\ast}^+)^\top \bfeta^\ast_{\mathcal{A}^\ast}) \\
        &= \mathbf{X}_{\mathcal{A}^\ast}^+ \mathbf{y} - \mathbf{X}_{\mathcal{A}^\ast}^+(\mathbf{X}_{\mathcal{A}^\ast}^+)^\top\mathbf{X}_{\mathcal{A}^\ast}^\top\mathbf{y} + \mathbf{X}_{\mathcal{A}^\ast}^+(\mathbf{X}_{\mathcal{A}^\ast}^+)^\top\mathbf{X}_{\mathcal{A}^\ast}^\top\mathbf{X}_{\mathcal{A}^\ast}\widetilde{\bfbeta}_{\mathcal{A}^\ast}(\lambda_{t+1})\\
        &= \mathbf{X}_{\mathcal{A}^\ast}^+\mathbf{X}_{\mathcal{A}^\ast}\widetilde{\bfbeta}_{\mathcal{A}^\ast}(\lambda_{t+1}), \tag{by $\mathbf{X}_{\mathcal{A}^\ast}^+(\mathbf{X}_{\mathcal{A}^\ast}^+)^\top\mathbf{X}_{\mathcal{A}^\ast}^\top = \mathbf{X}_{\mathcal{A}^\ast}^+$}
    \end{align*}
    again using that $\mathbf{A}^+(\mathbf{A}^+)^\top \mathbf{A}^\top = \mathbf{A}^+$. The solution $\widetilde{\bfbeta}_{\mathcal{A}^\ast}(\lambda_{t+1})$ to the above equation must be
    \begin{align*}
        \widetilde{\bfbeta}_{\mathcal{A}^\ast}(\lambda_{t+1}) = \mathbf{X}_{\mathcal{A}^\ast}^+(\mathbf{y} - \lambda_{t+1} (\mathbf{X}_{\mathcal{A}^\ast}^+)^\top\bfeta^\ast_{\mathcal{A}^\ast}) + \mathbf{b} = \widetilde{\bfbeta}^\ast_{\mathcal{A}^\ast}(\lambda_{t+1}) + \mathbf{b},
    \end{align*}
    where $\mathbf{b}\in\operatorname{null}(\mathbf{X}_{\mathcal{A}^\ast})$. Since $\widetilde{\bfbeta}_{\mathcal{A}^\ast}(\lambda_{t+1})$ must have minimum $\ell_2$-norm, $\mathbf{b}=\mathbf{0}$ and so
    \begin{align*}
        \widetilde{\bfbeta}^\ast_{\mathcal{A}^\ast}(\lambda_{t+1}) = \widetilde{\bfbeta}_{\mathcal{A}^\ast}(\lambda_{t+1}) = \begin{bmatrix*}
            \widetilde{\bfbeta}_{\mathcal{A}}(\lambda_{t+1}) \\ 0
        \end{bmatrix*},
    \end{align*}
    which implies that $\widetilde{\mathcal{P}}$ is continuous at $\lambda_{t+1}$. Also, 
    \begin{align*}
        \mathbf{X}_j^\top(\mathbf{y} - \mathbf{X}_{\mathcal{A}^\ast}\widetilde{\bfbeta}^\ast_{\mathcal{A}^\ast}(\lambda_{t+1})) = \mathbf{X}_j^\top(\mathbf{y} - \mathbf{X}_{\mathcal{A}}\widetilde{\bfbeta}_{\mathcal{A}}(\lambda_{t+1}))\quad \text{for all}~j\in[d].
    \end{align*}
    It only remains to prove that the variable $\widetilde{i}_{t+1}^{\rm join}$ satisfies the $\mathsf{KKT}_{\lambda_{t+1}}(\epsilon)$ condition once it is added to $\mathcal{A}$. Indeed, first notice that the variable $\widetilde{i}_{t+1}^{\rm join}$ achieves correlation
    \begin{align*}
        \big|\mathbf{X}_{\widetilde{i}_{t+1}^{\rm join}}^\top(\mathbf{y} - \mathbf{X}_{\mathcal{A}}\widetilde{\bfbeta}_{\mathcal{A}}(\Lambda^{\rm join}_{\widetilde{i}_{t+1}^{\rm join}}))| = \Lambda^{\rm join}_{\widetilde{i}_{t+1}^{\rm join}}
    \end{align*}
    at the point $\Lambda^{\rm join}_{\widetilde{i}_{t+1}^{\rm join}} = \Lambda^{\rm join}_{\widetilde{i}_{t+1}^{\rm join}}(\mathcal{A})$. Define $\Delta_{t+1}^{\rm join} \triangleq \min\{\lambda_t,\widetilde{\lambda}^{\rm join}_{t+1}\} - \Lambda^{\rm join}_{\widetilde{i}_{t+1}^{\rm join}}$ and note that $\Delta_{t+1}^{\rm join} \leq \frac{\epsilon}{1+\alpha_{\mathcal{A}}}\widetilde{\lambda}^{\rm join}_{t+1}$ since $\lambda_t \geq \Lambda^{\rm join}_{\widetilde{i}_{t+1}^{\rm join}}$ and $\big(1-\frac{\epsilon}{1+\alpha_{\mathcal{A}}}\big)\widetilde{\lambda}_{t+1}^{\rm join} \leq \Lambda^{\rm join}_{\widetilde{i}_{t+1}^{\rm join}}$. Then, at the point $\lambda_{t+1} = \min\{\lambda_t,\widetilde{\lambda}^{\rm join}_{t+1}\}$,
    \begin{align*}
        \big|\mathbf{X}_{\widetilde{i}_{t+1}^{\rm join}}^\top(\mathbf{y} - \mathbf{X}_{\mathcal{A}}\widetilde{\bfbeta}_{\mathcal{A}}(\widetilde{\lambda}^{\rm join}_{t+1}))\big| &= \big|\mathbf{X}_{\widetilde{i}_{t+1}^{\rm join}}^\top(\mathbf{y} - \mathbf{X}_{\mathcal{A}}\widetilde{\bfbeta}_{\mathcal{A}}(\Lambda^{\rm join}_{\widetilde{i}_{t+1}^{\rm join}})) + \Delta_{t+1}^{\rm join}\mathbf{X}_{\widetilde{i}_{t+1}^{\rm join}}^\top (\mathbf{X}_{\mathcal{A}}^+)^\top \bfeta_{\mathcal{A}}\big| \\
        &\geq \big|\mathbf{X}_{\widetilde{i}_{t+1}^{\rm join}}^\top(\mathbf{y} - \mathbf{X}_{\mathcal{A}}\widetilde{\bfbeta}_{\mathcal{A}}(\Lambda^{\rm join}_{\widetilde{i}_{t+1}^{\rm join}}))\big| - \Delta_{t+1}^{\rm join}\big|\mathbf{X}_{\widetilde{i}_{t+1}^{\rm join}}^\top (\mathbf{X}_{\mathcal{A}}^+)^\top \bfeta_{\mathcal{A}}\big| \tag{Triangle inequality}\\
        &= \Lambda^{\rm join}_{\widetilde{i}_{t+1}^{\rm join}} - \Delta_{t+1}^{\rm join}\big|\mathbf{X}_{\widetilde{i}_{t+1}^{\rm join}}^\top(\mathbf{X}_{\mathcal{A}}^+)^\top  \bfeta_{\mathcal{A}}\big| \tag{$\big|\mathbf{X}_{\widetilde{i}_{t+1}^{\rm join}}^\top(\mathbf{y} - \mathbf{X}_{\mathcal{A}}\widetilde{\bfbeta}_{\mathcal{A}}(\Lambda^{\rm join}_{\widetilde{i}_{t+1}^{\rm join}}))\big| = \Lambda^{\rm join}_{\widetilde{i}_{t+1}^{\rm join}}$}\\
        &= \Lambda^{\rm join}_{\widetilde{i}_{t+1}^{\rm join}} - \frac{\Delta_{t+1}^{\rm join}}{\lambda_t} \big|\mathbf{X}_{\widetilde{i}_{t+1}^{\rm join}}^\top (\mathbf{X}_{\mathcal{A}}^+)^\top  \mathbf{X}_{\mathcal{A}}^\top(\mathbf{y} - \mathbf{X}_{\mathcal{A}}\widetilde{\bfbeta}_{\mathcal{A}}(\lambda_t))\big| \\
        &\geq \Lambda^{\rm join}_{\widetilde{i}_{t+1}^{\rm join}} - \frac{\Delta_{t+1}^{\rm join}}{\lambda_t} \|\mathbf{X}_{\mathcal{A}}^+\mathbf{X}_{\widetilde{i}_{t+1}^{\rm join}}   \|_1\|\mathbf{X}_{\mathcal{A}}^\top(\mathbf{y} - \mathbf{X}_{\mathcal{A}}\widetilde{\bfbeta}_{\mathcal{A}}(\lambda_t))\|_\infty \tag{H\"older's inequality}\\
        &\geq \Lambda^{\rm join}_{\widetilde{i}_{t+1}^{\rm join}} - \Delta_{t+1}^{\rm join}\|\mathbf{X}_{\mathcal{A}}^+\mathbf{X}_{\widetilde{i}_{t+1}^{\rm join}} \|_1 \tag{$\|\mathbf{X}_{\mathcal{A}}^\top(\mathbf{y} - \mathbf{X}_{\mathcal{A}}\widetilde{\bfbeta}_{\mathcal{A}}(\lambda_t))\|_\infty \leq \lambda_t$}\\
        &\geq \Big(1-\frac{\epsilon}{1+\alpha_{\mathcal{A}}}\Big)\widetilde{\lambda}_{t+1}^{\rm join} - \frac{\epsilon\widetilde{\lambda}^{\rm join}_{t+1}}{1+\alpha_{\mathcal{A}}}\|\mathbf{X}_{\mathcal{A}}^+\mathbf{X}_{\widetilde{i}_{t+1}^{\rm join}} \|_1 \tag{$\Lambda^{\rm join}_{\widetilde{i}_{t+1}^{\rm join}} \geq \big(1-\frac{\epsilon}{1+\alpha_{\mathcal{A}}}\big)\widetilde{\lambda}_{t+1}^{\rm join}$}\\
        &\geq (1-\epsilon)\widetilde{\lambda}_{t+1}^{\rm join}, \tag{$\max_{j\in \mathcal{A}^c}\|\mathbf{X}_{\mathcal{A}}^+\mathbf{X}_j\|_1 \leq \alpha_{\mathcal{A}}$}
    \end{align*}
    where we used that $\widetilde{\bfbeta}(\lambda_t)$ satisfies the $\mathsf{KKT}_{\lambda_t}(\epsilon)$ condition at $\lambda_t$,  $\|\mathbf{X}_{\mathcal{A}}^\top(\mathbf{y} - \mathbf{X}_{\mathcal{A}}\widetilde{\bfbeta}_{\mathcal{A}}(\lambda_t))\|_\infty \leq \lambda_t$.
\end{proof}

After asserting the correctness of \cref{alg:classical_quantum}, we now analyse its complexity. Specifically, we show how to sample the joining variable $\widetilde{i}_{t+1}^{\rm join}$ from the set in \cref{eq:joining_time} by combining the approximate quantum minimum-finding subroutine from Chen and de Wolf~\cite{chen2021quantum} (\cref{fact:minimum_finding}) and the unitary map that approximates (parts of) the joining times $\Lambda_i^{\rm join}(\mathcal A) = \frac{\mathbf{X}_i^\top (\mathbf{y}-\mathbf{X}_{\mathcal{A}}\bfmu)}{\pm 1 - \mathbf{X}_i^\top \mathbf{X}_{\mathcal{A}}\bftheta}$ (\cref{lem:terms_estimation}). Our final complexity depends on the mutual incoherence between different column subspaces and the overlap between $\mathbf{y}$ and different column subspaces. 
\begin{theorem}\label{thr:approximate_quantum}
    Let $\mathbf{X}\in\mathbb{R}^{n\times d}$ and $\mathbf{y}\in\mathbb{R}^n$. Assume $\mathbf{X}$ is stored in a {\rm QROM} and we have access to {\rm QRAMs} of size $O(n)$. Let $\delta\in(0,1)$, $\epsilon > 0$, and $T\in\mathbb{N}$. For $\mathcal{A}\subseteq[d]$ of size $|\mathcal{A}| \leq T$, assume there are $\alpha_{\mathcal{A}},\gamma_{\mathcal{A}}\in(0,1)$ such that
    \begin{align*}
        \|\mathbf{X}_{\mathcal{A}}^+\mathbf{X}_{\mathcal{A}^c}\|_1 \leq \alpha_{\mathcal{A}} \quad\text{and}\quad
        \frac{\|\mathbf{X}^\top(\mathbf{I} - \mathbf{X}_{\mathcal{A}}\mathbf{X}_{\mathcal{A}}^+) \mathbf{y}\|_\infty}{\|\mathbf{X}\|_{\max}\|(\mathbf{I} - \mathbf{X}_{\mathcal{A}}\mathbf{X}_{\mathcal{A}}^+) \mathbf{y}\|_1} \geq \gamma_{\mathcal{A}}.
    \end{align*}
    The approximate quantum {\rm LARS} {\rm \cref{alg:classical_quantum}} outputs, with probability at least $1-\delta$, a continuous piecewise linear approximate regularisation path $\widetilde{\mathcal{P}} = \{\widetilde{\bfbeta}(\lambda)\in\mathbb{R}^d:\lambda>0\}$ with error $\lambda\epsilon\|\widetilde{\bfbeta}(\lambda)\|_1$ and at most $T$ kinks in time
    \begin{align*}
        O\left(\frac{\gamma^{-1}_{\mathcal{A}} + \sqrt{n}\|\mathbf{X}\|_{\max}\|\mathbf{X}^+_{\mathcal{A}}\|_2}{(1-\alpha_{\mathcal{A}})\epsilon}\sqrt{|\mathcal{I}|}\log^2(T/\delta)\poly\log(nd) + n|\mathcal{A}| + |\mathcal{A}|^2 \right)
    \end{align*}
    per iteration, where $\mathcal{A}$ and $\mathcal{I}$ are the active and inactive sets of the corresponding iteration.
\end{theorem}
\begin{proof}
    The correctness of the approximate quantum LARS \cref{alg:classical_quantum} follows from \cref{thr:correctness}, since it is a LARS algorithm wherein the joining variable $\widetilde{i}_{t+1}^{\rm join}$ is taken from the set
    \begin{align}\label{eq:sampling_set}
        \left\{j\in\mathcal{I}: \Lambda_j^{\rm join}(\mathcal A) \geq \Big(1-\frac{\epsilon}{1+\alpha_{\mathcal{A}}}\Big)\max_{i\in \mathcal{I}}  \{\Lambda_i^{\rm join}(\mathcal A) \} \right\}
    \end{align}
    and the corresponding joining point $\widetilde{\lambda}_{t+1}^{\rm join}$ is
    \begin{align*}
        \widetilde{\lambda}_{t+1}^{\rm join} = \Big(1-\frac{\epsilon}{1+\alpha_{\mathcal{A}}}\Big)^{-1}\Lambda^{\rm join}_{\widetilde{i}_{t+1}^{\rm join}}(\mathcal A).
    \end{align*}
    Regarding the time complexity of the algorithm, the most expensive steps are computing $\mathbf{X}_{\mathcal{A}}^+$ in time $O(n|\mathcal{A}| + |\mathcal{A}|^2)$, inputting $\mathbf{y} - \mathbf{X}_{\mathcal{A}}\bfmu$ and $\mathbf{X}_{\mathcal{A}}\bftheta$ into KP-trees in time $O(n\log{n})$, and sampling $\widetilde{i}_{t+1}^{\rm join}$. We now show how to obtain the joining variable $\widetilde{i}_{t+1}^{\rm join}$, i.e., how to sample it from the set in \cref{eq:sampling_set}. In the following, fix $\mathcal{A}$ and $\mathcal{I}$.
    
    The first step is to create a unitary operator that maps $|i\rangle|\bar{0}\rangle \mapsto |i\rangle|R_i\rangle$ such that, for all $i\in\mathcal{I}$, after measuring the state $|R_i\rangle$, with probability at least $1-\delta_0$, where $\delta_0 = O\big(\delta^2/(T^2 |\mathcal{I}|\log(T/\delta))\big)$, the outcome $r_i$ of the first register satisfies
    \begin{align}\label{eq:map1}
        |r_i - \Lambda_i^{\rm join}(\mathcal A) | \leq \epsilon_0.
    \end{align}
    For such, we use \cref{lem:terms_estimation} to build maps $|i\rangle|\bar{0}\rangle \mapsto |i\rangle|\Psi_i\rangle$ and $|i\rangle|\bar{0}\rangle \mapsto |i\rangle|\Phi_i\rangle$ whose respective outcomes $\psi_i$ and $\phi_i$ of the first registers satisfy
    \begin{align*}
        |\psi_i - \mathbf{X}_i^\top (\mathbf{y} - \mathbf{X}_{\mathcal{A}}\bfmu)| &\leq \epsilon_1\|\mathbf{X}_i\|_{\infty}\|\mathbf{y} - \mathbf{X}_{\mathcal{A}}\bfmu\|_1 \leq \epsilon_1\|\mathbf{X}\|_{\max}\|\mathbf{y} - \mathbf{X}_{\mathcal{A}}\bfmu\|_1 , \\
        |\phi_i - \mathbf{X}_i^\top \mathbf{X}_{\mathcal{A}} \bftheta| &\leq \epsilon_2 \|\mathbf{X}_i\|_{\infty}\|\mathbf{X}_{\mathcal{A}}\bftheta\|_1 \leq \epsilon_2 \|\mathbf{X}\|_{\max} \|\mathbf{X}_{\mathcal{A}}\bftheta\|_1,
    \end{align*}
    in time $O(\epsilon_1^{-1}\log(1/\delta_0)\poly\log(nd))$ and $O(\epsilon_2^{-1}\log(1/\delta_0)\poly\log(nd))$, respectively. Then, by using \cref{lem:error_propagation}, we can construct the desired map $|i\rangle|\bar{0}\rangle \mapsto |i\rangle|R_i\rangle$ as in \cref{eq:map1} with
    \begin{align}\label{eq:epsilon0}
        \epsilon_0 = \max_{i\in\mathcal{I}}\left\{\Lambda_i^{\rm join}(\mathcal A)\left(\frac{\epsilon_1\|\mathbf{X}\|_{\max}\|\mathbf{y} - \mathbf{X}_{\mathcal{A}}\bfmu\|_1}{|\mathbf{X}_i^\top (\mathbf{y}-\mathbf{X}_{\mathcal{A}}\bfmu)|} + \frac{\epsilon_2 \|\mathbf{X}\|_{\max}
         \|\mathbf{X}_{\mathcal{A}}\bftheta\|_1}{|{\pm} 1 - \mathbf{X}_i^\top \mathbf{X}_{\mathcal{A}}\bftheta|}\right)\right\}.
    \end{align}
    Let us upper bound $\epsilon_0$. First note that $\max_{i\in\mathcal{I}}\{\Lambda^{\rm join}_i(\mathcal A)\} = \lambda_{t+1}^{\rm join}$ by the definition of $\lambda_{t+1}^{\rm join}$. Regarding the other quantities,
    %
    \begin{align*}
        \|\mathbf{X}_{\mathcal{A}}\bftheta\|_1 = \|(\mathbf{X}_{\mathcal{A}}^+)^\top\bfeta_{\mathcal{A}}\|_1
        &= \frac{1}{\lambda_t} \|(\mathbf{X}_{\mathcal{A}}^+)^\top\mathbf{X}_{\mathcal{A}}^\top(\mathbf{y} - \mathbf{X}_{\mathcal{A}}\widetilde{\bfbeta}_{\mathcal{A}}(\lambda_t)) \|_1 \\
        &\leq \frac{1}{\lambda_t} \|\mathbf{X}_{\mathcal{A}}^+\|_\infty\|\mathbf{X}_{\mathcal{A}}^\top(\mathbf{y} - \mathbf{X}_{\mathcal{A}}\widetilde{\bfbeta}_{\mathcal{A}}(\lambda_t)) \|_\infty \tag{H\"older's inequality}\\
        &\leq \|\mathbf{X}_{\mathcal{A}}^+\|_\infty \tag{$\|\mathbf{X}_{\mathcal{A}}^\top (\mathbf{y} - \mathbf{X}_{\mathcal{A}}\widetilde{\bfbeta}_{\mathcal{A}}(\lambda_t))\|_\infty \leq \lambda_t$}\\
        &\leq \sqrt{n} \|\mathbf{X}_{\mathcal{A}}^+\|_2,
        \tag{$\|\mathbf{X}_{\mathcal{A}}^+\|_\infty
        \leq \sqrt{n}\|\mathbf{X}_{\mathcal{A}}^+\|_2$}
    \end{align*}
    using that $\|\mathbf{X}_{\mathcal{A}}^\top (\mathbf{y} - \mathbf{X}_{\mathcal{A}}\widetilde{\bfbeta}_{\mathcal{A}}(\lambda_t))\|_\infty \leq \lambda_t$ since $\widetilde{\bfbeta}_{\mathcal{A}}(\lambda_t)$ satisfies the $\mathsf{KKT}_{\lambda_t}(\epsilon)$ condition. Also,
    \begin{align*}
        |{\pm} 1 - \mathbf{X}_i^\top \mathbf{X}_{\mathcal{A}}\bftheta| &\geq 1 - |\mathbf{X}_i^\top \mathbf{X}_{\mathcal{A}}\bftheta| \\
        &= 1 - \frac{1}{\lambda_t}|\mathbf{X}_i^\top (\mathbf{X}_{\mathcal{A}}^+)^\top \mathbf{X}_{\mathcal{A}}^\top (\mathbf{y} - \mathbf{X}_{\mathcal{A}}\widetilde{\bfbeta}_{\mathcal{A}}(\lambda_t))| \\
        &\geq 1 - \frac{1}{\lambda_t}\|\mathbf{X}_{\mathcal{A}}^+\mathbf{X}_i \|_1\|\mathbf{X}_{\mathcal{A}}^\top (\mathbf{y} - \mathbf{X}_{\mathcal{A}}\widetilde{\bfbeta}_{\mathcal{A}}(\lambda_t))\|_\infty \tag{H\"older's inequality}\\
        &\geq 1-\alpha_{\mathcal{A}}, \tag{$\|\mathbf{X}_{\mathcal{A}}^\top (\mathbf{y} - \mathbf{X}_{\mathcal{A}}\widetilde{\bfbeta}_{\mathcal{A}}(\lambda_t))\|_\infty \leq \lambda_t$ and $\|\mathbf{X}_{\mathcal{A}}^+\mathbf{X}_i \|_1 \leq \alpha_{\mathcal{A}}$}
    \end{align*}
    The above inequalities are sufficient to bound the second term in \cref{eq:epsilon0}. Regarding the first term,
    \begin{align*}
         \max_{i\in\mathcal{I}}\left\{\left|\frac{\mathbf{X}_i^\top (\mathbf{y}-\mathbf{X}_{\mathcal{A}}\bfmu)}{\pm 1 - \mathbf{X}_i^\top \mathbf{X}_{\mathcal{A}}\bftheta}\right|\frac{\epsilon_1\|\mathbf{X}\|_{\max}\|\mathbf{y} - \mathbf{X}_{\mathcal{A}}\bfmu\|_1}{|\mathbf{X}_i^\top (\mathbf{y}-\mathbf{X}_{\mathcal{A}}\bfmu)|}\right\} &= \epsilon_1 \max_{i\in\mathcal{I}}\left\{ \frac{\|\mathbf{X}\|_{\max}\|\mathbf{y} - \mathbf{X}_{\mathcal{A}}\bfmu\|_1}{|{\pm} 1 - \mathbf{X}_i^\top \mathbf{X}_{\mathcal{A}}\bftheta|}\right\}\\
         &\leq \frac{\epsilon_1}{1-\alpha_{\mathcal{A}}}\|\mathbf{X}\|_{\max}\|\mathbf{y} - \mathbf{X}_{\mathcal{A}}\bfmu\|_1 \\
         &= \frac{\epsilon_1}{1-\alpha_{\mathcal{A}}}\|\mathbf{X}\|_{\max}\|\mathbf{y} - \mathbf{X}_{\mathcal{A}}\bfmu\|_1 \frac{\lambda_{t+1}^{\rm join}}{\max_{i\in\mathcal{I}}\left|\frac{\mathbf{X}_i^\top (\mathbf{y}-\mathbf{X}_{\mathcal{A}}\bfmu)}{\pm 1 - \mathbf{X}_i^\top \mathbf{X}_{\mathcal{A}}\bftheta}\right|}\\
         &\leq \frac{(1+\alpha_{\mathcal{A}})\epsilon_1 \lambda_{t+1}^{\rm join}}{1-\alpha_{\mathcal{A}}}\frac{\|\mathbf{X}\|_{\max}\|\mathbf{y} - \mathbf{X}_{\mathcal{A}}\bfmu\|_1}{\max_{i\in\mathcal{I}}|\mathbf{X}_i^\top(\mathbf{y} - \mathbf{X}_{\mathcal{A}}\bfmu)|} \\
         &\leq \frac{(1+\alpha_{\mathcal{A}})\epsilon_1\gamma_{\mathcal{A}}^{-1}\lambda_{t+1}^{\rm join}}{1-\alpha_{\mathcal{A}}},
    \end{align*}
     using that $\max_{i\in\mathcal{I}}|\mathbf{X}_i^\top(\mathbf{y} - \mathbf{X}_{\mathcal{A}}\bfmu)| = \|\mathbf{X}_{\mathcal{I}}^\top(\mathbf{y} - \mathbf{X}_{\mathcal{A}}\bfmu)\|_\infty = \|\mathbf{X}^\top(\mathbf{y} - \mathbf{X}_{\mathcal{A}}\bfmu)\|_\infty \geq \gamma_{\mathcal{A}}\|\mathbf{X}\|_{\max}\|\mathbf{y} - \mathbf{X}_{\mathcal{A}}\bfmu\|_1$ and that $1-\alpha_{\mathcal{A}} \leq |{\pm} 1 - \mathbf{X}_i^\top \mathbf{X}_{\mathcal{A}}\bftheta| \leq 1+\alpha_{\mathcal{A}}$. All the above leads to
    \begin{align*}
         \epsilon_0 \leq \frac{\epsilon_1 (1+\alpha_{\mathcal{A}})\gamma_{\mathcal{A}}^{-1} + \epsilon_2 \sqrt{n}\|\mathbf{X}\|_{\max} \|\mathbf{X}_{\mathcal{A}}^+\|_2}{1-\alpha_{\mathcal{A}}}\lambda^{\rm join}_{t+1}.
    \end{align*}
    By taking 
    \begin{align*}
        \epsilon_1 =  \frac{\epsilon(1-\alpha_{\mathcal{A}})\gamma_{\mathcal{A}}}{2(1+\alpha_{\mathcal{A}})^2} \quad\text{and}\quad \epsilon_2 =  \frac{\epsilon(1-\alpha_{\mathcal{A}})}{2(1+\alpha_{\mathcal{A}})\sqrt{n}\|\mathbf{X}\|_{\max}\|\mathbf{X}_{\mathcal{A}}^+\|_2},
    \end{align*}
    then $\epsilon_0 \leq \frac{\epsilon}{2(1+\alpha_{\mathcal{A}})}\lambda^{\rm join}_{t+1}$ in time $O\Big(\frac{\gamma_{\mathcal{A}}^{-1} + \sqrt{n}\|\mathbf{X}\|_{\max}\|\mathbf{X}^+_{\mathcal{A}}\|_2}{(1-\alpha_{\mathcal{A}})\epsilon}\log(1/\delta_0)\poly\log(nd) \Big)$. Finally, we apply the approximate quantum minimum-finding algorithm (\cref{fact:minimum_finding}) to obtain an index $\widetilde{i}_{t+1}^{\rm join}\in\mathcal{I}$ such that
    \begin{align*}
        \Lambda^{\rm join}_{\widetilde{i}_{t+1}^{\rm join}}(\mathcal A) \geq \max_{i\in \mathcal{I}}  \{\Lambda^{\rm join}_i(\mathcal A)\} - 2\epsilon_0 \geq \Big(1-\frac{\epsilon}{1+\alpha_{\mathcal{A}}}\Big)\max_{i\in \mathcal{I}}  \{\Lambda_i^{\rm join}(\mathcal A)\},
    \end{align*}
    i.e., $\widetilde{i}_{t+1}^{\rm join}$ belongs to the set from \cref{eq:sampling_set} with probability at least $1-\delta/T$ as promised. 
    The total complexity is $\widetilde{O}(\sqrt{|\mathcal{I}|}\log(T/\delta))$ times the complexity of constructing the mapping $|i\rangle|\bar{0}\rangle \mapsto |i\rangle|R_i\rangle$. The final success probability follows from a union bound.
\end{proof}

If $\|\mathbf{X}\|_{\max}\|\mathbf{X}_{\mathcal{A}}^+\|_2$ is constant and the parameters $\alpha_{\mathcal{A}},\gamma_{\mathcal{A}}\in(0,1)$ are bounded away from $1$ and~$0$, respectively, then the complexity of \cref{alg:classical_quantum} is $\widetilde{O}(\sqrt{\smash[b]{n|\mathcal{I}|}}/\epsilon + n|\mathcal{A}| + |\mathcal{A}|^2)$ per iteration. If also $|\mathcal{A}| = O(n)$, then we obtain a fully quadratic advantage $\widetilde{O}(\sqrt{nd})$ over the classical LARS algorithm.

\subsection{Approximate classical LARS algorithm}
\label{sec:approximate_classical}

By using similar techniques and results behind the approximate quantum LARS algorithm, it is possible to devise an analogous approximate classical LARS algorithm, which can be seen as a dequantised counterpart to our quantum algorithms. The idea is again to approximately compute the joining times $\Lambda_i^{\rm join}(\mathcal A)=\frac{\mathbf{X}_i^\top (\mathbf{y}-\mathbf{X}_{\mathcal{A}}\bfmu)}{\pm 1 - \mathbf{X}_i^\top \mathbf{X}_{\mathcal{A}}\bftheta}$, but by using the classical sampling procedure from \cref{lem:sampling-based_inner_product_estimation} instead. For well-behaved design matrices, the complexity per iteration is $\widetilde{O}(n|\mathcal{I}|/\epsilon^2 + n|\mathcal{A}| + |\mathcal{A}|^2)$. The approximate classical LARS algorithm for the pathwise Lasso is shown in \cref{alg:approximate_classical}.

\begin{algorithm}[ht]
\caption{Approximate classical LARS algorithm for the pathwise Lasso}
\DontPrintSemicolon
\label{alg:approximate_classical}

\KwIn{$T\in\mathbb{N}$, $\delta\in(0,1)$, $\epsilon>0$, $\mathbf{y}\in\mathbb{R}^n$, $\mathbf{X}\in\mathbb{R}^{n\times d}$, and $\{\alpha_{\mathcal{A}} \geq \|\mathbf{X}_{\mathcal{A}}^+\mathbf{X}_{\mathcal{A}^c}\|_1\}_{\mathcal{A}\subseteq[d]:|\mathcal{A}|\leq T}$}

\KwOut{Regularisation path $\widetilde{\mathcal{P}}$ with $T$ kinks and error $\lambda \epsilon \|\widetilde{\bfbeta}(\lambda)\|_1$ with probability $\geq 1-\delta$}

Initialise $\mathcal{A}=\operatorname{argmax}_{j\in[d]} |\mathbf{X}_j^\top \mathbf{y}|$, $\mathcal{I}=[d]\setminus\mathcal{A}$, $\lambda_0 = \|\mathbf{X}^\top\mathbf{y}\|_\infty$, $\widetilde{\bfbeta}(\lambda_0) = \mathbf{0}$, $t=0$

    
   \While{$\mathcal{I}\neq \varnothing$, $t \leq T$}
   {
        $\bfeta_{\mathcal{A}} \gets \frac{1}{\lambda_{t}}\mathbf{X}_{\mathcal{A}}^\top(\mathbf{y} - \mathbf{X}_{\mathcal{A}}\widetilde{\bfbeta}_{\mathcal{A}}(\lambda_{t}))$

        $\bfmu \gets \mathbf{X}_{\mathcal{A}}^+ \mathbf{y}$ \tcp*{$\bfmu \in \mathbb{R}^{|\mathcal{A}|}$}
        
        $\bftheta \gets (\mathbf{X}_{\mathcal{A}}^\top \mathbf{X}_{\mathcal{A}})^+ \bfeta_{\mathcal{A}}$ \tcp*{$\bftheta \in \mathbb{R}^{|\mathcal{A}|}$}

        Input $\mathbf{y} - \mathbf{X}_{\mathcal{A}}\bfmu$ and $\mathbf{X}_{\mathcal{A}}\bftheta$ into classical-samplable memories

        Define $\Lambda_i^{\rm join} \triangleq \frac{\mathbf{X}_i^\top (\mathbf{y}-\mathbf{X}_{\mathcal{A}}\bfmu)}{\pm 1 - \mathbf{X}_i^\top \mathbf{X}_{\mathcal{A}}\bftheta}$ and $\Lambda_i^{\rm cross} \triangleq \frac{\mu_i}{\theta_i} \cdot \mathbf{1}\left[\frac{\mu_i}{\theta_i} \leq \lambda_t\right]$
   
        $i_{t+1}^{\rm cross}\gets \operatorname{argmax}_{i\in \mathcal{A}} \{\Lambda_i^{\rm cross}\}$ and $\lambda_{t+1}^{\rm cross}\gets \max_{i\in \mathcal{A}} \{\Lambda_i^{\rm cross}\}$
        
        Compute $\{r_i\}_{i\in\mathcal{I}}$ such that $|r_i - \Lambda_i^{\rm join}| \leq \frac{\epsilon\lambda_{t+1}^{\rm join}}{1+\alpha_{\mathcal{A}}}$ with failure probability $\frac{\delta}{T|\mathcal{I}|}$ (\cref{lem:sampling-based_inner_product_estimation})
        
        $\widetilde{i}_{t+1}^{\rm join} \gets \operatorname{argmax}_{i\in\mathcal{I}} \{r_i\}$ and $\widetilde{\lambda}_{t+1}^{\rm join} \gets \big(1-\frac{\epsilon}{1+\alpha_{\mathcal{A}}}\big)^{-1}\max_{i\in\mathcal{I}} \{r_i\}$
        
        $\lambda_{t+1} \gets \min\{\lambda_t,\max\{\lambda_{t+1}^{\rm cross}, \widetilde{\lambda}_{t+1}^{\rm join}\}\}$\;
        
        $\widetilde{\bfbeta}_{\mathcal{A}}(\lambda_{t+1}) \gets  \bfmu - \lambda_{t+1}\bftheta$\;
        
        \If{$\lambda_{t+1} = \lambda_{t+1}^{\rm cross}$}
        {        
            Move $i_{t+1}^{\rm cross}$ from $\mathcal{A}$ to $\mathcal{I}$
            
        }
        \Else
        {
            Move $\widetilde{i}_{t+1}^{\rm join}$ from $\mathcal{I}$ to $\mathcal{A}$
        }
        
        $t\gets t+1$
    
  }
\Return coefficients $[(\lambda_0,\widetilde{\bfbeta}(\lambda_0)),(\lambda_1,\widetilde{\bfbeta}(\lambda_1)),\dots]$
\end{algorithm}

\begin{theorem}\label{thr:approximate_classical}
    Let $\mathbf{X}\in\mathbb{R}^{n\times d}$ and $\mathbf{y}\in\mathbb{R}^n$. Assume access to classical-samplable structures of size $O(n)$. 
    Let $\delta\in(0,1)$, $\epsilon > 0$, and $T\in\mathbb{N}$. For $\mathcal{A}\subseteq[d]$ of size $|\mathcal{A}|\leq T$, assume there are $\alpha_{\mathcal{A}},\gamma_{\mathcal{A}}\in(0,1)$ such that
    \begin{align*}
        \|\mathbf{X}_{\mathcal{A}}^+\mathbf{X}_{\mathcal{A}^c}\|_1 \leq \alpha_{\mathcal{A}} \quad\text{and}\quad
        \frac{\|\mathbf{X}^\top(\mathbf{I} - \mathbf{X}_{\mathcal{A}}\mathbf{X}_{\mathcal{A}}^+) \mathbf{y}\|_\infty}{\|\mathbf{X}\|_{\max}\|(\mathbf{I} - \mathbf{X}_{\mathcal{A}}\mathbf{X}_{\mathcal{A}}^+) \mathbf{y}\|_1} \geq \gamma_{\mathcal{A}}.
    \end{align*}
    The approximate classical {\rm LARS} {\rm \cref{alg:approximate_classical}} outputs, with probability at least $1-\delta$, a continuous piecewise linear approximate regularisation path $\widetilde{\mathcal{P}} = \{\widetilde{\bfbeta}(\lambda)\in\mathbb{R}^d:\lambda>0\}$ with error $\lambda\epsilon\|\widetilde{\bfbeta}(\lambda)\|_1$ and at most $T$ kinks in time
    \begin{align*}
        O\left(\frac{\gamma_{\mathcal{A}}^{-2} + n\|\mathbf{X}\|_{\max}^2\|\mathbf{X}^+_{\mathcal{A}}\|_2^2}{(1-\alpha_{\mathcal{A}})^2\epsilon^2}|\mathcal{I}|\log(Td/\delta)\poly\log{n} + n|\mathcal{A}| + |\mathcal{A}|^2 \right)
    \end{align*}
    per iteration, where $\mathcal{A}$ and $\mathcal{I}$ are the active and inactive sets of the corresponding iteration.
\end{theorem}
\begin{proof}
    We start by noticing that the joining time computation step outputs $\widetilde{i}_{t+1}^{\rm join} = \operatorname{argmax}_{i\in\mathcal{I}} \{r_i\}$ and $\widetilde{\lambda}_{t+1}^{\rm join} = \big(1-\frac{\epsilon}{1+\alpha_{\mathcal{A}}}\big)^{-1}\max_{i\in\mathcal{I}} \{r_i\}$, where $|r_i - \Lambda_i^{\rm join}(\mathcal A) | \leq \frac{\epsilon}{1+\alpha_{\mathcal{A}}}\lambda_{t+1}^{\rm join}$.
    %
    Since $\max_{i\in\mathcal{I}} \{r_i\} \geq r_{i_{t+1}^{\rm join}} \geq \Lambda^{\rm join}_{{i}_{t+1}^{\rm join}}(\mathcal A)  - \frac{\epsilon}{1+\alpha_{\mathcal{A}}}\lambda_{t+1}^{\rm join} = \big(1-\frac{\epsilon}{1+\alpha_{\mathcal{A}}}\big)\lambda_{t+1}^{\rm join}$, this means the joining variable $\widetilde{i}_{t+1}^{\rm join}$ is taken from the set
    \begin{align*}
        \left\{j\in\mathcal{I}:  \Lambda_j^{\rm join}(\mathcal A)  \geq \Big(1-\frac{\epsilon}{1+\alpha_{\mathcal{A}}}\Big)\max_{i\in \mathcal{I}}  \{ \Lambda_i^{\rm join}(\mathcal A)  \}\right\}.
    \end{align*}
    \cref{thr:correctness} thus guarantees the correctness of the algorithm. Regarding the time complexity, the most expensive steps are computing $\mathbf{X}_{\mathcal{A}}^+$ in time $O(n|\mathcal{A}| + |\mathcal{A}|^2)$, inputting $\mathbf{y} - \mathbf{X}_{\mathcal{A}}\bfmu$ and $\mathbf{X}_{\mathcal{A}}\bftheta$ into classical-samplable memories in time $O(n\log{n})$, and computing the $|\mathcal{I}|$ values $\{r_i\}_{i\in\mathcal{I}}$. By following the exact same steps as in the proof of \cref{thr:approximate_quantum}, approximating any quantity $ \Lambda_i^{\rm join}(\mathcal A) $ within precision $\frac{\epsilon}{1+\alpha_{\mathcal{A}}}\lambda_{t+1}^{\rm join}$ and failure probability at most $\frac{\delta}{T|\mathcal{I}|}$ requires time
    \begin{align*}
        O\left(\frac{\gamma_{\mathcal{A}}^{-2} + n\|\mathbf{X}\|^2_{\max}\|\mathbf{X}^+_{\mathcal{A}}\|^2_2}{(1-\alpha_{\mathcal{A}})^2\epsilon^2}\log(Td/\delta)\poly\log{n} \right)
    \end{align*}
    (note that the term $O(\poly\log{n})$ comes from sampling the vectors $\mathbf{y} - \mathbf{X}_{\mathcal{A}}\bfmu$ and $\mathbf{X}_{\mathcal{A}}\bftheta$). The final success probability follows from a union bound.
\end{proof}

If $\|\mathbf{X}\|_{\max}\|\mathbf{X}^+_{\mathcal{A}}\|_2 = o(1)$ and the parameters $\alpha_{\mathcal{A}}$ and $\gamma_{\mathcal{A}}$ are bounded away from $1$ and $0$, respectively, then it is possible to obtain a speedup compared to the standard LARS algorithm as we shall see in the next section. For several inputs $\mathbf{X}$ and $\mathbf{y}$, however, we expect \Cref{alg:approximate_classical} to perform worse than \Cref{alg:classicalLasso}. We interpret this as evidence for the difficulty in properly dequantising the approximate quantum LARS \Cref{alg:classical_quantum}. For ill-behaved design matrices, approximating a function of its entries should be as costly as exactly computing it.

\subsection{Gaussian random design matrix}
\label{sec:examples}

For $\mathcal{A}\subseteq[d]$, the complexity of the approximate LARS algorithms depends on the quantities
\begin{align*}
    \|\mathbf{X}\|_{\max}\|\mathbf{X}^+_{\mathcal{A}}\|_2, \qquad  \|\mathbf{X}_{\mathcal{A}}^+\mathbf{X}_{\mathcal{A}^c}\|_1, \qquad\text{and}\qquad \frac{\|\mathbf{X}^\top(\mathbf{I} - \mathbf{X}_{\mathcal{A}}\mathbf{X}_{\mathcal{A}}^+) \mathbf{y}\|_\infty}{\|\mathbf{X}\|_{\max}\|(\mathbf{I} - \mathbf{X}_{\mathcal{A}}\mathbf{X}_{\mathcal{A}}^+) \mathbf{y}\|_1}.
\end{align*}
In this section, we shall bound these quantities for the specific case when the design matrix $\mathbf{X}\in\mathbb{R}^{n\times d}$ is a random matrix from the $\mathbf{\Sigma}$-Gaussian ensemble. According to \cref{fact:continuous_distribution}, in this situation the Lasso solution is unique almost surely, thus $\operatorname{null}(\mathbf{X}_{\mathcal{A}}) = \{\mathbf{0}\}$ and $\mathbf{X}_{\mathcal{A}}^+ = (\mathbf{X}_{\mathcal{A}}^\top \mathbf{X}_{\mathcal{A}})^{-1} \mathbf{X}_{\mathcal{A}}^\top$. As we shall prove below, with high probability, $\|\mathbf{X}\|_{\max}\|\mathbf{X}^+_{\mathcal{A}}\|_2$ is $O(\sqrt{\log(d)/n)}$ and the mutual incoherence is at most $1/2$. The mutual overlap between $\mathbf{y}$ and $\mathbf{X}$, on the other hand, is $\Omega(\|\mathbf{y}\|_2/(\|\mathbf{y}\|_1\sqrt{\log{d}}))$ with high probability, which equals $\Omega(1/\sqrt{\log{d}})$ when $\mathbf{y}$ is sparse. These bounds yield quantum and classical complexities $\widetilde{O}(\sqrt{d}/\epsilon)$ and $\widetilde{O}(d/\epsilon^2)$ per iteration for \Cref{alg:classical_quantum,alg:approximate_classical}, respectively, as stated in \Cref{res:result_gaussian}.

Let us start our analysis with the quantity $\|\mathbf{X}\|_{\max}\|\mathbf{X}^+_{\mathcal{A}}\|_2$. The next result states that, if $\mathbf{\Sigma}$ is well behaved, which includes $\mathbf{\Sigma} = \mathbf{I}$, then $\|\mathbf{X}\|_{\max}\|\mathbf{X}_{\mathcal{A}}^+\|_2 = O(\sqrt{\log(d)/n})$ with high probability.
\begin{lemma}\label{lem:condition_number}
    Let $\mathbf{X}\in\mathbb{R}^{n\times d}$ be a random matrix from the $\mathbf{\Sigma}$-Gaussian ensemble, $\mathbf{\Sigma}\succ \mathbf{0}$. Given $\mathcal{A}\subseteq[d]$ of size $|\mathcal{A}|\leq n$, let $\mathbf{\Sigma}_{\mathcal{A}|\mathcal{A}^c} = \mathbf{\Sigma}_{\mathcal{A}\mathcal{A}} - \mathbf{\Sigma}_{\mathcal{A}\mathcal{A}^c}(\mathbf{\Sigma}_{\mathcal{A}^c\mathcal{A}^c})^{-1}\mathbf{\Sigma}_{\mathcal{A}^c\mathcal{A}} \succ \mathbf{0}$ be the conditional covariance matrix of $\mathbf{X}_{\mathcal{A}}|\mathbf{X}_{\mathcal{A}^c}$. Then, for $\delta\in(2e^{-n/2},1)$,
    \begin{align*}
        \mathbb{P}\left[\|\mathbf{X}\|_{\max}\|\mathbf{X}_{\mathcal{A}}^+\|_2 \leq \frac{\|\sqrt{\mathbf{\Sigma}}\|_1\sqrt{\ln(4nd/\delta)}}{(\sqrt{n}-\sqrt{2\ln(2/\delta)})\sigma_{\min}(\sqrt{\smash[b]{\mathbf{\Sigma}_{\mathcal{A}|\mathcal{A}^c}}}) - \sqrt{\operatorname{Tr}[\mathbf{\Sigma}_{\mathcal{A}|\mathcal{A}^c}]}}\right] \geq 1-\delta.
    \end{align*}
    In particular, if $\mathbf{\Sigma} = \mathbf{I}$ and $|\mathcal{A}| \leq n/2$, then
    \begin{align*}
        \mathbb{P}\left[\|\mathbf{X}\|_{\max}\|\mathbf{X}_{\mathcal{A}}^+\|_2 \leq \frac{\sqrt{\ln(4nd/\delta)}}{(1-1/\sqrt{2})\sqrt{n}-\sqrt{2\ln(2/\delta)}}\right] \geq 1-\delta.
    \end{align*}
\end{lemma}
\begin{proof}
    Let us start by bounding $\|\mathbf{X}\|_{\max}$. We can write $\mathbf{X} = \mathbf{Z}\sqrt{\mathbf{\Sigma}}$ where $\mathbf{Z}$ is drawn from the $\mathbf{I}$-Gaussian ensemble. Let $\mathbf{Z}^{(i)}\in\mathbb{R}^{d\times 1}$ be the $i$-th row of $\mathbf{Z}$. Then 
    \begin{align*}
        \|\mathbf{X}\|_{\max} = \|\mathbf{Z}\sqrt{\mathbf{\Sigma}}\|_{\max} = \max_{i\in[n]}\|\sqrt{\mathbf{\Sigma}}^\top\mathbf{Z}^{(i)}\|_\infty \leq \max_{i\in[n]}\|\sqrt{\mathbf{\Sigma}}^\top\|_\infty\|\mathbf{Z}^{(i)}\|_\infty = \|\sqrt{\mathbf{\Sigma}}\|_1 \|\mathbf{Z}\|_{\max}.
    \end{align*}
    By a Chernoff bound (\Cref{fact:Chernoff}) plus a union bound, we have that $\mathbb{P}[\|\mathbf{Z}\|_{\max} \leq \sqrt{\ln(4nd/\delta)}] \geq 1-\delta/2$, which implies that $\mathbb{P}[\|\mathbf{X}\|_{\max} \leq \sqrt{\ln(4nd/\delta)}\|\sqrt{\mathbf{\Sigma}}\|_1] \geq 1-\delta/2$.

    We now focus on $\|\mathbf{X}^+_{\mathcal{A}}\|_2$. The rows of $\mathbf{X}_{\mathcal{A}}$ are i.i.d.\ sampled from the Gaussian distribution $\mathcal{N}(\mathbf{0},\mathbf{\Sigma}_{\mathcal{A}|\mathcal{A}^c})$. According to~\cite[Theorem~6.1]{wainwright2019high},
    \begin{align*}
        \mathbb{P}\left[\|\mathbf{X}_{\mathcal{A}}^+\|_2^{-1} \leq (1-\epsilon)\sqrt{n}\sigma_{\min}(\sqrt{\smash[b]{\mathbf{\Sigma}_{\mathcal{A}|\mathcal{A}^c}}}) - \sqrt{\operatorname{Tr}[\mathbf{\Sigma}_{\mathcal{A}|\mathcal{A}^c}]} \right] \leq e^{-n\epsilon^2/2} \quad\text{for all}~\epsilon\in(0,1).
    \end{align*}
    By taking $\epsilon = \sqrt{2\ln(2/\delta)/n} < 1$, the above probability is at most $\delta/2$. The result follows by putting together the bounds on $\|\mathbf{X}\|_{\max}$ and $\|\mathbf{X}_{\mathcal{A}}^+\|_2$.

    In case $\mathbf{\Sigma} = \mathbf{I}$, then $\mathbf{\Sigma}_{\mathcal{A}|\mathcal{A}^c} = \mathbf{I}$, $\sigma_{\min}(\sqrt{\smash[b]{\mathbf{\Sigma}_{\mathcal{A}|\mathcal{A}^c}}}) = 1$, $\operatorname{Tr}[\mathbf{\Sigma}_{\mathcal{A}|\mathcal{A}^c}] = |\mathcal{A}|$, and $\|\mathbf{\Sigma}\|_1 = 1$. Therefore, $\|\mathbf{X}_{\mathcal{A}}^+\|_2^{-1} \geq \sqrt{n} - \sqrt{|\mathcal{A}|} - \sqrt{2\ln(2/\delta)} \geq (1-1/\sqrt{2})\sqrt{n} - \sqrt{2\ln(2/\delta)}$ with high probability.
    %
\end{proof}


We now move on to $\|\mathbf{X}_{\mathcal{A}}^+\mathbf{X}_{\mathcal{A}^c}\|_1$. The next lemma tells us that, if $|\mathcal{A}| = O(\sigma_{\min}(\mathbf{\Sigma}) n/\log{d})$, where $\sigma_{\min}(\mathbf{\Sigma})$ is the minimum eigenvalue of $\mathbf{\Sigma}$, then the mutual incoherence is a constant bounded away from $1$ with high probability. 
A similar result can be found in~\cite[Exercise~7.19]{wainwright2019high}.

\begin{lemma}\label{lem:bound_mutual_incoherence}
    Let $\mathbf{X}\in\mathbb{R}^{n\times d}$ be a random matrix from the $\mathbf{\Sigma}$-Gaussian ensemble, where $\mathbf{\Sigma}\succ \mathbf{0}$ has minimum eigenvalue $\sigma_{\min}(\mathbf{\Sigma})\in(0,1]$ and its diagonal entries are at most $1$. Given $\mathcal{A}\subseteq[d]$, suppose that there is $\bar{\alpha}\in[0,1)$ such that
    \begin{align*}
        \max_{j\in\mathcal{A}^c} \|\mathbf{\Sigma}_{j\mathcal{A}} (\mathbf{\Sigma}_{\mathcal{A}\mathcal{A}})^{-1}\|_1 \leq \bar{\alpha}.
    \end{align*}
    Given $\delta\in(0,1)$, if $|\mathcal{A}| \leq (1-\bar{\alpha})^2\sigma_{\min}(\mathbf{\Sigma)}n/72\ln(4d/\delta)$, then 
    \begin{align*}
        \mathbb{P}\left[\|\mathbf{X}_{\mathcal{A}}^+\mathbf{X}_{\mathcal{A}^c}\|_1 \leq \frac{1}{2} + \frac{\bar{\alpha}}{2} \right] \geq 1-\delta.
    \end{align*}
\end{lemma}
\begin{proof}
    First notice that
    \begin{align*}
        \|\mathbf{X}_{\mathcal{A}}^+\mathbf{X}_{\mathcal{A}^c}\|_1 = \|(\mathbf{X}_{\mathcal{A}}^\top \mathbf{X}_{\mathcal{A}})^{-1} \mathbf{X}_{\mathcal{A}}^\top\mathbf{X}_{\mathcal{A}^c}\|_1 =  \|\mathbf{X}_{\mathcal{A}^c}^\top \mathbf{X}_{\mathcal{A}}(\mathbf{X}_{\mathcal{A}}^\top \mathbf{X}_{\mathcal{A}})^{-1} \|_\infty.
    \end{align*}
    For any $\mathbf{u}\in\mathbb{R}^{|\mathcal{A}|}$ such that $\|\mathbf{u}\|_\infty = 1$, note that
    \begin{align*}
        \|\mathbf{X}_{\mathcal{A}^c}^\top \mathbf{X}_{\mathcal{A}}(\mathbf{X}_{\mathcal{A}}^\top \mathbf{X}_{\mathcal{A}})^{-1} \|_\infty \geq \|\mathbf{X}_{\mathcal{A}^c}^\top \mathbf{X}_{\mathcal{A}}(\mathbf{X}_{\mathcal{A}}^\top \mathbf{X}_{\mathcal{A}})^{-1}\mathbf{u} \|_\infty = \max_{j\in\mathcal{A}^c} |\mathbf{X}_{j}^\top \mathbf{X}_{\mathcal{A}}(\mathbf{X}_{\mathcal{A}}^\top \mathbf{X}_{\mathcal{A}})^{-1}\mathbf{u}|.
    \end{align*}
    Consider the quantity $|\mathbf{X}_{j}^\top \mathbf{X}_{\mathcal{A}}(\mathbf{X}_{\mathcal{A}}^\top \mathbf{X}_{\mathcal{A}})^{-1}\mathbf{u}|$ for fixed $j\in\mathcal{A}^c$. The zero-mean Gaussian vector $\mathbf{X}_j$ can be decomposed into a linear prediction based on $\mathbf{X}_{\mathcal{A}}$ plus a prediction error as (cf.\ \cite[Section~V.A]{wainwright2009sharp} and \cite[Equation (11.26) \& Exercise~11.3]{wainwright2019high})
    \begin{align}\label{eq:linear_prediction}
        \mathbf{X}_j^\top = \mathbf{\Sigma}_{j\mathcal{A}}(\mathbf{\Sigma}_{\mathcal{A}\mathcal{A}})^{-1}\mathbf{X}_{\mathcal{A}}^\top + \mathbf{W}_j^\top,
    \end{align}
    where $\mathbf{W}_j\in\mathbb{R}^n$ is a vector with i.i.d.\ $\mathcal{N}(0,[\mathbf{\Sigma}_{\mathcal{A}^c|\mathcal{A}}]_{jj})$ entries. Here the (positive-definite) matrix $\mathbf{\Sigma}_{\mathcal{A}^c|\mathcal{A}} = \mathbf{\Sigma}_{\mathcal{A}^c\mathcal{A}^c} - \mathbf{\Sigma}_{\mathcal{A}^c\mathcal{A}}(\mathbf{\Sigma}_{\mathcal{A}\mathcal{A}})^{-1}\mathbf{\Sigma}_{\mathcal{A}\mathcal{A}^c}$ is the conditional covariance matrix of $\mathbf{X}_{\mathcal{A}^c}|\mathbf{X}_{\mathcal{A}}$. Thus
    \begin{align*}
        |\mathbf{X}_{j}^\top \mathbf{X}_{\mathcal{A}}(\mathbf{X}_{\mathcal{A}}^\top \mathbf{X}_{\mathcal{A}})^{-1}\mathbf{u}| &\leq |\mathbf{\Sigma}_{j\mathcal{A}}(\mathbf{\Sigma}_{\mathcal{A}\mathcal{A}})^{-1}\mathbf{X}_{\mathcal{A}}^\top \mathbf{X}_{\mathcal{A}}(\mathbf{X}_{\mathcal{A}}^\top \mathbf{X}_{\mathcal{A}})^{-1}\mathbf{u}| +  |\mathbf{W}_{j}^\top \mathbf{X}_{\mathcal{A}}(\mathbf{X}_{\mathcal{A}}^\top \mathbf{X}_{\mathcal{A}})^{-1}\mathbf{u}| \\
        &\leq \|\mathbf{\Sigma}_{j\mathcal{A}}(\mathbf{\Sigma}_{\mathcal{A}\mathcal{A}})^{-1}\|_1\|\mathbf{u}\|_\infty + |\mathbf{W}_{j}^\top \mathbf{X}_{\mathcal{A}}(\mathbf{X}_{\mathcal{A}}^\top \mathbf{X}_{\mathcal{A}})^{-1}\mathbf{u}| \\
        &\leq \bar{\alpha} + |\mathbf{W}_{j}^\top \mathbf{X}_{\mathcal{A}}(\mathbf{X}_{\mathcal{A}}^\top \mathbf{X}_{\mathcal{A}})^{-1}\mathbf{u}|.
    \end{align*}
    Since $\mathbf{W}_j$ is independent of $\mathbf{X}_{\mathcal{A}}$, $\mathbf{W}_{j}^\top \mathbf{X}_{\mathcal{A}}(\mathbf{X}_{\mathcal{A}}^\top \mathbf{X}_{\mathcal{A}})^{-1}\mathbf{u}$ is a sub-Gaussian random variable with parameter (\cref{fact:sub_gaussian})
    \begin{align*}
        [\mathbf{\Sigma}_{\mathcal{A}^c|\mathcal{A}}]_{jj}\|\mathbf{X}_{\mathcal{A}}(\mathbf{X}_{\mathcal{A}}^\top \mathbf{X}_{\mathcal{A}})^{-1}\mathbf{u}\|_2^2 &\leq \|\mathbf{X}_{\mathcal{A}}(\mathbf{X}_{\mathcal{A}}^\top \mathbf{X}_{\mathcal{A}})^{-1}\mathbf{u}\|_2^2 \\
        &= \mathbf{u}^\top (\mathbf{X}_{\mathcal{A}}^\top \mathbf{X}_{\mathcal{A}})^{-1}\mathbf{u}
        \leq \|\mathbf{u}\|_2^2 \|(\mathbf{X}_{\mathcal{A}}^\top \mathbf{X}_{\mathcal{A}})^{-1}\|_2
        \leq |\mathcal{A}|\|(\mathbf{X}_{\mathcal{A}}^\top \mathbf{X}_{\mathcal{A}})^{-1}\|_2,
    \end{align*}
    where we used that $[\mathbf{\Sigma}_{\mathcal{A}^c|\mathcal{A}}]_{jj} \leq 1$ since $\mathbf{\Sigma}_{\mathcal{A}^c|\mathcal{A}}\preceq \mathbf{\Sigma}_{\mathcal{A}^c\mathcal{A}^c}$ and the diagonal entries of $\mathbf{\Sigma}$ are at most~$1$. We now relate $\|(\mathbf{X}_{\mathcal{A}}^\top \mathbf{X}_{\mathcal{A}})^{-1}\|_2$ with $\sigma_{\min}(\mathbf{\Sigma}_{\mathcal{A}\mathcal{A}}) \geq \sigma_{\min}(\mathbf{\Sigma})$. According to~\cite[Lemma~9]{wainwright2009sharp},
    \begin{align*}
        \mathbb{P}\left[\|(\mathbf{X}_{\mathcal{A}}^\top \mathbf{X}_{\mathcal{A}})^{-1}\|_2 \geq \frac{(1 + t + \sqrt{|\mathcal{A}|/n})^2}{n\sigma_{\min}(\mathbf{\Sigma})} \right] \leq 2e^{-nt^2/2} \quad\text{for all}~t>0.
    \end{align*}
    Define the event $\mathcal{E}(\mathbf{X}_{\mathcal{A}}) \triangleq \{\|(\mathbf{X}_{\mathcal{A}}^\top \mathbf{X}_{\mathcal{A}})^{-1}\|_2 \geq 9/(n\sigma_{\min}(\mathbf{\Sigma}))\}$. Since $|\mathcal{A}| \leq n$, the event $\mathcal{E}(\mathbf{X}_{\mathcal{A}})$ happens with probability at most $2e^{-n/2}$. Therefore, with probability at least $1-2e^{-n/2}$, the quantity $\mathbf{W}_{j}^\top \mathbf{X}_{\mathcal{A}}(\mathbf{X}_{\mathcal{A}}^\top \mathbf{X}_{\mathcal{A}})^{-1}\mathbf{u}$ is a sub-Gaussian random variable with parameter at most $9|\mathcal{A}|/(n\sigma_{\min}(\mathbf{\Sigma}))$. A Chernoff bound yields
    \begin{align*}
        {\mathbb{P}}\bigg[{\max_{j\in\mathcal{A}^c}}|\mathbf{W}_{j}^\top \mathbf{X}_{\mathcal{A}}(\mathbf{X}_{\mathcal{A}}^\top \mathbf{X}_{\mathcal{A}})^{-1}\mathbf{u}| \!\geq\! \frac{1-\bar{\alpha}}{2}\! \bigg] \!&\leq\! {\mathbb{P}}\bigg[{\max_{j\in\mathcal{A}^c}}|\mathbf{W}_{j}^\top \mathbf{X}_{\mathcal{A}}(\mathbf{X}_{\mathcal{A}}^\top \mathbf{X}_{\mathcal{A}})^{-1}\mathbf{u}| \!\geq\! \frac{1-\bar{\alpha}}{2}\bigg|\mathcal{E}^c(\mathbf{X}_{\mathcal{A}}) \!\bigg] \!+\! \mathbb{P}[\mathcal{E}(\mathbf{X}_{\mathcal{A}})] \\
        &\leq 2(d-|\mathcal{A}|)\exp\left(-\frac{(1-\bar{\alpha})^2 n  \sigma_{\min}(\mathbf{\Sigma})}{72|\mathcal{A}|}\right) + 2\exp(-n/2)\\
        &\leq 4d\exp\left(-\frac{(1-\bar{\alpha})^2 n  \sigma_{\min}(\mathbf{\Sigma})}{72|\mathcal{A}|}\right).
    \end{align*}
    The above probability is at most $\delta$ if $|\mathcal{A}| \leq (1-\bar{\alpha})^2\sigma_{\min}(\mathbf{\Sigma)}n/72\ln(4d/\delta)$. Thus, with probability at least $1-\delta$, $|\mathbf{X}_{j}^\top \mathbf{X}_{\mathcal{A}}(\mathbf{X}_{\mathcal{A}}^\top \mathbf{X}_{\mathcal{A}})^{-1}\mathbf{u}| \leq \frac{1}{2} + \frac{\bar{\alpha}}{2}$ for any $\mathbf{u}\in\mathbb{R}^{|\mathcal{A}|}$ such that $\|\mathbf{u}\|_\infty = 1$. 
\end{proof}

The above lemma applies to standard Gaussian matrices with $\mathbf{\Sigma} = \mathbf{I}$, in which case $\bar{\alpha} = 0$ and $\sigma_{\min}(\mathbf{\Sigma}) = 1$, and thus the mutual incoherence $\|\mathbf{X}_{\mathcal{A}}^+\mathbf{X}_{\mathcal{A}^c}\|_1$ is less than $1/2$ with high probability if $|\mathcal{A}| = O(n/\log{d})$.

Finally, we analyse the mutual overlap between $\mathbf{y}$ and $\mathbf{X}$. The next lemma states that the mutual overlap between $\mathbf{y}$ and $\mathbf{X}$ is $\Omega\Big(\frac{\|\mathbf{y}\|_2}{\|\mathbf{y}\|_1}\frac{1}{\sqrt{\smash[b]{\log{d}}}}\Big)$ with high probability if $|\mathcal{A}| = O(n)$, $\|\sqrt{\mathbf{\Sigma}}\|_1 = O(1)$, and the diagonal entries of the conditional covariance matrices $\mathbf{\Sigma}_{\mathcal{A}^c|\mathcal{A}}$ are bounded away from $0$ (this condition can be relaxed to a constant fraction of the diagonal entries be bounded away from~$0$). Hence, if $\|\mathbf{y}\|_1/\|\mathbf{y}\|_2 = O(1)$, e.g., in the case when $\mathbf{y}$ is sparse, then the mutual overlap is $\Omega(1/\sqrt{\log{d}})$, independent of $n$! 

\begin{lemma}\label{lem:bound_mutual_overlap}
    Let $\mathbf{y}\in\mathbb{R}^n$ and $\mathbf{X}\in\mathbb{R}^{n\times d}$ a random matrix from the $\mathbf{\Sigma}$-Gaussian ensemble, $\mathbf{\Sigma}\succ \mathbf{0}$. Given $\mathcal{A}\subseteq[d]$, suppose that
    \begin{align*}
        \mathbf{\Sigma}_{\mathcal{A}^c|\mathcal{A}} \succeq C_{\mathbf{\Sigma}}\mathbf{I} \quad\text{for}\quad C_{\mathbf{\Sigma}}>0,
    \end{align*}
    where $\mathbf{\Sigma}_{\mathcal{A}^c|\mathcal{A}} = \mathbf{\Sigma}_{\mathcal{A}^c\mathcal{A}^c} - \mathbf{\Sigma}_{\mathcal{A}^c\mathcal{A}}(\mathbf{\Sigma}_{\mathcal{A}\mathcal{A}})^{-1}\mathbf{\Sigma}_{\mathcal{A}\mathcal{A}^c}\succ \mathbf{0}$ is the conditional covariance matrix of $\mathbf{X}_{\mathcal{A}^c}|\mathbf{X}_{\mathcal{A}}$. Given $\delta\in(0,1)$, if $|\mathcal{A}| \leq \min\{n/2, n- 16\ln(3/\delta)\}$, then 
    \begin{align*}
        \mathbb{P}\left[\frac{\|\mathbf{X}^\top(\mathbf{I} - \mathbf{X}_{\mathcal{A}}\mathbf{X}_{\mathcal{A}}^+) \mathbf{y}\|_\infty}{\|\mathbf{X}\|_{\max}\|(\mathbf{I} - \mathbf{X}_{\mathcal{A}}\mathbf{X}_{\mathcal{A}}^+) \mathbf{y}\|_1} \geq \frac{(\delta/3)^{2/d}}{4\|\sqrt{\mathbf{\Sigma}}\|_1}\frac{\|\mathbf{y}\|_2}{\|\mathbf{y}\|_1}\sqrt{\frac{\pi C_{\mathbf{\Sigma}}}{\ln(6nd/\delta)}} \right] \geq 1-\delta.
    \end{align*}
\end{lemma}
\begin{proof}
    Let us start by bounding $\|\mathbf{X}\|_{\max}$. Similarly to \Cref{lem:condition_number}, $\|\mathbf{X}\|_{\max} \leq \|\sqrt{\mathbf{\Sigma}}\|_1\|\mathbf{Z}\|_{\max}$, where $\mathbf{Z}$ is drawn from the $\mathbf{I}$-Gaussian ensemble. By a Chernoff bound (\Cref{fact:Chernoff}) plus a union bound, $\mathbb{P}[\|\mathbf{Z}\|_{\max} \leq \sqrt{\ln(6nd/\delta)}] \geq 1-\delta/3$, which implies that $\mathbb{P}[\|\mathbf{X}\|_{\max} \leq \sqrt{\ln(6nd/\delta)}\|\sqrt{\mathbf{\Sigma}}\|_1] \geq 1-\delta/3$.
    
    We now move on to the remaining part of the expression. First notice that
    \begin{align*}
        \|\mathbf{X}^\top(\mathbf{I} - \mathbf{X}_{\mathcal{A}}\mathbf{X}_{\mathcal{A}}^+) \mathbf{y}\|_\infty = \max_{j\in\mathcal{A}^c}|\mathbf{X}_j^\top(\mathbf{I} - \mathbf{X}_{\mathcal{A}}\mathbf{X}_{\mathcal{A}}^+) \mathbf{y}|.
    \end{align*}
    Fix $j\in\mathcal{A}^c$. Again by \Cref{eq:linear_prediction}, we can write
    \begin{align*}
        \mathbf{X}_j^\top = \mathbf{\Sigma}_{j\mathcal{A}}(\mathbf{\Sigma}_{\mathcal{A}\mathcal{A}})^{-1}\mathbf{X}_{\mathcal{A}}^\top + \mathbf{W}_j^\top,
    \end{align*}
    where $\mathbf{W}_j\in\mathbb{R}^n$ is a vector with i.i.d.\ $\mathcal{N}(0,[\mathbf{\Sigma}_{\mathcal{A}^c|\mathcal{A}}]_{jj})$ entries. Therefore,
    \begin{align*}
        |\mathbf{X}_j^\top(\mathbf{I} - \mathbf{X}_{\mathcal{A}}\mathbf{X}_{\mathcal{A}}^+) \mathbf{y}| = |\mathbf{W}_j^\top(\mathbf{I} - \mathbf{X}_{\mathcal{A}}\mathbf{X}_{\mathcal{A}}^+) \mathbf{y}|,
    \end{align*}
    where we used that $\mathbf{X}_{\mathcal{A}}^\top(\mathbf{I} - \mathbf{X}_{\mathcal{A}}\mathbf{X}_{\mathcal{A}}^+) = \mathbf{0}$.
    Since $\mathbf{W}_j$ is independent of $\mathbf{X}_{\mathcal{A}}$, then $\mathbf{W}_j^\top(\mathbf{I} - \mathbf{X}_{\mathcal{A}}\mathbf{X}_{\mathcal{A}}^+) \mathbf{y}$ is a Gaussian random variable with variance 
    \begin{align*}
        [\mathbf{\Sigma}_{\mathcal{A}^c|\mathcal{A}}]_{jj}\|(\mathbf{I} - \mathbf{X}_{\mathcal{A}}\mathbf{X}_{\mathcal{A}}^+) \mathbf{y}\|_2^2 \geq C_{\mathbf{\Sigma}}\|(\mathbf{I} - \mathbf{X}_{\mathcal{A}}\mathbf{X}_{\mathcal{A}}^+) \mathbf{y}\|_2^2
    \end{align*}
    since $[\mathbf{\Sigma}_{\mathcal{A}^c|\mathcal{A}}]_{jj} \geq C_{\mathbf{\Sigma}}$. Furthermore, since $\mathbf{I} - \mathbf{X}_{\mathcal{A}}\mathbf{X}_{\mathcal{A}}^+$ is a random projection, i.e., a projection onto a random $(n-|\mathcal{A}|)$-dimensional subspace uniformly distributed in the Grassmann manifold $G_{n,n-|\mathcal{A}|}$, then L\'evy's lemma (\Cref{fact:levy_lemma}) yields\footnote{The function $\mathbf{u}\in\mathbb{R}^n\mapsto\|\mathbf{P}\mathbf{u}\|_2$ is $1$-Lipschitz for any orthogonal projector $\mathbf{P}\in\mathbb{R}^{n\times n}$.} (see also~\cite[Lemma~5.3.2]{Vershynin_2018})
    \begin{align*}
        \mathbb{P}\left[\|(\mathbf{I} - \mathbf{X}_{\mathcal{A}}\mathbf{X}_{\mathcal{A}}^+) \mathbf{y}\|_2 \geq \frac{1}{2}\sqrt{\frac{n-|\mathcal{A}|}{n}}\|\mathbf{y}\|_2 \right] \geq 1 - e^{-(n-|\mathcal{A}|)/16}.
    \end{align*}
    Define the event $\mathcal{E}(\mathcal{A})\triangleq \big\{\|(\mathbf{I} - \mathbf{X}_{\mathcal{A}}\mathbf{X}_{\mathcal{A}}^+) \mathbf{y}\|_2 \geq\sqrt{1-|\mathcal{A}|/n}\|\mathbf{y}\|_2/2\big\}$ and let $\operatorname{erf}(t) = \frac{2}{\sqrt{\pi}}\int_{0}^t e^{-z^2}\text{d}z$, $t\geq 0$, be the error function. Then, for $t\in[0,1]$ and conditioned on the event $\mathcal{E}(\mathcal{A})$ happening,
    \begin{align*}
        \mathbb{P}\left[\frac{|\mathbf{W}_j^\top(\mathbf{I} - \mathbf{X}_{\mathcal{A}}\mathbf{X}_{\mathcal{A}}^+) \mathbf{y}|}{\|(\mathbf{I} - \mathbf{X}_{\mathcal{A}}\mathbf{X}_{\mathcal{A}}^+) \mathbf{y}\|_1} \leq t\sqrt{\frac{\pi C_{\mathbf{\Sigma}}(n-|\mathcal{A}|)}{8n}}\frac{\|\mathbf{y}\|_2}{\|\mathbf{y}\|_1} ~\bigg|~ \mathcal{E}(\mathcal{A})\right] \leq \operatorname{erf}\left(\frac{t\sqrt{\pi}\|(\mathbf{I} - \mathbf{X}_{\mathcal{A}}\mathbf{X}_{\mathcal{A}}^+) \mathbf{y}\|_1}{2\|\mathbf{y}\|_1}\right) \leq t,
    \end{align*}
    using that $\|(\mathbf{I} - \mathbf{X}_{\mathcal{A}}\mathbf{X}_{\mathcal{A}}^+) \mathbf{y}\|_1 \leq \|\mathbf{y}\|_1$ and $\operatorname{erf}(t) \leq 2t/\sqrt{\pi}$. By the independence of all $\mathbf{W}_j$, $j\in\mathcal{A}^c$,
    \begin{align}\label{eq:bound_probability}
        \mathbb{P}\left[\max_{j\in\mathcal{A}^c}\frac{|\mathbf{W}_j^\top(\mathbf{I} - \mathbf{X}_{\mathcal{A}}\mathbf{X}_{\mathcal{A}}^+) \mathbf{y}|}{\|(\mathbf{I} - \mathbf{X}_{\mathcal{A}}\mathbf{X}_{\mathcal{A}}^+) \mathbf{y}\|_1} \leq t\sqrt{\frac{\pi C_{\mathbf{\Sigma}}(n-|\mathcal{A}|)}{8n}}\frac{\|\mathbf{y}\|_2}{\|\mathbf{y}\|_1} \right] &\leq t^{d-|\mathcal{A}|} + e^{-(n-|\mathcal{A}|)/16}.
    \end{align}
    By taking $t = (\delta/3)^{1/(d-n)}$ and $|\mathcal{A}| \leq n - 16\ln(3/\delta)$, then the above probability is at most $2\delta/3$. Finally, we put \Cref{eq:bound_probability} with $n-|\mathcal{A}| \geq n/2$ together with the bound on $\|\mathbf{X}\|_{\max}$.
\end{proof}


\begin{remark}
    The results from this section apply to a fixed active set $\mathcal{A}$. It is possible to bound $\|\mathbf{X}\|_{\max}\|\mathbf{X}^+_{\mathcal{A}}\|_2$, the mutual incoherence $\alpha_{\mathcal{A}}$, and the mutual overlap $\gamma_{\mathcal{A}}$ for several sets $\mathcal{A}$ by taking a union bound on the failure probability in {\rm \Cref{lem:condition_number,lem:bound_mutual_incoherence,lem:bound_mutual_overlap}}. Since the approximate regularisation path $\widetilde{\mathcal{P}}$ has at most $O(n)$ kinks for when $\mathbf{X}$ is a random Gaussian matrix ({\rm \cref{fact:continuous_distribution}}), {\rm \Cref{alg:classical_quantum,alg:approximate_classical}} need only to consider $O(n)$ different active sets. Therefore, we only need to consider a union bound over $O(n)$ different sets. 
\end{remark}

\section{Query lower bounds for Lasso}

In this section we prove classical and quantum query lower bounds for Lasso. Our results follow from the fact that an algorithm for Lasso can be used to solve specific set-finding problems, which were studied by Chen and de Wolf~\cite{chen2021quantum} (and also employed to derive their lower bounds).

\subsection{Set-finding problems}

We start by defining the set-finding problems of interest, which basically hide a set $W$ in the columns of a matrix $\mathbf{X}$.
\begin{definition}[Exact and approximate set-finding problems $\mathsf{ESF}^{d,N}_{w,p}$ and $\mathsf{ASF}^{d,N}_{w,p}$]
    Let $p\in(0,1/2)$ and $W\subset[d]$ of size $w$. Consider a matrix $\mathbf{X}\in\{-1,1\}^{N\times d}$ such that
    \begin{align*}
        \sum_{i=1}^N X_{ij} = 2pN \quad\text{if}~j\in W \qquad\text{and}\qquad \sum_{i=1}^N X_{ij} = 0 \quad\text{if}~j\notin W.
    \end{align*}
    In the \emph{exact set-finding problem} $\mathsf{ESF}^{d,N}_{w,p}$ the task is to find $W$. In the \emph{approximate set-finding problem} $\mathsf{ASF}^{d,N}_{w,p}$ the task is to find $\widetilde{W}\subseteq[d]$ such that $|W\Delta \widetilde{W}| \leq w/5$.
\end{definition}
%
Chen and de Wolf~\cite{chen2021quantum} proved the following query lower bound on any algorithms that solve $\mathsf{ESF}^{d,N}_{w,p}$ and $\mathsf{ASF}^{d,N}_{w,p}$.\footnote{In~\cite{chen2021quantum}, Chen and de Wolf define the approximate set-finding problem as \emph{worst-case set-finding problem} $\mathsf{WSF}_{d,w,p,N}$. We decided to change the name slightly in order to draw a clearer parallel with $\mathsf{ESF}^{d,n}_{w,p}$.}
\begin{fact}[{\cite[Corollary~4.9]{chen2021quantum}}]\label{fact:quantum_asf_lower_bound}
    Let $N\in\mathbb{N}$ be even and $p\in(0,1/2)$ an integer multiple of $1/N$. Every bounded-error quantum algorithm that solves the approximate set-finding problem $\mathsf{ASF}^{d,N}_{w,p}$ with matrix $\mathbf{X}\in\{-1,1\}^{N\times d}$ uses $\Omega(\sqrt{wd}/p)$ quantum queries to $\mathbf{X}$.
\end{fact}

\begin{fact}[{\cite[Theorem~B.5]{chen2021quantum}}]\label{fact:classical_esf_lower_bound}
    Let $N\in\mathbb{N}$ be even, $p\in(1/\sqrt{N},1/2)$ an integer multiple of $1/N$, and $w=1/2p$. Every classical algorithm that solves the exact set-finding problem $\mathsf{ESF}^{d,N}_{w,p}$ with matrix $\mathbf{X}\in\{-1,1\}^{N\times d}$ uses $\Omega((d-w)/p^2)$ queries to $\mathbf{X}$.
\end{fact}

We will also need a distributional version of the approximate set-finding problem.
\begin{definition}[Distributional set-finding problem $\mathsf{DSF}^{d,n}_{w,p}$]\label{def:distributional_set_finding_problem}
    Let $p\in(0,1/2)$, $n\in\mathbb{N}$, and $W\subset[d]$ of size $w$. Consider the distribution $\mathcal{D}_{p,W}$ over $(\mathbf{x},y)\in\{-1,1\}^d\times\{-1,1\}$ defined as follows:
    \begin{align*}
        \operatorname{Pr}[x_j = 1] =
        \begin{cases}
            \frac{1}{2} &\text{if}~j\notin W,\\
            \frac{1}{2} + p &\text{if}~j\in W,
        \end{cases}\qquad\text{and}\qquad
        \operatorname{Pr}[y=1] = 1.
    \end{align*}
    In the \emph{distributional set-finding problem} $\mathsf{DSF}^{d,n}_{w,p}$ the task is to find $\widetilde{W}\subset[d]$ such that $|\widetilde{W}\Delta W| \leq w/5$ given $n$ samples from $\mathcal{D}_{p,W}$.
\end{definition}

Chen and de Wolf proved that one can convert an instance of $\mathsf{ASF}^{d,N}_{w,p}$ (to be viewed as a worst-case instance) to an instance of $\mathsf{DSF}^{d,n}_{w,p}$ (to be viewed as an average-case instance).
\begin{fact}[{\cite[Theorem~4.10]{chen2021quantum}}]\label{fact:reduction}
    Let $N\in\mathbb{N}$ be even, $n\in\mathbb{N}$, $p\in(0,1/2)$ an integer multiple of $1/n$, and $w\in(2,d/2)$ a natural number. Let $\mathbf{X}\in\{-1,1\}^{N\times d}$ be a valid input to $\mathsf{ASF}^{d,N}_{w,p}$ and let $W\subset[d]$ be the set of indices of columns of $\mathbf{X}$ whose entries add to $2pn$. Let $\mathbf{R}\in[n]^{n\times d}$ be a matrix whose entries are i.i.d.\ samples from the uniform distribution over $[n]$. Define $\mathbf{Z}\in\{-1,1\}^{n\times d}$ as $Z_{ij} = X_{R_{ij},j}$. Then the $n$ vectors $(\mathbf{Z}_{i\cdot},1)\in\{-1,1\}^{d+1}$, where $\mathbf{Z}_{i\cdot}$ is the $i$-th row of $\mathbf{Z}$ and $i\in[n]$, are i.i.d.\ samples from the distribution $\mathcal{D}_{p,W}$ defining $\mathsf{DSF}^{d,n}_{w,p}$.
\end{fact}

Finally, the authors also proved that an algorithm for $\mathsf{DSF}^{d,n}_{w,p}$ can be used to solve $\mathsf{ESF}^{d,N}_{w,p}$.
\begin{fact}[{\cite[Theorem~B.6]{chen2021quantum}}]\label{fact:est_to_asf_reduction}
    Let $N\in\mathbb{N}$ be even, $p\in(1/\sqrt{N},1/4)$ be an integer multiple of $1/N$, and $w = 1/2p$. Suppose $\mathscr{A}$ is an algorithm for $\mathsf{DSF}^{d,n}_{w,p}$ that outputs $\widetilde{W}\subset[d]$ for every $W\subset[d]$ of size $w$ such that $|\widetilde{W}\Delta W| \leq w/200$ with probability at least $1 - c/\log{d}$ for a small enough constant $c>0$. Then there is an algorithm that solves $\mathsf{ESF}^{d,N}_{w,p}$ with success probability at least $99/100$ by querying $\mathscr{A}$ a number of $O(\log{d})$ times.
\end{fact}

\subsection{Solving set-finding problems with a Lasso solver}

We shall now prove that an algorithm for Lasso can be used to find a good approximation to the hidden set $W$, after which the sought-out lower bounds will follow from \Cref{fact:quantum_asf_lower_bound,fact:classical_esf_lower_bound}. In order to do so, we must first consider an expected version of the Lasso problem defined as
\begin{align*}
    \hat{\bfbeta} \in \argmin_{\bfbeta\in\mathbb{R}^d}\overline{\mathcal{L}}^{(\lambda)}_{\mathcal{D}}(\bfbeta) \quad\text{where}\quad \overline{\mathcal{L}}^{(\lambda)}_{\mathcal{D}}(\bfbeta) \triangleq \frac{1}{2}\operatorname*{\mathbb{E}}_{(\mathbf{x},y)\sim\mathcal{D}}[(\mathbf{x}^\top \bfbeta - y)^2] + \lambda\|\bfbeta\|_1 \quad\text{for}~\lambda > 0.
\end{align*}
Here, $\mathcal{D}$ is a distribution over $(\mathbf{x},y)\in \mathbb{R}^d\times \mathbb{R}$. We start by studying the expected Lasso $\overline{\mathcal{L}}^{(\lambda)}_{\mathcal{D}_{p,W}}(\bfbeta)$ relative to the distribution $\mathcal{D}_{p,W}$ from \Cref{def:distributional_set_finding_problem}.
\begin{lemma}\label{lem:properties_expected_lasso}
    Let $p\in(0,1/2)$ and $W\subset[d]$ of size $w$. Consider the distribution $\mathcal{D}_{p,W}$ over $\{-1,1\}^d\times\{-1,1\}$ from {\rm \Cref{def:distributional_set_finding_problem}}. Then the unique minimiser of the expected Lasso $\overline{\mathcal{L}}^{(\lambda)}_{\mathcal{D}_{p,W}}(\bfbeta)$ is the vector $\hat{\bfbeta}\in\mathbb{R}^d$ given by
    \begin{align*}
        \hat{\beta}_j = \begin{cases}
             \frac{2p}{1 + \lambda + 4p^2(w-1)} &\text{if}~j\in W,\\
             0 &\text{if}~j\notin W.
        \end{cases}
    \end{align*}
    Moreover, $\nabla^2 \overline{\mathcal{L}}^{(\lambda)}_{\mathcal{D}_{p,W}}(\bfbeta) \succeq (1-4p^2)\mathbf{I}_d$ for all $\bfbeta\in\mathbb{R}^d$, $\nabla \overline{\mathcal{L}}^{(\lambda)}_{\mathcal{D}_{p,W}}(\hat{\bfbeta}) = \mathbf{0}$, and 
    \begin{align*}
        \overline{\mathcal{L}}^{(\lambda)}_{\mathcal{D}_{p,W}}(\hat{\bfbeta}) = \frac{1}{2} + \frac{2pw(\lambda - 2p)}{1+\lambda + 4p^2(w-1)} + \frac{2p^2 w(1+4p^2(w-1))}{(1+\lambda+4p^2(w-1))^2} \leq \frac{1}{2} + \frac{2pw(\lambda - p)}{1 + \lambda + 4p^2(w-1)}.
    \end{align*}
\end{lemma}
\begin{proof}
    First notice that, for $i\neq j$, $\mathbb{E}_{(\mathbf{x},y)\sim\mathcal{D}_{p,W}}[x_i x_j] = 4p^2$ if $i,j\in W$ and $\mathbb{E}_{(\mathbf{x},y)\sim\mathcal{D}_{p,W}}[x_i x_j] = 0$ otherwise. Let $\hat{\bfbeta} \in \argmin_{\bfbeta\in\mathbb{R}^d}\overline{\mathcal{L}}^{(\lambda)}_{\mathcal{D}_{p,W}}(\bfbeta)$. Thus
    \begin{align}
        \overline{\mathcal{L}}^{(\lambda)}_{\mathcal{D}_{p,W}}(\bfbeta) &= \frac{1}{2}\operatorname*{\mathbb{E}}_{(\mathbf{x},y)\sim\mathcal{D}_{p,W}}\left[y^2 - 2y\sum_{i\in[d]}x_i\beta_i + \sum_{i,j\in[d]}x_i x_j \beta_i \beta_j \right] + \lambda\|\bfbeta\|_1 \nonumber\\
        &= \frac{1}{2}\left(1 - 4p\sum_{i\in W}\beta_i + \sum_{i\in [d]}\beta_i^2 + 4p^2\sum_{i \neq j\in W}\beta_i\beta_j\right) + \lambda\|\bfbeta\|_1 \label{eq:aux_1}
    \end{align}
    and
    \begin{align*}
        \nabla\overline{\mathcal{L}}^{(\lambda)}_{\mathcal{D}_{p,W}}(\bfbeta) &= \operatorname*{\mathbb{E}}_{(\mathbf{x},y)\sim\mathcal{D}_{p,W}}[(\mathbf{x}^\top\bfbeta - y)\mathbf{x}] + \lambda \mathbf{s} \implies \\ \nabla_j \overline{\mathcal{L}}^{(\lambda)}_{\mathcal{D}_{p,W}}(\bfbeta) &= \begin{cases}
            \beta_j + 4p^2\sum_{i\neq j\in W} \beta_i - 2p + \lambda s_j &\text{if}~j\in W,\\
            \beta_j +  \lambda s_j &\text{if}~j\notin W,
        \end{cases}
    \end{align*}
    where $\mathbf{s}\in\partial\|\bfbeta\|_1$ is a sub-gradient. Since $\hat{\bfbeta}$ minimises $\overline{\mathcal{L}}^{(\lambda)}_{\mathcal{D}_{p,W}}(\bfbeta)$, then $\nabla\overline{\mathcal{L}}^{(\lambda)}_{\mathcal{D}_{p,W}}(\hat{\bfbeta}) = \mathbf{0}$. Moreover, if we arrange the indices such that $\{1,\dots,w\}= W$ and $\{w+1,\dots, d\}=[d]\setminus W$, then
    \begin{align*}
        \nabla^2 \overline{\mathcal{L}}^{(\lambda)}_{\mathcal{D}_{p,W}}(\bfbeta) = \operatorname*{\mathbb{E}}_{(\mathbf{x},y)\sim\mathcal{D}_{p,W}}[\mathbf{x}^\top\mathbf{x}] = \begin{pmatrix*}
            4p^2 \mathbf{J}_w + (1-4p^2)\mathbf{I}_w & \mathbf{0}_{w\times \bar{w}}\\
            \mathbf{0}_{\bar{w}\times w} & 2\mathbf{I}_{\bar{w}}
        \end{pmatrix*} \succeq (1-4p^2)\mathbf{I}_d,
    \end{align*}
    where $\bar{w} \triangleq d - w$ and $\mathbf{J}_w$ is the all-ones matrix. We see that $\nabla^2 \overline{\mathcal{L}}^{(\lambda)}_{\mathcal{D}_{p,W}}(\bfbeta)$ is a constant matrix independent of $\bfbeta$. Regarding the optimal solution $\hat{\bfbeta}$, since $s_j = 0$ if $j\notin W$ and $s_j = \operatorname{sign}(\hat{\beta}_j)$ if $j\in W$, then
    \begin{align*}
        \begin{cases}
            (1-4p^2)\hat{\beta}_j + \lambda\operatorname{sign}(\hat{\beta}_j) - 2p + 4p^2\sum_{i\in W} \hat{\beta}_i = 0 &\text{if}~ j\in W, \\
            \hat{\beta}_j = 0 &\text{if}~j\notin W,
        \end{cases} \implies 
        \begin{cases}
            \hat{\beta}_j = \frac{2p}{1 + \lambda + 4p^2(w-1)} &\text{if}~j\in W,\\
            \hat{\beta}_j = 0 &\text{if}~j\notin W.
        \end{cases}
    \end{align*}
    Since the Hessian $\nabla^2 \overline{\mathcal{L}}^{(\lambda)}_{\mathcal{D}_{p,W}}(\bfbeta)$ is positive definite, the above solution is unique. Finally, the expression for $\overline{\mathcal{L}}^{(\lambda)}_{\mathcal{D}_{p,W}}(\hat{\bfbeta})$ follows from \Cref{eq:aux_1} and the expression for $\hat{\bfbeta}$.
\end{proof}

We now relate the entries of an approximate solution of the expected Lasso problem to the elements of the hidden set $W$.
\begin{lemma}\label{lem:approximate_lasso_set_finding}
    Let $\epsilon,b\geq 0$, $\lambda>0$, $w\in\mathbb{N}$, $p=(0,1/2)$, and $W\subset[d]$ of size $w$. Let $\widetilde{\bfbeta}\in\mathbb{R}^d$ be an approximate solution of the expected Lasso problem $\overline{\mathcal{L}}^{(\lambda)}_{\mathcal{D}_{p,W}}(\bfbeta)$ with error $\lambda \epsilon\|\widetilde{\bfbeta}\|_1 + b\epsilon$. Then
    \begin{align*}
        \sum_{j\in W}\left(\widetilde{\beta}_j - \frac{2p}{1 + \lambda + 4p^2(w-1)} \right)^2 + \sum_{j\notin W} \widetilde{\beta}_j^2 \leq \frac{\epsilon}{(1-\epsilon)(1-4p^2)}\left(1 + 2b + \frac{4pw(\lambda-p)}{1+\lambda+4p^2(w-1)}\right).
    \end{align*}
\end{lemma}
\begin{proof}
    Let $\hat{\bfbeta} = \argmin_{\bfbeta\in\mathbb{R}^d}\overline{\mathcal{L}}^{(\lambda)}_{\mathcal{D}_{p,W}}(\bfbeta)$ from \Cref{lem:properties_expected_lasso}. First notice that
    \begin{align}
        \overline{\mathcal{L}}^{(\lambda)}_{\mathcal{D}_{p,W}}(\widetilde{\bfbeta}) - \overline{\mathcal{L}}^{(\lambda)}_{\mathcal{D}_{p,W}}(\hat{\bfbeta}) \leq \lambda\epsilon\|\widetilde{\bfbeta}\|_1 + b\epsilon &\leq \epsilon\overline{\mathcal{L}}^{(\lambda)}_{\mathcal{D}_{p,W}}(\widetilde{\bfbeta}) + b\epsilon \implies \nonumber \\
        (1-\epsilon)\big(\overline{\mathcal{L}}^{(\lambda)}_{\mathcal{D}_{p,W}}(\widetilde{\bfbeta}) - \overline{\mathcal{L}}^{(\lambda)}_{\mathcal{D}_{p,W}}(\hat{\bfbeta})\big) &\leq \epsilon \overline{\mathcal{L}}^{(\lambda)}_{\mathcal{D}_{p,W}}(\hat{\bfbeta}) + b\epsilon. \label{eq:aux_2}
    \end{align}
    By expanding $\overline{\mathcal{L}}^{(\lambda)}_{\mathcal{D}_{p,W}}(\widetilde{\bfbeta})$ around $\hat{\bfbeta}$,
    \begin{align*}
        \frac{\epsilon}{1-\epsilon}(\overline{\mathcal{L}}^{(\lambda)}_{\mathcal{D}_{p,W}}(\hat{\bfbeta}) + b) &\geq \overline{\mathcal{L}}^{(\lambda)}_{\mathcal{D}_{p,W}}(\widetilde{\bfbeta}) - \overline{\mathcal{L}}^{(\lambda)}_{\mathcal{D}_{p,W}}(\hat{\bfbeta}) \tag{by \Cref{eq:aux_2}}\\
        &= \frac{1}{2}(\widetilde{\bfbeta} - \hat{\bfbeta})^\top \nabla^2 \overline{\mathcal{L}}^{(\lambda)}_{\mathcal{D}_{p,W}}(\hat{\bfbeta})(\widetilde{\bfbeta} - \hat{\bfbeta}) \tag{by $\nabla \overline{\mathcal{L}}^{(\lambda)}_{\mathcal{D}_{p,W}}(\hat{\bfbeta}) = \mathbf{0}$}\\
        &\geq \frac{1-4p^2}{2}\|\widetilde{\bfbeta} - \hat{\bfbeta}\|_2^2. \tag{by $\nabla^2 \overline{\mathcal{L}}^{(\lambda)}_{\mathcal{D}_{p,W}}(\hat{\bfbeta}) \succeq (1-4p^2)\mathbf{I}_d$}
    \end{align*}
    This finishes the proof after inserting the expression for $\overline{\mathcal{L}}^{(\lambda)}_{\mathcal{D}_{p,W}}(\hat{\bfbeta})$ from \Cref{lem:properties_expected_lasso}.
\end{proof}

We are now ready to prove that approximate solutions to expected Lasso problems can be used to find good approximations to the hidden set $W$.
\begin{theorem}\label{thr:lasso_algorithm_set_finding}
    Let $\lambda > \sqrt{5/d}$, $\epsilon\in (1/d\lambda^2,\min\{1/\lambda^2,1/5\})$, $w=\lfloor1/(\lambda^2\epsilon)\rfloor$, $p=1/(2\lfloor 1/(\lambda^2\epsilon)\rfloor)$, and $W\subset[d]$ of size $w$. Let $\widetilde{\bfbeta}\in\mathbb{R}^d$ be an approximate solution of the expected Lasso problem $\overline{\mathcal{L}}^{(\lambda)}_{\mathcal{D}_{p,W}}(\bfbeta)$ with error $\lambda \epsilon\|\widetilde{\bfbeta}\|_1/400 + \epsilon/400$. Then the set $\widetilde{W} \triangleq \{i\in[d]:|\widetilde{\beta}_i| \geq \lambda\epsilon/2\}$ satisfies $|\widetilde{W}\Delta W| \leq w/5$.
\end{theorem}
\begin{proof}
    Let $\nu \triangleq \frac{2p}{1 + \lambda + 4p^2(w-1)}$. Since $\widetilde{\bfbeta}$ minimises $\overline{\mathcal{L}}^{(\lambda)}_{\mathcal{D}_{p,W}}(\bfbeta)$ up to error $\lambda\epsilon\|\widetilde{\bfbeta}\|_1$, \Cref{lem:approximate_lasso_set_finding} yields
    \begin{align*}
        \sum_{j\in W}(\widetilde{\beta}_j - \nu)^2 + \sum_{j\notin W} \widetilde{\beta}_j^2 \leq \frac{\epsilon/400}{(1-\epsilon/400)(1-4p^2)}\left(2 + \frac{4pw(\lambda-p)}{1+\lambda+4p^2(w-1)}\right)
        \leq \frac{\epsilon}{60},
    \end{align*}
    where we used that $\epsilon\leq 1/5$, $w\geq 1$, $p\leq 1/8$, and that $2pw \leq 1$. Therefore,
    \begin{itemize}
        \item at most $w/10$ many $j\in [d]\setminus W$ have $|\widetilde{\beta}_j| \geq \sqrt{\frac{\epsilon}{60}\frac{10}{w}} = \sqrt{\frac{1}{6}\frac{\epsilon}{\lfloor1/\lambda^2\epsilon\rfloor}}$;
        \item at least $w-w/10$ many $j\in W$ have $\widetilde{\beta}_j \geq \nu - \sqrt{\frac{1}{6}\frac{\epsilon}{\lfloor1/\lambda^2\epsilon\rfloor}}$.
    \end{itemize}
    Note that $\nu - \sqrt{\frac{1}{6}\frac{\epsilon}{\lfloor1/\lambda^2\epsilon\rfloor}} = \frac{1/\lfloor1/\lambda^2\epsilon\rfloor}{1+\lambda + 1/\lfloor1/\lambda^2\epsilon\rfloor - 1/\lfloor1/\lambda^2\epsilon\rfloor^2} - \sqrt{\frac{1}{6}\frac{\epsilon}{\lfloor1/\lambda^2\epsilon\rfloor}} \geq \frac{\lambda\epsilon}{2} \geq \sqrt{\frac{1}{6}\frac{\epsilon}{\lfloor1/\lambda^2\epsilon\rfloor}}$ since it holds that $\lambda^2\epsilon \leq \frac{1}{\lfloor1/\lambda^2\epsilon\rfloor} \leq \frac{\lambda^2\epsilon}{1-\lambda^2\epsilon} \leq \frac{3\lambda^2\epsilon}{4}$. This means that the set $\widetilde{W} \triangleq \{i\in[d]:|\widetilde{\beta}_i| \geq \lambda\epsilon/2\}$ omits at most $w/10$ indices $j\in W$ and includes at most $w/10$ indices $j\in[d]\setminus W$. Therefore $|\widetilde{W}\Delta W| \leq w/5$.
\end{proof}

The above result tells us that a Lasso algorithm that finds an approximate expected Lasso solution with respect to $\mathcal{D}_{p,W}$ can also find a set $\widetilde{W}\subset[d]$ such that $|\widetilde{W}\Delta W| \leq w/5$.

\subsection{Quantum lower bound}

We now finally piece all the results from the previous sections together in order to obtain query lower bounds for Lasso algorithms. The first step is to relate the expected Lasso $\overline{\mathcal{L}}^{(\lambda)}_{\mathcal{D}_{p,W}}(\bfbeta)$ with the usual Lasso $\mathcal{L}^{(\lambda)}_{\mathbf{X},\mathbf{y}}(\bfbeta)$ from \Cref{eq:lasso_problem} considered throughout the paper. For such, we consider the normalised Lasso problem with inputs $(\mathbf{X},\mathbf{y})\in\mathbb{R}^{n\times d}\times\mathbb{R}^n$ defined as
\begin{align*}
    \hat{\bfbeta} \in \argmin_{\bfbeta \in \mathbb{R}^d} \overline{\mathcal{L}}^{(\lambda)}_{\mathbf{X},\mathbf{y}}(\bfbeta) \quad\text{where}\quad \overline{\mathcal{L}}^{(\lambda)}_{\mathbf{X},\mathbf{y}}(\bfbeta) \triangleq \frac{1}{2n}\|\mathbf{y}-\mathbf{X}{\bfbeta}\|_2^2 + \lambda \|\bfbeta\|_1 \quad \text{for}~ \lambda >0.
\end{align*}
\begin{fact}[{\cite[Theorem~11.16]{mohri2018foundations}}]
    Let $\mathcal{D}$ be a distribution over $[-1,1]^d\times[-1,1]$. Given a sample set $\{(\mathbf{x}_i,y_i)\}_{i\in[n]}$ containing $n$ i.i.d.\ samples from $\mathcal{D}$, let $\mathbf{X}\in[-1,1]^{n\times d}$ and $\mathbf{y}\in[-1,1]^n$ be such that the $i$-th row of $\mathbf{X}$ is $\mathbf{x}_i$ and the $i$-th entry of $\mathbf{y}$ is $y_i$. For all $\delta >0$ and $\bfbeta\in\mathbb{R}^d$ such that $\|\bfbeta\|_1 \leq \Lambda$, with probability at least $1-\delta$ over the choice of $(\mathbf{X},\mathbf{y})$,
    \begin{align*}
        \overline{\mathcal{L}}^{(\lambda)}_{\mathcal{D}}(\bfbeta) - \overline{\mathcal{L}}^{(\lambda)}_{\mathbf{X},\mathbf{y}}(\bfbeta) \leq \Lambda(1+\Lambda)\sqrt{\frac{2\ln(2d)}{n}} + (1+\Lambda)^2\sqrt{\frac{\ln(1/\delta)}{8n}}.
    \end{align*}
\end{fact}

Using a reduction from $\mathsf{ASF}^{d,N}_{w,p}$ to algorithms for Lasso, we now prove our quantum lower bound.
\begin{theorem}\label{thr:lower_bound}
    Let $n < d$ such that $n = \Omega(\log{d})$, $\lambda = \Omega(n)$, and $\epsilon \in (0,1/5)$ such that $\epsilon = \Omega(\sqrt{n^3\log{d}}/\lambda^2)$. There is $\mathbf{X}\in\{-1,1\}^{n\times d}$ such that every bounded-error quantum algorithm that computes an $\epsilon\lambda\|\widetilde{\bfbeta}\|_1$-minimiser for the Lasso $\mathcal{L}^{(\lambda)}_{\mathbf{X},1^n}(\bfbeta)$ uses 
    \begin{align*}
        \Omega\left(\frac{n^3\sqrt{d}}{\lambda^3\epsilon^{3/2}} \right) \quad\text{quantum queries to}~\mathbf{X}.
    \end{align*}
\end{theorem}
\begin{proof}
    Let $w=\lfloor1/(\overline{\lambda}^2\epsilon)\rfloor$ and $p=1/(2\lfloor 1/(\overline{\lambda}^2\epsilon)\rfloor)$. Consider a valid input $\mathbf{X}'\in\{-1,1\}^{N\times d}$ for $\mathsf{ASF}^{d,N}_{w,p}$, and let $W\subset[d]$ be the set of $w$ indices of columns of $\mathbf{X}'$ whose entries add up to $2pN$. According to \Cref{fact:reduction}, we can obtain from $\mathbf{X}'$ a number of $n=\Omega(\log{d})$ i.i.d.\ samples $\{(\mathbf{x}_i,1)\}_{i\in[n]}$ from the distribution $\mathcal{D}_{p,W}$. Let $\mathbf{X}\in\{-1,1\}^{n\times d}$ be the matrix whose rows are $\{\mathbf{x}_i\}_{i\in[n]}$. For a small enough constant $c\in(0,1)$, suppose we have a quantum Lasso algorithm which returns a $c\lambda\epsilon\|\widetilde{\bfbeta}\|_1$-minimiser $\widetilde{\bfbeta}\in\mathbb{R}^d$ with high probability for the Lasso $\mathcal{L}^{(\lambda)}_{\mathbf{X},1^n}(\bfbeta)$ with inputs $(\mathbf{X},\mathbf{y}=1^n)$. The vector $\widetilde{\bfbeta}$ is also a $c\overline{\lambda}\epsilon\|\widetilde{\bfbeta}\|_1$-minimiser to the normalised Lasso $\overline{\mathcal{L}}^{(\overline{\lambda})}_{\mathbf{X},1^n}(\bfbeta)$ where $\overline{\lambda} = \lambda/n$. Note that $\overline{\lambda} = \Omega(1)$. Consider now the expected Lasso $\overline{\mathcal{L}}^{(\overline{\lambda})}_{\mathcal{D}_{p,W}}(\bfbeta)$. According to \Cref{lem:properties_expected_lasso}, the solution $\hat{\bfbeta}\in\argmin_{\bfbeta\in\mathbb{R}^d}\overline{\mathcal{L}}^{(\overline{\lambda})}_{\mathcal{D}_{p,W}}(\bfbeta)$ is
    \begin{align}\label{eq:optimal_solution}
        \hat{\beta}_j = \begin{cases}
             \frac{2p}{1 + \overline{\lambda} + 4p^2(w-1)} &\text{if}~j\in W,\\
             0 &\text{if}~j\notin W.
        \end{cases}
    \end{align}
    For $n = \Omega(\epsilon^{-2}(1+\Lambda)^4\log(d/\delta))$, with probability at least $1-\delta$ we have that
    \begin{align}\label{eq:relation_lassos}
        \overline{\mathcal{L}}^{(\overline{\lambda})}_{\mathcal{D}_{p,W}}(\bfbeta) - \overline{\mathcal{L}}^{(\overline{\lambda})}_{\mathbf{X},\mathbf{y}}(\bfbeta) \leq c\epsilon \quad\text{for all}~\bfbeta\in\mathbb{R}^d~\text{such that}~\|\bfbeta\|_1 \leq \Lambda.
    \end{align}
    Assume for now that $\Lambda$ is large enough so that $\max\{\|\widetilde{\bfbeta}\|_1,\|\hat{\bfbeta}\|_1\} \leq \Lambda$. According to \Cref{eq:optimal_solution}, we know that $\|\hat{\bfbeta}\|_1 = \frac{2pw}{1+\overline{\lambda}+4p^2(w-1)} \leq \frac{1}{1+\overline{\lambda}}$. On the other hand,
    \begin{align*}
        \overline{\lambda}\|\widetilde{\bfbeta}\|_1 \leq \overline{\mathcal{L}}_{\mathbf{X},1^n}^{(\overline{\lambda})}(\widetilde{\bfbeta}) \leq \overline{\mathcal{L}}_{\mathbf{X},1^n}^{(\overline{\lambda})}(\hat{\bfbeta}) + c\overline{\lambda}\epsilon\|\widetilde{\bfbeta}\|_1 \leq \overline{\mathcal{L}}^{(\overline{\lambda})}_{\mathcal{D}_{p,W}}(\hat{\bfbeta}) + c\epsilon + c\overline{\lambda}\epsilon\|\widetilde{\bfbeta}\|_1,
    \end{align*}
    since $\widetilde{\bfbeta}$ is a $c\overline{\lambda}\epsilon\|\widetilde{\bfbeta}\|_1$-minimiser and by using \Cref{eq:relation_lassos}. By also using that $\overline{\mathcal{L}}^{(\lambda)}_{\mathcal{D}_{p,W}}(\hat{\bfbeta}) \leq \frac{1}{2} + \frac{2pw(\overline{\lambda} - p)}{1 + \overline{\lambda} + 4p^2(w-1)}$ according to \Cref{lem:properties_expected_lasso}, we obtain that
    \begin{align*}
        \|\widetilde{\bfbeta}\|_1 \leq \frac{1}{c\overline{\lambda}(1-\epsilon)}\left(\frac{1}{2} + \frac{2pw(\overline{\lambda} - p)}{1 + \overline{\lambda} + 4p^2(w-1)} + c\epsilon\right) \leq \frac{1}{c\overline{\lambda}} + \frac{5}{4c(1+\overline{\lambda})}. \tag{since $\epsilon \leq 1/5$}
    \end{align*}
    We can therefore set $\Lambda = 3/(c\overline{\lambda}) = O(1)$ in order to employ \Cref{eq:relation_lassos} with $\widetilde{\bfbeta}$ and $\hat{\bfbeta}$. Note that $n=\Omega(\epsilon^{-2}(1+\Lambda)^2\log(d/\delta))$ is satisfied within our range of parameters since $\frac{n^4}{\epsilon^2\lambda^4}\log{d} = O(n)$. Thus, by \Cref{eq:relation_lassos} and a simple Chernoff bound in order to argue that $\overline{\mathcal{L}}^{(\overline{\lambda})}_{\mathbf{X},\mathbf{y}}(\hat{\bfbeta}) - \overline{\mathcal{L}}^{(\overline{\lambda})}_{\mathcal{D}_{p,W}}(\hat{\bfbeta}) \leq c\epsilon$ (which only applies to $\hat{\bfbeta}$ and not all $\bfbeta$ such that $\|\bfbeta\|_1 \leq \Lambda$ as in \Cref{eq:relation_lassos}),
    \begin{align*}
        \overline{\mathcal{L}}^{(\overline{\lambda})}_{\mathcal{D}_{p,W}}(\widetilde{\bfbeta}) - \overline{\mathcal{L}}^{(\overline{\lambda})}_{\mathcal{D}_{p,W}}(\hat{\bfbeta}) \leq 2c\epsilon + \overline{\mathcal{L}}^{(\overline{\lambda})}_{\mathbf{X},1^n}(\widetilde{\bfbeta}) - \overline{\mathcal{L}}^{(\overline{\lambda})}_{\mathbf{X},1^n}(\hat{\bfbeta}) \leq 2c\epsilon + c\overline{\lambda}\epsilon\|\widetilde{\bfbeta}\|_1,
    \end{align*}
    which means that $\widetilde{\bfbeta}$ is an approximate solution to $\overline{\mathcal{L}}^{(\overline{\lambda})}_{\mathcal{D}_{p,W}}(\bfbeta)$ with error $2c\epsilon + c\overline{\lambda}\epsilon\|\widetilde{\bfbeta}\|_1$. Now note that the condition $\epsilon\in(1/d\overline{\lambda}^2,\min\{1/\overline{\lambda}^2,1/5\})$ is satisfied since $1/d\overline{\lambda}^2 = O(\sqrt{n^3\log{d}}/\lambda^2)$ and $1/\overline{\lambda} = O(1)$ and $1/\overline{\lambda}^2 = n^2/\lambda^2 = \Omega(\sqrt{n^3\log{d}}/\lambda^2)$. Therefore, according to \Cref{thr:lasso_algorithm_set_finding}, we can use $\widetilde{\bfbeta}$ in order to output $\widetilde{W}\subset[d]$ such that $|\widetilde{W}\Delta W| \leq w/5$ with high probability, thus solving the $\mathsf{ASF}^{d,N}_{w,p}$ problem. Since the query complexity of $\mathsf{ASF}^{d,N}_{w,p}$ is $\Omega(\sqrt{wd}/p)$ by \Cref{fact:quantum_asf_lower_bound}, this yields that the query complexity of the Lasso algorithm for $\mathcal{L}^{(\lambda)}_{\mathbf{X},1^n}(\bfbeta)$ is
    \[
        \Omega\left(\frac{\sqrt{wd}}{p}\right) = \Omega\left(\frac{1}{\overline{\lambda}^2\epsilon}\sqrt{\frac{d}{\overline{\lambda}^2\epsilon}}\right) = \Omega\left(\frac{n^3\sqrt{d}}{\lambda^3\epsilon^{3/2}} \right). \qedhere
    \]
\end{proof}

\subsection{Classical lower bound}

We now turn our attention to proving a classical query lower bound by reducing the $\mathsf{ESF}^{d,N}_{w,p}$ problem to an approximate Lasso one. The proof is similar to the quantum lower bound, the main difference being that the reduction from \Cref{fact:est_to_asf_reduction} is used instead of \Cref{fact:reduction}.
\begin{theorem}\label{thr:classical_lower_bound}
    Let $n < d$ such that $n = \Omega(\log{d})$, $\lambda = \Omega(n)$, and $\epsilon \in (0,1/5)$ such that $\epsilon = \Omega(\sqrt{n^3\log{d}}/\lambda^2)$. There is $\mathbf{X}\in\{-1,1\}^{n\times d}$ such that every bounded-error classical algorithm that computes an $\epsilon\lambda\|\widetilde{\bfbeta}\|_1$-minimiser for the Lasso $\mathcal{L}^{(\lambda)}_{\mathbf{X},1^n}(\bfbeta)$ uses 
    \begin{align*}
        \Omega\left(\frac{n^4 d}{\lambda^4\epsilon^{2}\log{d}} \right) \quad\text{classical queries to}~\mathbf{X}.
    \end{align*}
\end{theorem}
\begin{proof}
    Let $w=\lfloor1/(\overline{\lambda}^2\epsilon)\rfloor$ and $p=1/(2\lfloor 1/(\overline{\lambda}^2\epsilon)\rfloor)$. Consider a valid input $\mathbf{X}'\in\{-1,1\}^{N\times d}$ for $\mathsf{ESF}^{d,N}_{w,p}$. For a small enough constant $c\in(0,1)$, suppose we have a classical Lasso algorithm which returns a $c\lambda\epsilon\|\widetilde{\bfbeta}\|_1$-minimiser $\widetilde{\bfbeta}\in\mathbb{R}^d$ with high probability for the Lasso $\mathcal{L}^{(\lambda)}_{\mathbf{X},1^n}(\bfbeta)$ with inputs  $(\mathbf{X},\mathbf{y} = 1^n)$ depending on $\mathbf{X}'$. The vector $\widetilde{\bfbeta}$ is also a $c\overline{\lambda}\epsilon\|\widetilde{\bfbeta}\|_1$-minimiser to the normalised Lasso $\overline{\mathcal{L}}^{(\overline{\lambda})}_{\mathbf{X},1^n}(\bfbeta)$ where $\overline{\lambda} = \lambda/n$. Similarly to the proof of \Cref{thr:lower_bound}, one can show that $\widetilde{\bfbeta}$ is, with high probability, an approximate solution to $\overline{\mathcal{L}}^{(\overline{\lambda})}_{\mathcal{D}_{p,W}}(\bfbeta)$ with error $2c\epsilon + c\overline{\lambda}\epsilon\|\widetilde{\bfbeta}\|_1$ if $n = \Omega(\epsilon^{-2}(1+\Lambda)^4\log(d/\delta))$ where $\Lambda = 3/(c\overline{\lambda}) = O(1)$. This condition is satisfied within our range of parameters since $\frac{n^4}{\epsilon^2\lambda^4}\log{d} = O(n)$. On the other hand, again we note that the condition $\epsilon\in(1/d\overline{\lambda}^2,\min\{1/\overline{\lambda}^2,1/5\})$ is satisfied since $1/d\overline{\lambda}^2 = O(\sqrt{n^3\log{d}}/\lambda^2)$ and $1/\overline{\lambda} = O(1)$ and $1/\overline{\lambda}^2 = n^2/\lambda^2 = \Omega(\sqrt{n^3\log{d}}/\lambda^2)$. Therefore, by \Cref{thr:lasso_algorithm_set_finding}, we can use the classical Lasso algorithm to solve $\mathsf{DSF}^{d,n}_{w,p}$. Finally, this means that, by \Cref{fact:est_to_asf_reduction}, $O(\log{d})$ uses of the classical Lasso algorithm can solve $\mathsf{ESF}^{d,N}_{w,p}$. Since the query complexity of $\mathsf{ESF}^{d,N}_{w,p}$ is $\Omega((d-w)/p^2)$ by \Cref{fact:classical_esf_lower_bound}, this yields that the query complexity of the classical Lasso algorithm for $\mathcal{L}^{(\lambda)}_{\mathbf{X},1^n}(\bfbeta)$ is
    \[
        \Omega\left(\frac{d-w}{p^2\log{d}}\right) = \Omega\left(\frac{d - 1/(\overline{\lambda}^2\epsilon)}{\overline{\lambda}^4\epsilon^2\log{d}}\right) = \Omega\left(\frac{n^4d}{\epsilon^2\lambda^4\log{d}} \right). \qedhere
    \]
\end{proof}

\section{Bounds in noisy regime}\label{sect:bounds_noisy}

In this section, we assume that the observation vector $\mathbf{y}\in\mathbb{R}^n$ and the design matrix $\mathbf{X}\in\mathbb{R}^{n\times d}$ are connected via the standard linear model 
\[
\mathbf{y} = \mathbf{X}\bfbeta^* + \mathbf{w},
\]
where the true solution $\bfbeta^*\in\mathbb R^d$ (also referred to as regression vector) is sparse and $\mathbf{w}\in\mathbb R^n$ is a noise vector. Our aim is to determine how close the Lasso solution produced by the LARS algorithm, especially its approximate version from \cref{alg:classical_quantum}, is to the true solution $\bfbeta^*$. More specifically, we shall focus on bounding the \emph{mean-squared prediction error} $\|\mathbf{X}(\bfbeta^\ast - \widetilde{\bfbeta})\|_2^2/n$. Throughout this section we shall assume the following.
\begin{assumption}\label{ass: maxnormbdd}
    Given the design matrix $\mathbf{X}\in\mathbb{R}^{n\times d}$, then $\max_{i\in[d]}\norm{\mathbf{X}_i}_2\leq \sqrt{C n}$ for some $C>0$.  
\end{assumption}
The results from this section are standard in the Lasso literature~\cite{wainwright2019high, buhlmann2011statistics} for the case when the regularisation path from LARS algorithm exactly minimises the Lasso cost function. We generalise them for case when $\widetilde{\mathcal{P}}$ is an approximate regularisation path with error $\lambda\epsilon\|\widetilde{\bfbeta}\|_1$. 

\subsection{Slow rates}

It is possible to obtain bounds on the mean-squared prediction error without barely any assumptions on the design matrix $\mathbf{X}\in\mathbb{R}^{n\times d}$ and the noise vector $\mathbf{w}\in\mathbb{R}^n$, as the next result shows.
\begin{theorem}\label{thm:Lasso_error}
    Let $\lambda > 0$ and $\epsilon \in [0,1)$. Any approximate solution $\widetilde{\bfbeta}$ with error $\lambda\epsilon\|\widetilde{\bfbeta}\|_1$ of the Lasso with regularisation parameter $\lambda \geq \| \mathbf{X}^\top \mathbf{w} \|_\infty/(1-\epsilon)$ satisfies
    %
    %
    \begin{align*}
        \norm{\mathbf{X}(\bfbeta^* - \widetilde{\bfbeta})}_2^2 \leq 2(2-\epsilon)\lambda \norm{\bfbeta^*}_1.
    \end{align*}
\end{theorem}
\begin{proof}
Notice first that, since $\widetilde{\bfbeta}$ is a minimiser up to additive error $\lambda\epsilon\|\widetilde{\bfbeta}\|_1$, then
\[
\mathcal{L}_\lambda(\widetilde{\bfbeta}) - \mathcal{L}_\lambda(\bfbeta^\ast) \leq \lambda\epsilon\|\widetilde{\bfbeta}\|_1 \implies \frac{1}{2}\norm{\mathbf{y}-\mathbf{X}\widetilde{\bfbeta}}_2^2 + \lambda(1 -\epsilon) \norm{\widetilde{\bfbeta}}_1  \leq \frac{1}{2}\norm{\mathbf{y} - \mathbf{X}\bfbeta^*}_2^2 + \lambda \norm{\bfbeta^*}_1.
\]
Use the equality $\mathbf{y}=\mathbf{X}\bfbeta^* + \mathbf{w}$ in the above inequality to obtain
\begin{align*}
    \frac{1}{2}\norm{\mathbf{X}(\bfbeta^*-\widetilde{\bfbeta}) + \mathbf{w}}_2^2 + \lambda(1-\epsilon) \norm{\widetilde{\bfbeta}}_1  &\leq \frac{1}{2}\norm{\mathbf{w}}_2^2 + \lambda \norm{\bfbeta^*}_1,
\end{align*}
from which we get the following inequality
\begin{align}
    \frac{1}{2}\norm{\mathbf{X}(\bfbeta^*-\widetilde{\bfbeta})}_2^2 + \lambda(1-\epsilon)\norm{\widetilde{\bfbeta}}_1 + \frac{1}{2}\norm{\mathbf{w}}_2^2 + \mathbf{w}^\top \mathbf{X}(\bfbeta^*-\widetilde{\bfbeta}) &\leq \frac{1}{2}\norm{\mathbf{w}}_2^2 + \lambda \norm{\bfbeta^*}_1  \iff \nonumber\\
    \frac{1}{2}\norm{\mathbf{X}(\bfbeta^*-\widetilde{\bfbeta})}_2^2 + \lambda(1-\epsilon)\norm{\widetilde{\bfbeta}}_1 + \mathbf{w}^\top \mathbf{X}(\bfbeta^*-\widetilde{\bfbeta}) &\leq \lambda \norm{\bfbeta^*}_1 \label{eq:inequality1}
\end{align}
and finally
\begin{align*}
    \frac{1}{2}\norm{\mathbf{X}(\bfbeta^*-\widetilde{\bfbeta})}_2^2   &\leq \mathbf{w}^\top \mathbf{X}(\widetilde{\bfbeta} - \bfbeta^*)  + \lambda \norm{\bfbeta^*}_1  - \lambda(1-\epsilon)\| \widetilde{\bfbeta}\|_1 \tag{\Cref{eq:inequality1}}\\
    &\leq \| \mathbf{X}^\top \mathbf{w}\|_\infty \|\widetilde{\bfbeta} - \bfbeta^*\|_1   + \lambda \norm{\bfbeta^*}_1  - \lambda(1-\epsilon)\| \widetilde{\bfbeta}\|_1 \tag{H\"older's inequality}\\
    &\leq  \| \mathbf{X}^\top \mathbf{w}\|_\infty(\| \widetilde{\bfbeta}\| + \|\bfbeta^*\|_1)   + \lambda \norm{\bfbeta^*}_1  - \lambda(1-\epsilon)\| \widetilde{\bfbeta}\|_1 \tag{Triangle inequality}\\
    &= \norm{\widetilde{\bfbeta}} \left(\|\mathbf{X}^\top \mathbf{w}\|_\infty - \lambda(1-\epsilon) \right) + \norm{\bfbeta^*} (\| \mathbf{X}^\top \mathbf{w}\|_\infty+\lambda ) \\
    &\leq \norm{{\bfbeta}^*} (\| \mathbf{X}^\top \mathbf{w}\|_\infty+\lambda ) \tag{$\|\mathbf{X}^\top \mathbf{w}\|_\infty \leq (1-\epsilon)\lambda$}\\
    &\leq (2-\epsilon)\lambda \norm{\bfbeta^*}. \tag*{($\|\mathbf{X}^\top \mathbf{w}\|_\infty \leq (1-\epsilon)\lambda$) \quad\qedhere}
\end{align*}
\end{proof}

We give some examples of the above result using common linear regression models.

\paragraph*{Linear Gaussian model.} In the classical linear Gaussian model, the design matrix is deterministic, meaning that $\mathbf{X}\in\mathbb{R}^{n\times d}$ is fixed, while the noise vector $\mathbf{w}\in\mathbb{R}^n$ has i.i.d.\ $\mathcal{N}(0,\sigma^2)$ entries. We then bound the probability that $\lambda \geq \| \mathbf{X}^\top \mathbf{w} \|_\infty/(1-\epsilon)$. Using that $\mathbf{X}_i^\top \mathbf{w}$ is a sub-Gaussian random variable with parameter $\sigma\|\mathbf{X}_i\|_2 \leq \sigma\sqrt{Cn}$ (\cref{fact:sub_gaussian}), we bound the complement of this probability:
\begin{align*}
    \mathbb{P}\left[\norm{\mathbf{X}^\top \mathbf{w}}_\infty> t \right]
    = \mathbb{P}\left[\max_{i\in[d]} |\mathbf{X}_i^\top \mathbf{w}|>t\right]
    \leq \sum_{i=1}^d\mathbb{P}[|\mathbf{X}_i^\top \mathbf{w}|>t]
     \leq 2d\exp \left( -\frac{t^2}{2C\sigma^2 n} \right),
\end{align*}
where the first inequality follows from a union bound.
Setting $t=\lambda(1-\epsilon)=\sqrt{2C\sigma^2n\log(2d/\delta)}$, the above probability is at most $\delta$. Consequently,
\begin{align*}
    \mathbb{P}\left[\norm{\mathbf{X}^\top \mathbf{w}}_\infty \leq \lambda(1-\epsilon)\right] \geq 1-\delta. 
\end{align*}
Substituting this value of $\lambda$ in \cref{thm:Lasso_error}, with probability at least $1-\delta$,
\begin{align*}
    \frac{\norm{\mathbf{X}(\bfbeta^*-\widetilde{\bfbeta})}_2^2}{n}  \leq \frac{2(2-\epsilon)\lambda\norm{\bfbeta^*}}{n} = \frac{2-\epsilon}{1-\epsilon}\sqrt{\frac{8C\sigma^2 \log(2d/\delta)}{n}}\norm{\bfbeta^*}.
\end{align*}
We see that the mean square error vanishes by a square root of $n$ factor as $n$ becomes large. Without additional assumptions, this is known as a slow rate.

\paragraph*{Compressed sensing.} In compressed sensing, it is common for the design matrix $\mathbf{X}\in\mathbb{R}^{n\times d}$ to be chosen by the user, and a standard choice is the standard Gaussian matrix with i.i.d.\ $\mathcal{N}(0,1)$ entries. Assume further that the noise vector $\mathbf{w}\in\mathbb{R}^n$ is deterministic with $\|\mathbf{w}\|_\infty \leq \sigma$. Thus $\mathbf{X}_i^\top \mathbf{w}$ is a sub-Gaussian random variable with parameter $\|\mathbf{w}\|_2 \leq \sqrt{n}\|\mathbf{w}\|_\infty \leq \sqrt{n}\sigma$ (\cref{fact:sub_gaussian}). By following the exact same steps as the previous example (this time with $C=1$), we conclude that, with probability at least $1-\delta$,
\begin{align*}
    \frac{\norm{\mathbf{X}(\bfbeta^*-\widetilde{\bfbeta})}_2^2}{n}  \leq \frac{2-\epsilon}{1-\epsilon}\sqrt{\frac{8\sigma^2 \log(2d/\delta)}{n}}\norm{\bfbeta^*}.
\end{align*}
So the mean square error also vanishes by a square root of $n$ factor as $n$ becomes large. We shall see next that, by imposing additional restrictions on the design matrix, faster rates can be obtained.

\subsection{Fast Rates}

To obtain faster rates, we need additional assumptions on the design matrix $\mathbf{X}$. The condition below originates actually in the compressed sensing literature. We assume the true solution $\bfbeta^*$ has a hard finite support $\operatorname{supp}(\bfbeta^*) $.

\begin{definition}[Restricted Eigenvalue condition] 
    For $S \subseteq [d]$ and $\zeta\in\mathbb{R}$, let
    \[
    \mathbb{C}_\zeta (S) \triangleq \{\bfDelta\in \mathbb{R}^d : \norm{\bfDelta_{S^c}}_1\leq \zeta \norm{\bfDelta_{S}}_1 \}.
    \]
    A matrix $\mathbf{X}\in\mathbb{R}^{n\times d}$ satisfies the Restricted Eigenvalue {\rm (RE)} condition over $S$ with parameters $(\kappa, \zeta)$ if 
    \[
    \frac{1}{n}\norm{\mathbf{X}\bfDelta}_2^2 \geq \kappa \norm{\bfDelta}_2^2 \quad \forall \bfDelta \in \mathbb{C}_\zeta (S). 
    \]
\end{definition}
%
An interpretation of the {\rm RE} condition is that it bounds away the minimum eigenvalues of the matrix $\mathbf{X}^\top \mathbf{X}$ in specific directions. Ideally, we would like to bound the curvature of the cost function $\|\mathbf{y} - \mathbf{X}\bfbeta\|_2^2$ in all directions, which would guarantee strong bounds on the mean square error. However, in the high-dimensional setting $d\gg n$, $\mathbf{X}^\top \mathbf{X}$ is a $d\times d$ matrix with rank at most $n$, so it is impossible to guarantee a positive curvature in all directions. Therefore, we must relax such stringent condition on the curvature and require that it holds only on a subset $\mathbb{C}_\zeta (S)$ of vectors. The RE condition yields the following stronger result.

\begin{theorem}
    Let $\lambda,\kappa > 0$ and $\epsilon\in[0,1/4]$. Let $S \triangleq \operatorname{supp}(\bfbeta^\ast)$. Suppose $\mathbf{X}\in\mathbb{R}^{n\times d}$ satisfies the {\rm RE} condition over $S$ with parameters $(\kappa, 5)$. Any approximate solution $\widetilde{\bfbeta}\in\mathbb{R}^d$ with error $\lambda\epsilon\|\widetilde{\bfbeta}\|_1$ of the Lasso with regularisation parameter $\lambda \geq 4\norm{\mathbf{X}^\top \mathbf{w}}_\infty$ satisfies
    \[
        \|\mathbf{X}(\widetilde{\bfbeta} - \bfbeta^\ast)\|_2^2  + 3\lambda\|\widetilde{\bfbeta} - \bfbeta^\ast\|_1 \leq 81\lambda^2\frac{|S|}{n\kappa} + 18\lambda\epsilon\|\bfbeta_S^\ast\|_1.
    \]
\end{theorem}
The proof has many variations, and we give the most direct one below. Note the above implies
\begin{align*}
    \|\mathbf{X}(\widetilde{\bfbeta} - \bfbeta^\ast)\|_2^2 \leq 81\lambda^2\frac{|S|}{n\kappa} + 18\lambda\epsilon\|\bfbeta_S^\ast\|_1 \quad\text{and}\quad
    \|\widetilde{\bfbeta} - \bfbeta^\ast\|_1 \leq 27\lambda\frac{|S|}{n\kappa} + 6\epsilon\|\bfbeta_S^\ast\|_1.
\end{align*}

\begin{proof}
We use the basic optimality of $\widetilde{\bfbeta}$,  
\[
\frac{1}{2}\norm{\mathbf{y}-\mathbf{X}\widetilde{\bfbeta}}_2^2 + \lambda(1-\epsilon) \norm{\widetilde{\bfbeta}}_1  \leq \frac{1}{2}\norm{\mathbf{y}-\mathbf{X}\bfbeta^*}_2^2 + \lambda \norm{\bfbeta^*}_1.
\]
Define $\bfDelta \triangleq \widetilde{\bfbeta} -\bfbeta^*$. Expanding out the terms like in \cref{thm:Lasso_error},
\begin{align*}
    0 \leq 2\norm{\mathbf{X}\bfDelta}_2^2 &\leq 4\mathbf{w}^\top \mathbf{X}(\widetilde{\bfbeta} - \bfbeta^*)  + 4\lambda \norm{\bfbeta^*}_1  - 4\lambda(1 -\epsilon)\| \widetilde{\bfbeta}\|_1 \tag{\Cref{eq:inequality1}}\\
    &\leq 4\|\mathbf{X}^\top \mathbf{w}\|_\infty \|\bfDelta\|_1 + 4\lambda \norm{\bfbeta^*}_1  - 4\lambda(1-\epsilon)\| \widetilde{\bfbeta}\|_1 \tag{H\"older's inequality}\\
    &\leq \lambda \|\bfDelta\|_1 + 4\lambda \norm{\bfbeta^*}_1  - 4\lambda(1-\epsilon)\| \widetilde{\bfbeta}\|_1 \tag{$4\|\mathbf{X}^\top \bfDelta\|_\infty \leq \lambda$}\\
    &\leq \lambda (\|\widetilde{\bfbeta}\|_1 + \|\bfbeta^\ast\|_1) + 4\lambda \norm{\bfbeta^*}_1  - 4\lambda(1-\epsilon)\| \widetilde{\bfbeta}\|_1, \tag{Triangle inequality}
\end{align*}
from which we infer that $\|\widetilde{\bfbeta}\|_1 \leq \frac{5}{3 - 4\epsilon}\|\bfbeta^\ast\|_1 \leq \frac{5}{2}\|\bfbeta^\ast\|_1$ since $\epsilon \leq \frac{1}{4}$. Let us now return to the inequality $2\norm{\mathbf{X}\bfDelta}_2^2 \leq \lambda \|\bfDelta\|_1 + 4\lambda( \norm{\bfbeta^*}_1  - \| \widetilde{\bfbeta}\|_1) + 4\lambda\epsilon\| \widetilde{\bfbeta}\|_1$ and further decompose the right-hand side into the respective partitioning sets $S$ and $S^c$ by using that $\bfbeta_S^* + \bfDelta_S = \widetilde{\bfbeta}_S$ and $\bfDelta_{S^c}= \widetilde{\bfbeta}_{S^c} $, since $\bfbeta_{S^c}^* = \mathbf{0}$. We also use the inequality $\|\widetilde{\bfbeta}\|_1 \leq \frac{5}{2}\|\bfbeta^\ast\|_1 = \frac{5}{2}\|\bfbeta^\ast_S\|_1$. 
Therefore
\begin{align*}
    2\norm{\mathbf{X}\bfDelta}_2^2 &\leq \lambda \|\bfDelta\|_1 + 4\lambda(\norm{\bfbeta^*}_1  - \| \widetilde{\bfbeta}\|_1) + 10\lambda\epsilon\|\bfbeta^\ast_S\|_1 \\
    &= \lambda (\| \bfDelta_S\|_1 + \| \bfDelta_{S^c}\|_1)  + 4\lambda (\norm{\bfbeta^*}_1  - \| \widetilde{\bfbeta}\|_1) + 10\lambda\epsilon\|\bfbeta^\ast_S\|_1\\
    &=  \lambda (\| \bfDelta_S\|_1 + \| \bfDelta_{S^c}\|_1)  + 4\lambda(\norm{\bfbeta_S^*}_1 - \|\bfbeta_S^* + \bfDelta_S\|_1 - \|\bfDelta_{S^c}\|_1) + 10\lambda\epsilon\|\bfbeta^\ast_S\|_1 \\ 
    &\leq \lambda (\| \bfDelta_S\|_1 + \| \bfDelta_{S^c}\|_1) + 4\lambda (\|\bfDelta_S\|_1 - \|\bfDelta_{S^c}\|_1) + 10\lambda\epsilon\|\bfbeta^\ast_S\|_1 \tag{\Cref{eq:auxiliary_equation}} \\
    &= \lambda(5\|\bfDelta_S\|_1 - 3\|\bfDelta_{S^c}\|_1) + 10\lambda\epsilon\|\bfbeta^\ast_S\|_1,
\end{align*}
where we used that
\begin{align}\label{eq:auxiliary_equation}
 \| \bfbeta_S^*\|_1 - \| \bfbeta_S^* +\bfDelta_S\|_1 \leq \| \bfbeta_S^*\|_1 - (\| \bfbeta_S^*\|_1  - \|\bfDelta_S\|_1) = \| 
     \bfDelta_S\|_1.   
\end{align}
We now proceed by considering two cases: (Case 1) $\|\bfDelta_S\|_1 \geq \epsilon \|\bfbeta^\ast_S\|_1$ or (Case 2) $\|\bfDelta_S\|_1 < \epsilon \|\bfbeta^\ast_S\|_1$. Therefore it must hold that either (Case 1)
\begin{align*}
    2\|\mathbf{X}\bfDelta\|_2^2 + 3\lambda\|\bfDelta_{S^c}\|_1 \leq 15\lambda\|\bfDelta_S\|_1
\end{align*}
or (Case 2)
\begin{align*}
    2\|\mathbf{X}\bfDelta\|_2^2  + 3\lambda\|\bfDelta_{S^c}\|_1 \leq 15\lambda\epsilon\|\bfbeta_S^\ast\|_1.
\end{align*}
In the first case, we find that $\|\bfDelta_{S^c}\|_1 \leq 5\|\bfDelta_S\|_1 \implies \bfDelta \in\mathbb{C}_{5}(S)$, which means that we can now apply the {\rm RE} condition. More specifically, from $2\norm{\mathbf{X}\bfDelta}_2^2 + 3\lambda\|\bfDelta_{S^c}\|_1 \leq 15 \lambda \norm{\bfDelta_S}_1$ we get
\begin{align*} 
    2\norm{\mathbf{X}\bfDelta}_2^2 + 3\lambda\|\bfDelta\|_1 &\leq 18 \lambda \norm{\bfDelta_S}_1 
    \leq 18\lambda \sqrt{|S|}\|\bfDelta\|_2
    \overset{\rm RE}{\leq} 18\lambda\sqrt{\frac{|S|}{n\kappa}}\|\mathbf{X}\bfDelta\|_2
    \leq 81\lambda^2\frac{|S|}{n\kappa} + \|\mathbf{X}\bfDelta\|_2^2,
\end{align*}
where the last inequality follows from $2ab \leq a^2 + b^2$ for $a,b\in\mathbb{R}$. Thus $\norm{\mathbf{X}\bfDelta}_2^2 + 3\lambda\|\bfDelta\|_1 \leq  81\lambda^2\frac{|S|}{n\kappa}$.

In the second case, from $2\|\mathbf{X}\bfDelta\|_2^2  + 3\lambda\|\bfDelta_{S^c}\|_1 \leq 15\lambda\epsilon\|\bfbeta_S^\ast\|_1$ we get
\[
    2\|\mathbf{X}\bfDelta\|_2^2  + 3\lambda\|\bfDelta\|_1 \leq 18\lambda\epsilon\|\bfbeta_S^\ast\|_1 \implies \|\mathbf{X}\bfDelta\|_2^2  + 3\lambda\|\bfDelta\|_1 \leq 18\lambda\epsilon\|\bfbeta_S^\ast\|_1.
\]
Putting both Cases 1 and 2 together leads to the desired result.
\end{proof}

\paragraph*{Linear Gaussian model and compressed sensing.} Once again, we can apply the above result to common linear regression models. In the case when $\mathbf{X}\in\mathbb{R}^{n\times d}$ is fixed and $\mathbf{w}\in\mathbb{R}^n$ has i.i.d.\ $\mathcal{N}(0,\sigma^2)$ entries, the same reasoning from the previous section yields
\begin{align*}
    \mathbb{P}\left[\lambda \geq 4\norm{\mathbf{X}^\top \mathbf{w}}_\infty \right] \geq 1-\delta
\end{align*}
by setting $\lambda = \sqrt{32C\sigma^2n\log(2d/\delta)}$. This value of $\lambda$ leads to, with probability at least $1-\delta$,
\begin{align*}
    \frac{\|\mathbf{X}(\widetilde{\bfbeta} - \bfbeta^\ast)\|_2^2}{n} &\leq \frac{2592C\sigma^2|S|\log(2d/\delta)}{\kappa n} + 18\epsilon\sqrt{\frac{32C\sigma^2\log(2d/\delta)}{\kappa n}}\|\bfbeta_S^\ast\|_1.
\end{align*}
The same bound (with $C=1$) holds for the model where $\mathbf{X}\in\mathbb{R}^{n\times d}$ is the standard Gaussian matrix with i.i.d.\ $\mathcal{N}(0,1)$ entries and $\mathbf{w}\in\mathbb{R}^n$ is fixed with $\|\mathbf{w}\|_\infty \leq \sigma$. Thus, under the RE condition, if $\epsilon = O(\sqrt{\log(d)/n})$, then the mean square error vanishes by an $n$ factor as $n$ becomes large, which can be substantially smaller than the $\sqrt{\log(d)/n}$ bound from the previous section. This is known as a fast rate.

\section{Discussion and future work}\label{sect:discussion}

We studied and quantised the LARS pathwise algorithm (or homotopy method) proposed by Efron \textit{et al.}~\cite{efron2004least} and Osborne \textit{et al.}~\cite{osborne2000Lasso,osborne2000new} which produces the set of Lasso solutions for varying the penalty term $\lambda$. By assuming quantum access to the Lasso input, we proposed two quantum algorithms. The first one (\cref{alg:warmup_classical_quantum}) simply replaces the classical search within the joining time calculation with the quantum minimum-finding subroutine from D\"{u}rr and H\o{}yer~\cite{durr1996quantum}. Similar to the classical LARS algorithm, it outputs the exact Lasso path, but has an improved runtime by a quadratic factor in the number of features $d$. Our second quantum algorithm (\cref{alg:classical_quantum}), on the other hand, outputs an approximate Lasso path $\widetilde{\mathcal{P}}$ by computing the joining times up to some error. This is done by approximating the joining times  using quantum amplitude estimation~\cite{brassard2002quantum} and finding their maximum via the approximate quantum minimum-finding subroutine from Chen and de Wolf~\cite{chen2021quantum}. Consequently, the runtime of our approximate quantum LARS algorithm is quadratically improved in both the number of features $d$ and observations $n$. We established the correctness of \cref{alg:classical_quantum} by employing an approximate version of the KKT conditions and a duality gap, and showed that $\widetilde{\mathcal{P}}$ is a minimiser of the Lasso cost function up to additive error $\lambda\epsilon\Vert\widetilde \bfbeta(\lambda)\Vert_1$. Finally, we dequantised \cref{alg:classical_quantum} and proposed an approximate classical LARS algorithm (\cref{alg:approximate_classical}) based on sampling.

The time complexity of our approximate LARS algorithms, both classical and quantum, directly depends on the design matrix $\mathbf{X}$ via the quantity $\|\mathbf{X}\|_{\max}\|\mathbf{X}^+_{\mathcal{A}}\|_2$, the mutual incoherence of $\mathbf{X}$, and the mutual overlap between $\mathbf{X}$ and the vector of observations $\mathbf{y}$. When $\mathbf{X}$ is a random matrix from the $\mathbf{\Sigma}$-Gaussian ensemble, where $\mathbf{\Sigma}$ is a covariance matrix with ``good'' properties, and $\mathbf{y}$ is sparse, we showed that these three quantities are well bounded and that \Cref{alg:classical_quantum} and \Cref{alg:approximate_classical} have complexities $\widetilde{O}(\sqrt{d}/\epsilon)$ and $\widetilde{O}(d/\epsilon^2)$ per iteration, respectively, which exponentially improves the dependence on $n$ compared to the $O(nd)$ complexity of the standard LARS algorithm. The condition that $\mathbf{y}$ is sparse follows naturally from the fact that there could be data situations in which $\mathbf{X}$ being Gaussian leads to very small, in absolute magnitude, responses in the dependent variable $\mathbf{y}$. Another reason for $\mathbf{y}$ being sparse is inaccurate or incomplete measurements~\cite{candes2006stable}. Usually, sparsity can be artificially generated by the corresponding  matrix $\mathbf{X}$ through setting some of the regression coefficients in $\bfbeta^\ast$  to zero or setting some of the components in $\mathbf{y}$ to zero. In both cases, our algorithms would be robust enough and  exhibit speed up to recover the approximate original coefficient vector. Even though we only proved such result for random Gaussian matrices, we conjecture that it should also apply to other matrices.

We point out some future directions of research following our work. In \cref{alg:classical_quantum,alg:approximate_classical}, the crossing times are computed exactly and the joining times are estimated up to some error. A natural extension would be to allow errors also in computing the crossing times. This would require new techniques in order to retain the path solution continuity. Another direction is to reduce the number of iterations of the LARS algorithm. In the classical setting, Mairal and Yu~\cite{mairal2012complexity} proposed an approximate LARS algorithm with a maximum number of iterations by employing first-order optimisation methods when two kinks are too close to each other. Moreover, it would be interesting to design a fully quantum LARS algorithm by using efficient quantum subroutines for matrix multiplication and matrix inversion, e.g., based on block-encoding techniques~\cite{low2019hamiltonian,gilyen2019quantum,chakraborty2019power}. Also, we believe that our quantum query lower bound is not tight and can be improved, specially the dependence on $\epsilon$ (also not tight in the work of Chen and de Wolf~\cite{chen2021quantum}).

Our algorithms are designed to work only in the fault tolerant regime. Another future direction is to study the LARS algorithm in the NISQ (Noisy Intermediate-Scale Quantum) setting~\cite{preskill2018quantum}, where algorithms operate on a relatively small number of qubits and have shallow circuit depths. 
In addition, while we have studied the plain-vanilla LARS algorithm, other Lasso related algorithms still remain unexplored in the quantum setting. For example, fused Lasso~\cite{tibshirani2005sparsity, xin2016efficient, hoefling2010path}, which penalises the $\ell_1$-norm of both the coefficients and their successive differences, is a generalisation of Lasso with applications to support vector classifier~\cite{gunn1998support, rebentrost2014quantum, suykens1999least}. On the other hand, grouped Lasso is a generalised model for linear regression with $\ell_1$ and $\ell_2$-penalties~\cite{simon2013sparse, friedman2010note}. This model is used in settings where the design matrix can be 
decomposed into groups of submatrices. Simon \textit{et al.}~\cite{simon2013sparse} proposed a sparse grouped Lasso algorithm which finds applications in sparse graphical modelling. We believe that similar techniques employed in this work could be applied to these alternative Lasso settings. 

\section{Acknowledgements}
JFD thanks Hanzhong Liu and Bin Yu for very useful clarifications regarding Ref.~\cite{mairal2012complexity} and Lasso in general, Yanlin Chen and Ronald de Wolf for helpful clarifications regarding Ref.~\cite{chen2021quantum}, András Gilyén for Ref.~\cite{van2017quantum}, Yihui Quek for Ref.~\cite{quek2020robust}, and Rahul Jain, Josep Lumbreras, and Marco Tomamichel for interesting discussions. JFD also thanks Iosif Pinelis for answering a question on MathOverflow~\cite{mathoverflow}.
JFD acknowledges funding from ERC grant No.\ 810115-DYNASNET. DL acknowledges funding from the Latvian Quantum Initiative under EU Recovery and Resilience Facility under project No.\ 2.3.1.1.i.0/1/22/I/CFLA/001 in the final part of the project. CSP gratefully acknowledges Ministry of Education (MOE), AcRF Tier 2 grant (Reference No: MOE-T2EP20220-0013) for the funding of this research. This research is supported by the National Research Foundation, Singapore and A*STAR under its CQT Bridging Grant and its Quantum Engineering Programme under grant NRF2021-QEP2-02-P05.

\DeclareRobustCommand{\DE}[2]{#2}
\DeclareRobustCommand{\VANDE}[3]{#3}
\bibliographystyle{alphaurl}
\bibliography{QuantumLassobib}
\appendix

\section{Summary of symbols}
\label{app:symbols}

The symbols and their corresponding concept are summarised in the table below.

{\renewcommand{\arraystretch}{1.01}
\begin{table}[ht]
\centering
\begin{tabular}{|c|c|c|}
\hline
Symbol & Explanation  \\
\hline
$n$ & Number of observations or sample points  \\
$d$ & Number of features \\
$\mathbf{y}$ & Vector of observations \\
$\mathbf{X}$ & Design matrix \\
$\bfbeta$ & Optimisation variables  \\
$\hat \bfbeta$ & Optimal Lasso solution \\
$\widetilde \bfbeta$ & Approximate Lasso solution \\
$\lambda$ & Regularisation or penalty parameter \\
$\mathcal{L}^{(\lambda)}_{\mathbf{X},\mathbf{y}}(\bfbeta)$ & Lasso function with input $(\mathbf{X},\mathbf{y})$\\
$\overline{\mathcal{L}}^{(\lambda)}_{\mathbf{X},\mathbf{y}}(\bfbeta)$ & Normalised lasso function with input $(\mathbf{X},\mathbf{y})$\\
$\overline{\mathcal{L}}^{(\lambda)}_{\mathcal{D}}(\bfbeta)$ & Expected Lasso function relative to distribution $\mathcal{D}$\\
$\mathcal P$ & Optimal regularisation Path \\
$\widetilde{\mathcal P}$ & Approximate regularisation Path \\
$\mathcal A$ & Active set  \\
$\mathcal I$ & Inactive set  \\
$\bfeta$ & Equicorrelation signs  \\
$\alpha_{\mathcal{A}}$ & Mutual incoherence \\
$\gamma_{\mathcal{A}}$ & Mutual overlap  \\
\hline
\end{tabular}\caption{Summary of symbols and their corresponding concept used in the paper.}
\end{table}
}

\end{document}